\documentclass[a4paper,UKenglish,cleveref,autoref,thm-restate,preprint]{lipics-v2021}
\usepackage{hyperref}
\usepackage{colortbl}
\usepackage{xspace}
\usepackage{amsmath}
\usepackage{amssymb}
\usepackage{ifthen}
\usepackage{xspace}
\usepackage[dvipsnames]{xcolor}
\usepackage{subcaption}
\usepackage{dcolumn}
\usepackage{tikz}
\usepackage{booktabs}
\usepackage{graphicx}
\usepackage[linesnumbered,ruled,noend,vlined]{algorithm2e}
\usepackage{paralist}
\usepackage{environ}
\usepackage{fontawesome}
\NewEnviron{killcontents}{}

\usetikzlibrary{automata,positioning,trees,shapes,petri,arrows,snakes,backgrounds,calc,fit,arrows.meta,shadows,math,decorations.markings}
\usetikzlibrary{decorations.pathmorphing}
\usetikzlibrary{shapes.multipart,arrows,automata}
\usetikzlibrary{shapes.geometric}
\usetikzlibrary{positioning}
\def\blindreview{}

\newif\ifTR
\TRtrue
\ifTR
\else
\fi

\newcounter{claimcounter}
\renewenvironment{claim}{\refstepcounter{claimcounter}{\medskip\noindent \underline{Claim \theclaimcounter:}}\itshape}{\smallskip}
\crefname{claimcounter}{Claim}{Claims}

\Crefname{algocf}{Algorithm}{Algorithms}
\crefname{figure}{Figure}{Figures}

\InputIfFileExists{cutmargins.tex}{%
	\pdfpagesattr{/CropBox [92 112 523 758]} % LNCS page made big
}{%
}

\newcommand{\vh}[1]{\textcolor{orange}{\ifmmode \text{[#1]}\else [VH: #1] \fi}}
\newcommand{\ol}[1]{\textcolor{blue}{\ifmmode \text{[OL: #1]}\else [OL: #1] \fi}}

\renewcommand{\ol}[1]{}

\newcommand{\blinded}[1]{\ifx\blindreview\undefined #1 \else \textcolor{black!65}{[blinded for review]}\fi}

\newcommand{\sofi}[0]{(S,O,f,i)}
\newcommand{\sofiprime}[0]{(S',O',f', i')}
\newcommand{\sofiof}[1]{(S_{#1},O_{#1},f_{#1},i_{#1})}
\newcommand{\sofiprimeof}[1]{(S_{#1},O'_{#1},f'_{#1},i'_{#1})}

\newcommand{\rank}[0]{\mathit{rank}}
\newcommand{\reach}[0]{\mathit{reach}}
\newcommand{\rankof}[1]{\mathit{rank}(#1)}
\newcommand{\rankofin}[2]{\mathit{rank}_{#2}(#1)}

\newcommand{\rankwof}[1]{\rank_{\word}(#1)}

\newcommand{\partto}[0]{\mathrel{\rightharpoonup}}

\newcommand{\infof}[1]{\inf(#1)}

\newcommand{\variant}[2]{#1 \lhd #2}
\newcommand{\img}[0]{\mathrm{img}}
\newcommand{\imgof}[1]{\img(#1)}
\newcommand{\restrof}[2]{#1 \raisebox{-.5ex}{$|$}_{#2}}

\newcommand{\aut}[0]{\mathcal{A}}
\newcommand{\but}[0]{\mathcal{B}}
\newcommand{\cut}[0]{\mathcal{C}}
\newcommand{\dut}[0]{\mathcal{D}}
\newcommand{\M}[0]{\mathcal{M}}

\newcommand{\J}[0]{\mathcal{J}}

\newcommand{\cT}{\mathcal{T}}
\newcommand{\cR}{\mathcal{R}}

\newcommand{\stack}[0]{\mathcal{S}}
\newcommand{\states}[0]{\mathcal{Q}}

\newcommand{\autex}[0]{\mathcal{A}_{\mathit{ex}}}

\newcommand{\transover}[1]{\overset{#1}{\rightarrow}}
\newcommand{\ltr}[1]{\transover{#1}}

\newcommand{\word}[0]{\alpha}
\newcommand{\wordof}[1]{\seqof{\word}{#1}}
\newcommand{\dagg}[0]{\mathcal{G}}
\newcommand{\dagof}[1]{\dagg_{#1}}
\newcommand{\dagw}[0]{\dagof{\word}}

\newcommand{\seqof}[2]{#1_{#2}}
\newcommand{\dagwiof}[1]{\dagw^{#1}}
\newcommand{\graph}[0]{\mathcal{G}}
\newcommand{\graphof}[1]{\graph_{#1}}
\newcommand{\scc}[0]{\mathit{SCC}}
\newcommand{\sccof}[1]{\scc(#1)}
\newcommand{\infreach}[0]{\textit{inf-reach}}
\newcommand{\infreachof}[1]{\infreach(#1)}

\newcommand{\maxinfreachof}[1]{\lceil #1 \rceil}
\newcommand{\mininfreachof}[1]{\lfloor #1 \rfloor}
\newcommand{\condcoarse}[0]{\varphi_{\mathit{coarse}}}
\newcommand{\condfine}[0]{\varphi_{\mathit{fine}}}
\newcommand{\oddof}[1]{\mathit{odd}(#1)}

\newcommand{\algkv}[0]{\textsc{KV}\xspace}
\newcommand{\algfkv}[0]{\textsc{FKV}\xspace}
\newcommand{\algschewe}[0]{\textsc{Schewe}\xspace}
\newcommand{\algschewerao}[0]{\textsc{Schewe}_{\textsc{RedAvgOut}}\xspace}
\newcommand{\algmaxrank}[0]{\textsc{MaxRank}\xspace}
\newcommand{\algbackoff}[0]{\textsc{BackOff}\xspace}

\newcommand{\direct}[0]{\mathit{di}}

\newcommand{\simof}[1]{\mathrel{\preceq_{#1}}}

\newcommand{\simdi}[0]{\mathrel{\simof \direct}}

\newcommand{\defiff}[0]{\mathrel{\stackrel{\mathrm{def}}{\equiv}}}

\newcommand{\lang}[0]{\mathcal{L}}
\newcommand{\langof}[1]{\lang(#1)}
\newcommand{\langautof}[2]{\lang_{#1}(#2)}
\newcommand{\langaut}[1]{\lang_{\aut}(#1)}

\newcommand{\subseteqlang}[0]{\mathrel{\subseteq_{\lang}}}

\newcommand{\numsetof}[1]{[#1]}

\newcommand{\simby}[0]{\preceq}
\newcommand{\dirsimby}[0]{\mathrel{\simby_{\mathit{di}}}}

\newcommand{\limpl}[0]{\Rightarrow}

\newcommand{\rankrestr}[0]{\textsc{RankRestr}\xspace}

\newcommand{\algdelay}[0]{\textsc{Delay}\xspace}
\newcommand{\ranksimred}[0]{\textsc{RankSim}\xspace}
\newcommand{\succrankred}[0]{\textsc{SuccRank}\xspace}

\newcommand{\ignore}[0]{\cdot}

\newcommand{\auxdelta}[0]{\Delta^{\bullet}}

\newcommand{\fpmax}[0]{f'_{\mathit{max}}}
\newcommand{\maxrank}[0]{\textit{max-rank}}
\newcommand{\maxrankof}[1]{\maxrank(#1)}

\newcommand{\purgedi}[0]{\textsc{Purge}_{\mathit{di}}}

\newcommand{\ors}[0]{\mathrel{\preceq_{\mathit{ors}}}}
\newcommand{\orsf}[0]{\mathrel{\preceq^f_{\mathit{ors}}}}
\newcommand{\orsfT}[0]{\mathrel{\preceq^{fT}_{\mathit{ors}}}}

\newcommand{\orsr}[0]{\mathrel{\preceq_R}}
\newcommand{\orsraa}[0]{\mathrel{\preceq_{R}^{\forall\forall}}}
\newcommand{\orsrf}[0]{\mathrel{\preceq_R^f}}
\newcommand{\orsrfT}[0]{\mathrel{\preceq_R^{fT}}}

\newcommand{\bigO}[0]{\mathcal{O}}
\newcommand{\bigOof}[1]{\bigO(#1)}

\newcommand{\StateSize}[0]{\mathit{StateSize}}
\newcommand{\RankMax}[0]{\mathit{RankMax}}

\newcommand{\ranker}[0]{\textsc{Ranker}\xspace}
\newcommand{\rankermaxrank}[0]{\ranker_{\textsc{MaxR}}\xspace}
\newcommand{\rankerrankrestr}[0]{\ranker_{\textsc{RRestr}}\xspace}
\newcommand{\rabit}[0]{\textsc{Rabit}\xspace}
\newcommand{\spot}[0]{\textsc{Spot}\xspace}
\newcommand{\seminator}[0]{\textsc{Seminator}~2\xspace}
\newcommand{\goal}[0]{\textsc{GOAL}\xspace}
\newcommand{\roll}[0]{\textsc{ROLL}\xspace}
\newcommand{\fribourg}[0]{\textsc{Fribourg}\xspace}
\newcommand{\piterman}[0]{\textsc{Piterman}\xspace}
\newcommand{\safra}[0]{\textsc{Safra}\xspace}
\newcommand{\autfilt}[0]{\texttt{autfilt}\xspace}
\newcommand{\ltldstar}[0]{\textsc{LTL2dstar}\xspace}
\newcommand{\goalpic}[0]{{\scriptsize \faSoccerBallO}}

\definecolor{rowgray}{gray}{0.85}
\newcommand{\emphcell}[0]{\cellcolor{rowgray}}
\newcommand{\centercell}[1]{\multicolumn{1}{c}{#1}}

\title{Reducing (to) the Ranks (Technical Report)\\ {\Large Efficient Rank-based B\"{u}chi Automata Complementation}} %TODO Please add

\titlerunning{Reducing (to) the Ranks} %TODO optional, please use if title is longer than one line

\author{Vojt\v{e}ch Havlena}{Faculty of Information Technology, Brno University of Technology, Czech Republic}{ihavlena@fit.vutbr.cz}{https://orcid.org/0000-0003-4375-7954}{}%TODO mandatory, please use full name; only 1 author per \author macro; first two parameters are mandatory, other parameters can be empty. Please provide at least the name of the affiliation and the country. The full address is optional

\author{Ond\v{r}ej Leng\'{a}l}{Faculty of Information Technology, Brno University of Technology, Czech Republic}{lengal@fit.vutbr.cz}{https://orcid.org/0000-0002-3038-5875}{}

\authorrunning{V. Havlena, O. Leng\'{a}l} %TODO mandatory. First: Use abbreviated first/middle names. Second (only in severe cases): Use first author plus 'et al.'

\Copyright{Vojt\v{e}ch Havlena and Ond\v{r}ej Leng\'{a}l} %TODO mandatory, please use full first names. LIPIcs license is "CC-BY";  http://creativecommons.org/licenses/by/3.0/

 %TODO mandatory; please add comma-separated list of keywords

\category{} %optional, e.g. invited paper

\relatedversiondetails{CONCUR'21}{https://doi.org/10.4230/LIPIcs.CONCUR.2021.9}

\nolinenumbers %uncomment to disable line numbering

\EventEditors{Serge Haddad and Daniele Varacca}
\EventNoEds{2}
\EventLongTitle{32nd International Conference on Concurrency Theory (CONCUR 2021)}
\EventShortTitle{CONCUR 2021}
\EventAcronym{CONCUR}
\EventYear{2021}
\EventDate{August 23--27, 2021}
\EventLocation{Virtual Conference}
\EventLogo{}
\SeriesVolume{203}
\ArticleNo{9}
\hideLIPIcs

\begin{document}

\maketitle

\begin{abstract}
  This paper provides several optimizations of the rank-based approach for
  complementing B\"{u}chi automata.
  We start with Schewe's theoretically optimal construction and develop
  a~set of techniques for pruning its state space that are key to obtaining
  small complement automata in practice.
  In particular, the reductions (except one) have the property that they
  preserve (at least some) so-called \emph{super-tight runs}, which are runs
  whose ranking is \emph{as tight as possible}.
  Our evaluation on a~large benchmark shows that the optimizations indeed
  significantly help the rank-based approach and that, in a~large number of
  cases, the obtained complement is the smallest from those produced by
  state-of-the-art tools for B\"{u}chi complementation.
\end{abstract}

%*******************************************************************************
\vspace{-0.0mm}
\section{Introduction}\label{sec:intro}
\vspace{-0.0mm}
%*******************************************************************************

B\"uchi automata (BA) complementation remains an intensively studied problem
since 1962, when B\"uchi introduced the automata model over infinite words as
a~foundation for a~decision procedure of a~fragment of a~second-order
arithmetic~\cite{buchi1962decision}. Since then, efficient BA complementation
became an important task from both theoretical and practical side. It~is
a~crucial operation in some approaches for termination analysis of
programs~\cite{fogarty2009buchi,heizmann2014termination,ChenHLLTTZ18} as well as
in decision procedures concerning reasoning about programs and computer
systems, such as S1S~\cite{buchi1962decision} or the temporal logics ETL and
QPTL~\cite{SistlaVW85}.

B\"uchi launched a hunt for an optimal and efficient complementation technique with
his doubly exponential complementation approach~\cite{buchi1962decision}.
% based on the infinite Ramsey theorem.
A couple of years later, Safra proposed a~complementation via deterministic
Rabin automata with an $n^{\bigOof{n}}$
\emph{upper bound} of the size of the complement.
Simultaneously with finding an efficient complementation algorithm,
another search for the theoretical \emph{lower bound} was under way.
Michel showed in~\cite{michel1988complementation} that a~lower bound of the size of
a complement BA is $n!$ (approx.\ $(0.36n)^n$). This result was further refined
by Yan to $(0.76n)^n$ in~\cite{yan}.
From the theoretical point of view, it seemed that the
problem was already solved since Safra's construction asymptotically matched
the lower bound.
From the practical point of view, however, a~factor in the exponent
has a~great impact on the size of the complemented automaton
(and, consequently, also affects the performance of real-world applications).
This gap became a topic of many
works~\cite{KupfermanV01,FriedgutKV06,vardi2008automata,kahler2008complementation,yan2006lower}.
The efforts finally led to the~construction of Schewe in~\cite{Schewe09}
producing complement BAs whose sizes match the lower bound modulo
a~$\bigOof{n^2}$ polynomial factor.

Schewe's construction stores in a macrostate partial information about all
runs over some word in an input BA. In order to track the information about all
runs, a macrostate contains a set of states representing a single level in a run
DAG of some word with a number assigned to each state representing its rank. The
number of macrostates (and hence the size of the complement) is
combinatorially related to the maximum rank that occurs in macrostates.

Although the construction of Schewe is worst-case optimal, it may in practice
still generate a~lot of states that are not necessary.
% and negatively affect the
% size of the complemented automaton.
%
% \ifTR
% In our previous work in~\cite{ChenHL19}, we employed direct
% and delay simulation between states of the original automaton to remove states
% from the complement that are not necessary for accepting a~word.
% In the current paper,
% we build upon~\cite{ChenHL19} and develop other optimizations for reducing the
% size of the complement that push the
% rank-based approach by a~significant step further.
% \else
% The first attempt of optimizing real-world performance of the Schewe's
% construction was in~\cite{ChenHL19}.
% In~\cite{ChenHL19}, direct and delay simulations between states of the original
% automaton are employed to remove states from the complement that are not
% necessary for accepting a~word.
% In the current paper, we develop new optimizations for reducing the size of
% the complement that push the rank-based approach by a~significant step further.
% \fi
% \ol{this needs to be changed -- give a bigger overview, talk about super-tight
% runs, decreasing the max rank, etc}
%
In this work, we propose novel optimizations that (among others) reduce this
maximum considered rank.
We build on the novel notion of a~\emph{super-tight run}, i.e., a run
in the complement that uses as small ranks as possible.
The macrostates not occurring in some super-tight run can be safely removed from the
automaton. % without affecting its language.
Further, based on reasoning about super-tight runs, we are able to
reduce the maximum rank within a~macrostate.
In particular, we reduce the maximum
considered ranking using a~reasoning about the deterministic support of an input
automaton or by a relation based on direct simulation implying rank ordering
computed \emph{a priori} from the input automaton.
The developed optimizations give, to the best of our knowledge, \textbf{the
most competitive BA complementation procedure}, as witnessed by our %extensive
experimental evaluation.

% Namely,
% \begin{inparaenum}[(i)]
% 	\item \emph{Delaying transitions:} reducing the number of nondeterministical
% 		guesses of states introducing tight ranks.
% 	\item \emph{Successors constraints:} reducing the maximum considered ranking
% 	in the tight part based on reasoning about the DFA support of an input automaton.
% 	\item \emph{Rank simulation:} removing macrostates with incompatible rankings
% 	based on reasoning about states with an odd ranking.
% 	\item \emph{Max rank construction:} considering only runs with the maximal
% 	ranking.
% \end{inparaenum}

These optimizations require some additional computational cost,
but from the perspective of BA complementation, their cost is still negligible
and, as we show in our experimental evaluation, their effect on the size of the
output is often profound, in many cases by one or more orders of magnitude.
Rank-based complementation with our optimizations is now competitive with other
approaches, in a~large number of cases (21\,\%) obtaining
a~strictly smaller complement than \emph{any other existing tool} and in the majority of
cases (63\,\%) obtaining an automaton at least as small as the smallest
automaton provided by any other tool.

%*******************************************************************************
\vspace{-0.0mm}
\section{Preliminaries}\label{sec:prelims}
\vspace{-0.0mm}
%*******************************************************************************

% %*******************************************************************************
% \vspace{-0.0mm}
% \subsection{B\"{u}chi Automata}\label{sec:label}
% \vspace{-0.0mm}
% %*******************************************************************************

% \subparagraph*{Functions, words, and alphabets.}
We fix a~finite nonempty alphabet~$\Sigma$ and the first infinite ordinal
$\omega = \{0, 1, \ldots\}$.
For $n\in\omega$, by $\numsetof{n}$ we denote the
set $\{ 0, \dots, n \}$.
An (infinite) word~$\word$ is represented as a~function
$\word\colon \omega \to \Sigma$ where the $i$-th symbol is denoted as $\wordof i$.
% \vh{A~finite word~$w$ of length $n+1$ is represented as a~function $w: \numsetof{n}
% \to \Sigma$.
% The finite word of length~$0$ is denoted as~$\epsilon$.}
We~abuse
notation and sometimes also represent~$\word$ as an~infinite sequence $\word =
\wordof 0 \wordof 1 \dots$
% \vh{ and $w$ as a finite sequence $w = w_0\dots w_{n-1}$.}
The suffix $\wordof i \wordof{i+1} \ldots$ of~$\word$ is denoted
by~$\wordof {i:\omega}$.
We~use~$\Sigma^\omega$ to denote the set of all infinite words over~$\Sigma$.
% \vh{ and~$\Sigma^*$ to denote the set of all finite words.
% For $L\subseteq \Sigma^*$ we define $L^* = \{
% u\in\Sigma^*~|~u=w_1\cdots w_n \wedge \forall 1 \leq i \leq n: w_i\in L \}$ and
% $L^\omega = \{ \word\in\Sigma^\omega~|~\alpha=w_1w_2\cdots \wedge \forall i \geq
% 1: w_i\in L \}$ (note that $\{\epsilon\}^\omega = \emptyset$).
% Given $L_1, L_2 \subseteq \Sigma^*$, we use $L_1 L_2$ to denote the set
% $\{w_1 w_2 \mid w_1 \in L_1, w_2 \in L_2\}$.}
Furthermore, for a~total function $f\colon X \to Y$ and a partial function $h\colon X
\partto Y$, we use $\variant f h$ to denote the total function $g\colon  X \to Y$
defined as $g(x) = h(x)$ when $h(x)$ is defined and $g(x) = f(x)$ otherwise.
Moreover, we use $\imgof f$ to denote the \emph{image} of~$f$, i.e., $\imgof f =
\{f(x) \in Y \mid x \in X\}$ and for a set $C\subseteq X$ we use $f_{|C}$ to
denote the \emph{restriction} of $f$ to $C$, i.e., $f_{|C} = f\cap(C\times Y)$.

\vspace{-0mm}
\subparagraph*{B\"{u}chi automata.}
A~(nondeterministic) \emph{B\"{u}chi automaton} (BA) over~$\Sigma$
is a~quadruple $\aut = (Q, \delta, I, F)$ where $Q$ is a~finite set of
\emph{states}, $\delta$ is a~\emph{transition function} $\delta\colon Q \times
\Sigma \to 2^Q$, and $I, F \subseteq Q$ are the sets of \emph{initial} and
\emph{accepting} states respectively.
We sometimes treat~$\delta$ as a~set of transitions $p \ltr a q$, for instance,
we use $p \ltr a q \in \delta$ to denote that $q \in \delta(p, a)$.
Moreover, we extend $\delta$ to sets of states $P
\subseteq Q$ as $\delta(P, a) = \bigcup_{p \in P} \delta(p,a)$.
We use $\delta^{-1}(q, a)$ to denote the set $\{s \in Q \mid s \ltr{a} q \in \delta\}$.
For a set of states $S$ we define \emph{reachability} from $S$ as $\reach_\delta(S) = \mu Z.\,S \cup \bigcup_{a\in\Sigma}\delta(Z, a)$. A~\emph{run}
of~$\aut$ from~$q \in Q$ on an input word $\word$ is an infinite sequence $\rho\colon
\omega \to Q$ that starts in~$q$ and respects~$\delta$, i.e., $\rho_0 = q$ and
$\forall i \geq 0\colon \rho_i \ltr{\wordof i}\rho_{i+1} \in \delta$.
Let $\infof \rho$ denote the states occurring in~$\rho$ infinitely often.
We say that~$\rho$ is \emph{accepting} iff $\infof \rho \cap F \neq \emptyset$.
A~word $\word$ is accepted by~$\aut$ from a~state~$q \in Q$ if there is an
accepting run $\rho$ of $\aut$ from~$q$, i.e., $\rho_0 = q$. The set
$\langautof{\aut} q = \{\word \in \Sigma^\omega \mid \aut \text{ accepts } \word
\text{ from } q\}$ is called the \emph{language} of~$q$ (in~$\aut$). Given a~set
of states~$R \subseteq Q$, we define the language of~$R$ as $\langautof \aut R =
\bigcup_{q \in R} \langautof \aut q$ and the language of~$\aut$ as~$\langof \aut =
\langautof \aut I$. For a~pair of states~$p$ and~$q$ in~$\aut$, we use $p
\subseteqlang q$ to denote $\langaut p \subseteq \langaut q$.
%
% Without loss of generality, in this paper, we assume~$\aut$ to be complete, i.e., for every state~$q$ and
% symbol~$a$, it holds that $\delta(q, a) \neq \emptyset$.
$\aut$~is complete iff for every state~$q$ and
symbol~$a$, it holds that $\delta(q, a) \neq \emptyset$.
% \vh{A~\emph{trace} over a~word~$\word$ is an infinite sequence $\pi = q_0
% \ltr{\wordof 0} q_1 \ltr{\wordof 1} \cdots$ such that $\rho = q_0 q_1 \ldots$ is
% a~run of~$\aut$ over~$\word$ from~$q_0$.
% We say~$\pi$ is \emph{fair} if it contains infinitely many accepting states.
% Moreover, we use $p \transover{w} q$ for $w \in \Sigma^*$ to denote that~$q$ is
% reachable from~$p$ over the word~$w$; if a~path from~$p$ to~$q$ over~$w$ contains an
% accepting state, we can write $p \ftransover{w} q$.}
In this paper, we fix a~BA $\aut = (Q, \delta, I, F)$.

\vspace{-0mm}
\subparagraph*{Simulation.} The \emph{(maximum) direct simulation} on~$\aut$ is the relation
${\dirsimby} \subseteq Q\times Q$ defined as the largest relation s.t.~$p\dirsimby q$ implies
\begin{inparaenum}[(i)]
	\item $p \in F \limpl q \in F$ and
	\item $p \ltr{a} p' \in \delta \limpl \exists q'\in Q\colon q \ltr{a} q'
    \in \delta \wedge p' \dirsimby q'$ for each $a\in\Sigma$.
\end{inparaenum}
Note that~$\dirsimby$ is a~preorder and $\dirsimby\ \subseteq\
\subseteqlang$~\cite{MayrC13}.

\vspace{-0.0mm}
\section{Complementing B\"{u}chi Automata}\label{sec:complement}
\vspace{-0.0mm}
%%%%%%%%%%%%%%%%%%%%%%%%%%%%%%%%%%%%%%%%%%%%%%%%%%%%%%%%%%%%%%%%%%%%%%%%%%%%%%%%
In this section we first describe the basic rank-based complementation algorithm
proposed by Kupferman and Vardi in~\cite{KupfermanV01} and then its optimization
presented by Schewe in~\cite{Schewe09}.
% In order to help the reader build an intuition about the algorithm, we start
% from the simplest rank-based complementation algorithm proposed by Kupferman and
% Vardi in~\cite{KupfermanV01} and then refine it in two steps.
% Readers familiar with the constructions can skim through these parts
% (\crefrange{sec:rundag}{sec:schewe}).
After that, we present some results related to runs with the minimal ranking.
Missing proofs for this and the following section can be found in~\cite{techrep}.

%*******************************************************************************
\vspace{-0.0mm}
\subsection{Run DAGs}\label{sec:rundag}
\vspace{-0.0mm}
%*******************************************************************************
\newcommand{\figexampleaut}[0]{
% \begin{figure}
% \begin{center}
% \input{figs/aut_ex.tikz}
% \end{center}
% \vspace*{-5mm}
% \caption{The BA $\autex$}
% \label{fig:aut_ex}
% \end{figure}
%
\begin{figure}
  \begin{minipage}{0.22\textwidth}
	\begin{subfigure}{\textwidth}
		\centering
  \scalebox{0.9}{
    \begin{tikzpicture}[->,>=stealth',shorten >=0pt,auto,node distance=1.5cm,
                    scale=0.8,transform shape,initial text={}]
  \tikzstyle{every state}=[inner sep=3pt,minimum size=5pt]
  \tikzstyle{empty}=[]
  \tikzstyle{initstate}=[fill=yellow!30]

  \node[state,initstate,accepting] (r) {$r$};
  \node[state,initstate,right of=r] (s) {$s$};
  \node[state,accepting,right of=s] (t) {$t$};

  \node[empty,above of=r,node distance=8mm] (above_r) {};
  \node[empty,above of=s,node distance=8mm] (above_s) {};

  \path (above_r) edge (r);
  \path (above_s) edge (s);

  \path (r) edge[loop below]  node {$a$} (r)
        (r) edge  node {$b$} (s)
        (s) edge[loop below] node {$b$} (s)
        (s) edge  node {$b$} (t)
        (t) edge[loop below] node {$a$} (t);

\end{tikzpicture}
  }
	\caption{\centering}
	\label{fig:example_aut}
	\end{subfigure}\\
	\begin{subfigure}{\textwidth}
		\centering
  \scalebox{0.9}{
    \begin{tikzpicture}[->,>=stealth',shorten >=0pt,auto,node distance=1.2cm,
                    scale=0.8,transform shape,initial text={}]
  \tikzstyle{every state}=[inner sep=3pt,minimum size=5pt,rectangle,rounded corners=1mm]
  \tikzstyle{empty}=[]

  \node[state,accepting]   (r0) {$r,0$};
  \node[state,right of=r0] (s0) {$s,0$};

  \node[state,below of=s0]           (s1) {$s,1$};
  \node[state,accepting,right of=s1] (t1) {$t,1$};

  \node[state,below of=s1]           (s2) {$s,2$};
  \node[state,accepting,right of=s2] (t2) {$t,2$};

  % \node[state,below of=s2]           (s3) {$s,3$};
  % \node[state,accepting,right of=s3] (t3) {$t,3$};

  \node[empty,below of=s2,node distance=4mm] (s3) {$\vdots$};
  \node[empty,right of=s3,node distance=4mm] (t3) {$\ddots$};

  \draw (r0) edge (s1);
  \draw (s0) edge (s1);
  \draw (s0) edge (t1);

  \draw (s1) edge (s2);
  \draw (s1) edge (t2);

  % \draw (s2) edge (s3);
  % \draw (s2) edge (t3);

  \node[empty, node distance=8mm,below left of=r0] (sym1) {$b$};
  \node[empty, below of=sym1] (sym2) {$b$};
  % \node[empty, below of=sym2] (sym3) {$b$};
  \node[empty, below of=sym2,node distance=10mm] (sym3) {$\vdots$};

  \begin{pgfonlayer}{background}
  \node[draw,dashed,rectangle,fill=YellowGreen!80,draw=black!70,rounded corners=8pt,inner sep=3pt,fit=(r0) (r0) (r0) (r0)] (a) {};
  \node[draw,dashed,rectangle,fill=YellowGreen!50,draw=black!70,rounded corners=8pt,inner sep=3pt,fit=(s0) (s0) (s2) (s2)] (b) {};
  \node[draw,dashed,rectangle,fill=YellowGreen!20,draw=black!70,rounded corners=8pt,inner sep=3pt,fit=(t1) (t1) (t2) (t2)] (c) {};
  \end{pgfonlayer}

  \node[empty, above of=r0, node distance=8mm, text=black] (rank2) {\it rank 2};
  \node[empty, above of=s0, node distance=8mm, text=black] (rank1) {\it rank 1};
  \node[empty, above of=t1, node distance=8mm, text=black] (rank0) {\it rank 0};

\end{tikzpicture}
  }
	\caption{\centering}
	\label{fig:example_dag}
	\end{subfigure}
  \end{minipage}
  \begin{subfigure}{0.32\textwidth}
    \centering
    \vspace*{8mm}
    \scalebox{0.9}{
      \begin{tikzpicture}[->,>=stealth',shorten >=0pt,auto,node distance=1.2cm,
                    scale=0.8,transform shape,initial text={}]
  \tikzstyle{every state}=[inner sep=3pt,minimum size=5pt,rectangle,rounded corners=1mm]
  \tikzstyle{empty}=[inner sep=0pt]
  \tikzstyle{initstate}=[fill=yellow!30]

  \node[state, initial,initstate,accepting] (r1) {$\big(\{r{:}4, s{:}4\}, \emptyset\big)$};
  \node[state, below of=r1] (r2) {$\big(\{s{:}4, t{:}4\}, \{s,t\}\big)$};
  \node[state, below of=r2] (r3) {$\big(\{s{:}3, t{:}4\}, \{t\}\big)$};
  \node[state, accepting, below of=r3] (r4) {$\big(\{s{:}3, t{:}2\}, \emptyset\big)$};
  \node[state, below of=r4] (r5) {$\big(\{s{:}3, t{:}2\}, \{t\}\big)$};

  \node[empty,minimum size=30pt,xshift=5.5mm] at(r2.west) (r2inv) {};

  \node[state, white, fill=white!20, right of=r3, node distance=25mm, minimum height=50mm, minimum width=8mm] (sink) {};

  \path (r1) edge  node [right] {$b$} (r2)
        (r2inv) edge[loop left]  node[pos=0.15,below] {$b$} (r2inv)
        (r2) edge  node [right] {$b$} (r3)
        (r3) edge  node [right] {$b$} (r4)
        (r4) edge [bend left]  node [right] {$b$} (r5)
        (r5) edge [bend left]  node [left] {$b$} (r4);

  \begin{pgfonlayer}{background}
    \path[dashed,gray] (r1)  edge [bend left]  node {} (sink)
      (r2) edge [bend left] node {} (sink)
      (r3) edge node {} (sink)
      (r4) edge [bend right]  node {} (sink)
      (r5) edge [bend right]  node {} (sink);
  \end{pgfonlayer}

\end{tikzpicture}
    }
    \vspace*{4mm}
    \caption{\centering }
    \label{fig:kv_example}
  \end{subfigure}
	\begin{subfigure}{0.37\textwidth}
    \vspace*{4mm}
 	   \centering
		 \scalebox{0.9}{
	 	 \begin{tikzpicture}[->,>=stealth',shorten >=0pt,auto,node distance=1.8cm,
                    scale=0.8,transform shape,initial text={}]
  \tikzstyle{every state}=[inner sep=3pt,minimum size=5pt,rectangle,rounded corners=1mm]
  \tikzstyle{empty}=[]
  \tikzstyle{initstate}=[fill=yellow!30]
  \tikzstyle{wobbly}=[decorate, decoration={snake,amplitude=.2mm,segment length=2mm,post length=1mm}]

  \node[state,initial,initstate] (rs) {$\{r,s\}$};
  \node[state, below of=rs] (r) {$\{r\}$};
  \node[state, left of=rs,yshift=10mm] (st) {$\{s,t\}$};
  \node[state, below of=r] (s) {$\{s\}$};
  \node[state, left of=r] (t) {$\{t\}$};
  \node[state, accepting, left of=s] (em) {$\emptyset$};

  \node[state, accepting, right of=rs, node distance=30mm,yshift=10mm] (r1) {$\big(\{s{:}1, t{:}0\}, \emptyset\big)$};
  \node[state, below of=r1,xshift=-3mm] (r2) {$\big(\{s{:}1, t{:}0\}, \{t\}\big)$};
  \node[state, accepting,below of=r2,yshift=-4mm] (r3) {$\big(\{s{:}1\},\emptyset\big)$};

  % \node[state, white, fill=white, below of=r2, xshift=-8mm, node distance=17mm, minimum width=20mm, minimum height=10mm] (sink) {};

  \path (rs) edge  node[right] {$a$} (r)
        (rs) edge  node[above,pos=0.2] {$b$} (st)
        (r) edge  node[right] {$b$} (s)
        (r) edge  node[below,pos=0.2] {$b$} (r3)
        (s) edge  node[xshift=1mm] {$b$} (st)
        % (s) edge[bend right=50]  node {$b$} (r1)
        (s) edge  node[above] {$a$} (em)
        (st) edge  node[left] {$a$} (t)
        (t) edge  node[left] {$b$} (em)
        (r) edge [loop right] node[pos=0.2,above] (lab2) {$a$} (r)
        (t) edge [loop left] node {$a$} (t)
        (em) edge [loop left] node (lab1) {$a,b$} (em)
        (st) edge [loop left] node {$b$} (st)
        (rs) edge   node[pos=0.3] {$b$} (r1)
        (st) edge  node {$b$} (r1)
        (r1) edge [bend left]  node {$b$} (r2)
        (r2) edge [bend left]  node {$b$} (r1)
        (r3) edge  node[right] {$b$} (r2);

  \draw (s) .. controls +(50mm,-10mm) and +(10mm,-20mm) .. node[pos=0.1,above,yshift=-0.5mm] {$b$} (r1.south east);

  \begin{pgfonlayer}{background}
  \node[draw,dashed,rectangle,fill=red!15,draw=black!70,rounded corners=8pt,inner sep=3pt,fit=(r1) (r2) (r3)] (tight) {};
  \node[draw,dashed,rectangle,fill=blue!15,draw=black!70,rounded corners=8pt,inner sep=3pt,fit=(st) (rs) (r) (t) (em) (s) (lab1) (lab2)] (waiting) {};
  \end{pgfonlayer}

  \node[above of=st,node distance=6.5mm,xshift=-8mm] {\emph{waiting}};
  \node[above of=r1,node distance=7mm,xshift=8mm] {\emph{tight}};

\end{tikzpicture}
		 }
	 	 \caption{\centering}
	 	 \label{fig:fkv_example}
	\end{subfigure}
	\caption{
    (\subref{fig:example_aut})~$\autex$.
    (\subref{fig:example_dag})~The run DAG of~$\autex$ over~$b^\omega$.
    (\subref{fig:kv_example})~A~part of $\algkv(\autex)$.
    (\subref{fig:fkv_example})~$\algfkv(\autex)$; we highlight the
      \emph{waiting} and the \emph{tight} parts.}
    \vspace{-2mm}
\end{figure}
}

\newcommand{\figschewe}[0]{
\begin{figure}[t]
	% \begin{subfigure}{0.37\textwidth}
  %   \vspace*{4mm}
 	   \centering
		 \scalebox{0.9}{
	 	 \begin{tikzpicture}[->,>=stealth',shorten >=0pt,auto,node distance=1.8cm,
                    scale=0.8,transform shape,initial text={}]
  \tikzstyle{every state}=[inner sep=3pt,minimum size=5pt,rectangle,rounded corners=1mm]
  \tikzstyle{empty}=[]
  \tikzstyle{initstate}=[fill=yellow!30]
  \tikzstyle{wobbly}=[decorate, decoration={snake,amplitude=.2mm,segment length=2mm,post length=1mm}]

  \node[state,initial,initstate] (rs) {$\{r,s\}$};
  \node[state, below of=rs] (r) {$\{r\}$};
  \node[state, left of=rs,yshift=10mm] (st) {$\{s,t\}$};
  \node[state, below of=r] (s) {$\{s\}$};
  \node[state, left of=r] (t) {$\{t\}$};
  \node[state, accepting, left of=s] (em) {$\emptyset$};

  \node[state, accepting, right of=rs, node distance=30mm,yshift=10mm] (r1) {$\big(\{s{:}1, t{:}0\}, \emptyset,0\big)$};
  \node[state, below of=r1,xshift=-3mm] (r2) {$\big(\{s{:}1, t{:}0\}, \{t\},0\big)$};
  \node[state, fill=gray!40,accepting,below of=r2,yshift=-4mm] (r3) {$\big(\{s{:}1\},\emptyset,0\big)$};

  % \node[state, white, fill=white, below of=r2, xshift=-8mm, node distance=17mm, minimum width=20mm, minimum height=10mm] (sink) {};

  \path (rs) edge  node[right] {$a$} (r)
        (rs) edge  node[above,pos=0.2] {$b$} (st)
        (r) edge  node[right] {$b$} (s)
        (r) edge[wobbly]  node[below,pos=0.2] {$b$} (r3)
        (s) edge  node[xshift=1mm] {$b$} (st)
        % (s) edge[bend right=50]  node {$b$} (r1)
        (s) edge  node[above] {$a$} (em)
        (st) edge  node[left] {$a$} (t)
        (t) edge  node[left] {$b$} (em)
        (r) edge [loop right] node[pos=0.2,above] (lab2) {$a$} (r)
        (t) edge [loop left] node {$a$} (t)
        (em) edge [loop left] node (lab1) {$a,b$} (em)
        (st) edge [loop left] node {$b$} (st)
        (rs) edge[wobbly]  node[pos=0.3] {$b$} (r1)
        (st) edge  node {$b$} (r1)
        (r1) edge [bend left]  node {$b$} (r2)
        (r2) edge [bend left]  node {$b$} (r1)
        (r3) edge [wobbly]  node[right] {$b$} (r2);

  \draw[wobbly] (s) .. controls +(50mm,-10mm) and +(10mm,-20mm) .. node[pos=0.1,above,yshift=-0.5mm] {$b$} (r1.south east);

  \begin{pgfonlayer}{background}
  \node[draw,dashed,rectangle,fill=red!15,draw=black!70,rounded corners=8pt,inner sep=3pt,fit=(r1) (r2) (r3)] (tight) {};
  \node[draw,dashed,rectangle,fill=blue!15,draw=black!70,rounded corners=8pt,inner sep=3pt,fit=(st) (rs) (r) (t) (em) (s) (lab1) (lab2)] (waiting) {};
  \end{pgfonlayer}

  \node[above of=st,node distance=6.5mm,xshift=-8mm] {\emph{waiting}};
  \node[above of=r1,node distance=7mm,xshift=8mm] {\emph{tight}};

\end{tikzpicture}
		 }
	%  	 \label{fig:schewe_example}
	% \end{subfigure}
\caption{
    $\algschewe(\autex)$.
      The optimization \algdelay (\cref{sec:delay}) will remove the 4~wobbly transitions and macrostate $(\{s{:}1\},\emptyset,0)$.}
\label{fig:schewe_example}
\end{figure}
}

% \newcommand{\figexampledag}[0]{
% \begin{wrapfigure}[9]{r}{3.8cm}
% \vspace*{-37mm}
% \hspace*{-1mm}
% \begin{minipage}{3.8cm}
% \input{figs/run_dag_ex.tikz}
% \end{minipage}
% \caption{The run DAG of~$\autex$ over~$b^\omega$}
% \label{fig:example_dag}
% \end{wrapfigure}
% }
%
% \figexampledag

% \begin{figure}[t]
%   \begin{subfigure}[b]{0.49\linewidth}
%   % \includegraphics[width=0.49\linewidth,keepaspectratio]{figs/rundag_example_aut.png}
%   \input{figs/aut_ex.tikz}
%   % \caption{An example of a~B\"{u}chi automaton~$\autex$ with the language $a^\omega + (a^*
%   %   bb^+ + b^+)a^\omega$}
%   \caption{The BA $\autex$}
%   \label{label}
%   \end{subfigure}
%   \begin{subfigure}[b]{0.49\linewidth}
%   % \includegraphics[width=0.49\linewidth,keepaspectratio]{figs/rundag_example_dag.png}
%   \input{figs/run_dag_ex.tikz}
%   \caption{The run DAG of~$\autex$ over~$b^\omega$}
%   \end{subfigure}
% \caption{A B\"{u}chi automaton and its run DAG}
% \label{fig:example_dag}
% \end{figure}

In this section, we recall the terminology from~\cite{Schewe09} (which is
a~minor modification of the terminology from~\cite{KupfermanV01}), which we use
heavily in the paper. We fix the definition of the \emph{run DAG} of~$\aut$ over
a~word~$\word$ to be a~DAG (directed acyclic graph) $\dagw = (V,E)$ of
vertices~$V$ and edges~$E$ where
\vspace{-0mm}
\begin{itemize}
  \setlength{\itemsep}{0mm}
  \item  $V \subseteq Q \times \omega$ s.t. $(q, i) \in V$ iff there is
    a~run~$\rho$ of $\aut$ from $I$ over $\word$ with $\rho_i = q$,
  \item  $E \subseteq V \times V$ s.t.~$((q, i), (q',i')) \in E$ iff $i' = i+1$
    and $q' \in \delta(q, \wordof i)$.
\end{itemize}
\vspace{-0mm}

\noindent
Given $\dagw$ as above, we will write $(p, i) \in \dagw$ to denote that $(p, i) \in V$.
We call $(p,i)$ \emph{accepting} if $p$ is an accepting state.
$\dagw$ is \emph{rejecting} if it contains no path with infinitely many
accepting vertices.
A~vertex~$v \in \dagw$ is \emph{finite} if the set of vertices reachable
from $v$ is finite, \emph{infinite} if it is not finite, and
\emph{endangered} if $v$ cannot reach an accepting vertex.

% \begin{figure}
% \begin{center}
% \input{figs/aut_ex.tikz}
% \end{center}
% \vspace*{-5mm}
% \caption{The BA $\autex$}
% \label{fig:aut_ex}
% \end{figure}
%
\begin{figure}
  \begin{minipage}{0.22\textwidth}
	\begin{subfigure}{\textwidth}
		\centering
  \scalebox{0.9}{
    \begin{tikzpicture}[->,>=stealth',shorten >=0pt,auto,node distance=1.5cm,
                    scale=0.8,transform shape,initial text={}]
  \tikzstyle{every state}=[inner sep=3pt,minimum size=5pt]
  \tikzstyle{empty}=[]
  \tikzstyle{initstate}=[fill=yellow!30]

  \node[state,initstate,accepting] (r) {$r$};
  \node[state,initstate,right of=r] (s) {$s$};
  \node[state,accepting,right of=s] (t) {$t$};

  \node[empty,above of=r,node distance=8mm] (above_r) {};
  \node[empty,above of=s,node distance=8mm] (above_s) {};

  \path (above_r) edge (r);
  \path (above_s) edge (s);

  \path (r) edge[loop below]  node {$a$} (r)
        (r) edge  node {$b$} (s)
        (s) edge[loop below] node {$b$} (s)
        (s) edge  node {$b$} (t)
        (t) edge[loop below] node {$a$} (t);

\end{tikzpicture}
  }
	\caption{\centering}
	\label{fig:example_aut}
	\end{subfigure}\\
	\begin{subfigure}{\textwidth}
		\centering
  \scalebox{0.9}{
    \begin{tikzpicture}[->,>=stealth',shorten >=0pt,auto,node distance=1.2cm,
                    scale=0.8,transform shape,initial text={}]
  \tikzstyle{every state}=[inner sep=3pt,minimum size=5pt,rectangle,rounded corners=1mm]
  \tikzstyle{empty}=[]

  \node[state,accepting]   (r0) {$r,0$};
  \node[state,right of=r0] (s0) {$s,0$};

  \node[state,below of=s0]           (s1) {$s,1$};
  \node[state,accepting,right of=s1] (t1) {$t,1$};

  \node[state,below of=s1]           (s2) {$s,2$};
  \node[state,accepting,right of=s2] (t2) {$t,2$};

  % \node[state,below of=s2]           (s3) {$s,3$};
  % \node[state,accepting,right of=s3] (t3) {$t,3$};

  \node[empty,below of=s2,node distance=4mm] (s3) {$\vdots$};
  \node[empty,right of=s3,node distance=4mm] (t3) {$\ddots$};

  \draw (r0) edge (s1);
  \draw (s0) edge (s1);
  \draw (s0) edge (t1);

  \draw (s1) edge (s2);
  \draw (s1) edge (t2);

  % \draw (s2) edge (s3);
  % \draw (s2) edge (t3);

  \node[empty, node distance=8mm,below left of=r0] (sym1) {$b$};
  \node[empty, below of=sym1] (sym2) {$b$};
  % \node[empty, below of=sym2] (sym3) {$b$};
  \node[empty, below of=sym2,node distance=10mm] (sym3) {$\vdots$};

  \begin{pgfonlayer}{background}
  \node[draw,dashed,rectangle,fill=YellowGreen!80,draw=black!70,rounded corners=8pt,inner sep=3pt,fit=(r0) (r0) (r0) (r0)] (a) {};
  \node[draw,dashed,rectangle,fill=YellowGreen!50,draw=black!70,rounded corners=8pt,inner sep=3pt,fit=(s0) (s0) (s2) (s2)] (b) {};
  \node[draw,dashed,rectangle,fill=YellowGreen!20,draw=black!70,rounded corners=8pt,inner sep=3pt,fit=(t1) (t1) (t2) (t2)] (c) {};
  \end{pgfonlayer}

  \node[empty, above of=r0, node distance=8mm, text=black] (rank2) {\it rank 2};
  \node[empty, above of=s0, node distance=8mm, text=black] (rank1) {\it rank 1};
  \node[empty, above of=t1, node distance=8mm, text=black] (rank0) {\it rank 0};

\end{tikzpicture}
  }
	\caption{\centering}
	\label{fig:example_dag}
	\end{subfigure}
  \end{minipage}
  \begin{subfigure}{0.32\textwidth}
    \centering
    \vspace*{8mm}
    \scalebox{0.9}{
      \begin{tikzpicture}[->,>=stealth',shorten >=0pt,auto,node distance=1.2cm,
                    scale=0.8,transform shape,initial text={}]
  \tikzstyle{every state}=[inner sep=3pt,minimum size=5pt,rectangle,rounded corners=1mm]
  \tikzstyle{empty}=[inner sep=0pt]
  \tikzstyle{initstate}=[fill=yellow!30]

  \node[state, initial,initstate,accepting] (r1) {$\big(\{r{:}4, s{:}4\}, \emptyset\big)$};
  \node[state, below of=r1] (r2) {$\big(\{s{:}4, t{:}4\}, \{s,t\}\big)$};
  \node[state, below of=r2] (r3) {$\big(\{s{:}3, t{:}4\}, \{t\}\big)$};
  \node[state, accepting, below of=r3] (r4) {$\big(\{s{:}3, t{:}2\}, \emptyset\big)$};
  \node[state, below of=r4] (r5) {$\big(\{s{:}3, t{:}2\}, \{t\}\big)$};

  \node[empty,minimum size=30pt,xshift=5.5mm] at(r2.west) (r2inv) {};

  \node[state, white, fill=white!20, right of=r3, node distance=25mm, minimum height=50mm, minimum width=8mm] (sink) {};

  \path (r1) edge  node [right] {$b$} (r2)
        (r2inv) edge[loop left]  node[pos=0.15,below] {$b$} (r2inv)
        (r2) edge  node [right] {$b$} (r3)
        (r3) edge  node [right] {$b$} (r4)
        (r4) edge [bend left]  node [right] {$b$} (r5)
        (r5) edge [bend left]  node [left] {$b$} (r4);

  \begin{pgfonlayer}{background}
    \path[dashed,gray] (r1)  edge [bend left]  node {} (sink)
      (r2) edge [bend left] node {} (sink)
      (r3) edge node {} (sink)
      (r4) edge [bend right]  node {} (sink)
      (r5) edge [bend right]  node {} (sink);
  \end{pgfonlayer}

\end{tikzpicture}
    }
    \vspace*{4mm}
    \caption{\centering }
    \label{fig:kv_example}
  \end{subfigure}
	\begin{subfigure}{0.37\textwidth}
    \vspace*{4mm}
 	   \centering
		 \scalebox{0.9}{
	 	 \begin{tikzpicture}[->,>=stealth',shorten >=0pt,auto,node distance=1.8cm,
                    scale=0.8,transform shape,initial text={}]
  \tikzstyle{every state}=[inner sep=3pt,minimum size=5pt,rectangle,rounded corners=1mm]
  \tikzstyle{empty}=[]
  \tikzstyle{initstate}=[fill=yellow!30]
  \tikzstyle{wobbly}=[decorate, decoration={snake,amplitude=.2mm,segment length=2mm,post length=1mm}]

  \node[state,initial,initstate] (rs) {$\{r,s\}$};
  \node[state, below of=rs] (r) {$\{r\}$};
  \node[state, left of=rs,yshift=10mm] (st) {$\{s,t\}$};
  \node[state, below of=r] (s) {$\{s\}$};
  \node[state, left of=r] (t) {$\{t\}$};
  \node[state, accepting, left of=s] (em) {$\emptyset$};

  \node[state, accepting, right of=rs, node distance=30mm,yshift=10mm] (r1) {$\big(\{s{:}1, t{:}0\}, \emptyset\big)$};
  \node[state, below of=r1,xshift=-3mm] (r2) {$\big(\{s{:}1, t{:}0\}, \{t\}\big)$};
  \node[state, accepting,below of=r2,yshift=-4mm] (r3) {$\big(\{s{:}1\},\emptyset\big)$};

  % \node[state, white, fill=white, below of=r2, xshift=-8mm, node distance=17mm, minimum width=20mm, minimum height=10mm] (sink) {};

  \path (rs) edge  node[right] {$a$} (r)
        (rs) edge  node[above,pos=0.2] {$b$} (st)
        (r) edge  node[right] {$b$} (s)
        (r) edge  node[below,pos=0.2] {$b$} (r3)
        (s) edge  node[xshift=1mm] {$b$} (st)
        % (s) edge[bend right=50]  node {$b$} (r1)
        (s) edge  node[above] {$a$} (em)
        (st) edge  node[left] {$a$} (t)
        (t) edge  node[left] {$b$} (em)
        (r) edge [loop right] node[pos=0.2,above] (lab2) {$a$} (r)
        (t) edge [loop left] node {$a$} (t)
        (em) edge [loop left] node (lab1) {$a,b$} (em)
        (st) edge [loop left] node {$b$} (st)
        (rs) edge   node[pos=0.3] {$b$} (r1)
        (st) edge  node {$b$} (r1)
        (r1) edge [bend left]  node {$b$} (r2)
        (r2) edge [bend left]  node {$b$} (r1)
        (r3) edge  node[right] {$b$} (r2);

  \draw (s) .. controls +(50mm,-10mm) and +(10mm,-20mm) .. node[pos=0.1,above,yshift=-0.5mm] {$b$} (r1.south east);

  \begin{pgfonlayer}{background}
  \node[draw,dashed,rectangle,fill=red!15,draw=black!70,rounded corners=8pt,inner sep=3pt,fit=(r1) (r2) (r3)] (tight) {};
  \node[draw,dashed,rectangle,fill=blue!15,draw=black!70,rounded corners=8pt,inner sep=3pt,fit=(st) (rs) (r) (t) (em) (s) (lab1) (lab2)] (waiting) {};
  \end{pgfonlayer}

  \node[above of=st,node distance=6.5mm,xshift=-8mm] {\emph{waiting}};
  \node[above of=r1,node distance=7mm,xshift=8mm] {\emph{tight}};

\end{tikzpicture}
		 }
	 	 \caption{\centering}
	 	 \label{fig:fkv_example}
	\end{subfigure}
	\caption{
    (\subref{fig:example_aut})~$\autex$.
    (\subref{fig:example_dag})~The run DAG of~$\autex$ over~$b^\omega$.
    (\subref{fig:kv_example})~A~part of $\algkv(\autex)$.
    (\subref{fig:fkv_example})~$\algfkv(\autex)$; we highlight the
      \emph{waiting} and the \emph{tight} parts.}
    \vspace{-2mm}
\end{figure}

We assign ranks to vertices of run DAGs as follows:
Let $\dagwiof 0 = \dagw$ and~$j = 0$.
Repeat the following steps until the~fixpoint or for at most $2n + 1$ steps,
where $n = |Q|$.
% \figexampledag
\vspace{-0mm}
\begin{itemize}
  \setlength{\itemsep}{0mm}
  \item  Set $\rankwof{v} := j$ for all finite vertices $v$ of~$\dagwiof j$
    and let $\dagwiof{j+1}$ be $\dagwiof{j}$ minus the vertices with the
    rank~$j$.
  \item  Set $\rankwof{v} := j+1$ for all endangered vertices $v$
    of~$\dagwiof {j+1}$ and let $\dagwiof{j+2}$ be $\dagwiof{j+1}$ minus the
    vertices with the rank~$j+1$.
  \item  Set $j := j + 2$.
\end{itemize}
\vspace{-0mm}

\noindent
For all vertices~$v$ that have not been assigned a~rank yet, we assign
$\rankwof{v} := \omega$.
% (Note that since~$\aut$ is complete, then $\dagwiof 1 = \dagwiof 0$.)
See \cref{fig:example_aut} for an example BA~$\autex$ and
\cref{fig:example_dag} for the~run DAG of $\autex$ over~$b^\omega$.

%*******************************************************************************
\vspace{-0.0mm}
\subsection{Basic Rank-Based Complementation}
\vspace{-0.0mm}
%*******************************************************************************

The intuition in rank-based complementation algorithms is that states in
the complemented automaton~$\cut$ \emph{track} all runs of the original
automaton~$\aut$ on the given word and the possible \emph{ranks} of each of the
runs.
Loosely speaking, an accepting run of a~complement automaton~$\cut$ on
a~word~$\word \notin \langof \aut$
represents the run DAG of~$\aut$ over~$\word$ (in the complement, each state in
a~\emph{macrostate} is assigned a~\emph{rank})\footnote{%
This is not entirely true; there may be more accepting runs of~$\cut$
over~$\word$, with ranks assigned to states of~$\aut$ that are higher than the
ranks in the run DAG.
There will, however, be a~\emph{minimum} run of~$\cut$ that matches the run DAG
(in the terminology of \cref{sec:super-tight-runs}, such a~run corresponds to
a~\emph{super-tight run}).
}.%

The complementation procedure works with the notion of level rankings of states of~$\aut$, originally proposed
in~\cite{KupfermanV01,FriedgutKV06}.
For $n = |Q|$, a~\emph{(level) ranking} is a~function $f\colon Q \to \numsetof{2n}$ such that
$\{f(q_f) \mid q_f \in F\} \subseteq \{0, 2, \ldots, 2n\}$, i.e., $f$~assigns even
ranks to accepting states of~$\aut$.
We use $\cR$ to denote the set of all rankings and $\oddof f$ to denote the set
of states given an~odd ranking by~$f$, i.e.~$\oddof f = \{q \in Q \mid f(q) \text{ is odd}\}$.
For a~ranking~$f$, the \emph{rank} of~$f$ is defined as~$\rankof f = \max\{f(q)
\mid q \in Q\}$.
We use $f \leq f'$ iff for every state $q \in Q$ we have $f(q) \leq f'(q)$
and $f < f'$ iff $f \leq f'$ and there is a~state $p \in Q$ with $f(p) <
f'(p)$.

The simplest rank-based procedure, called \algkv, constructs the BA~$\algkv(\aut) = (Q',
\delta', I', F')$ whose components are defined as follows~\cite{KupfermanV01}:

\vspace{-2mm}
\begin{itemize}
  \setlength{\itemsep}{0mm}
  \item  $Q' = 2^Q \times 2^Q \times \cR$ is a~set of \emph{macrostates} denoted as $(S, O, f)$,
  \item  $I' = \{I\} \times \{\emptyset\} \times \cR$,
  \item  $(S', O', f') \in \delta'((S, O, f), a)$ iff
    \begin{itemize}
      \item  $S' = \delta(S, a)$,
      \item  for every $q \in S$ and $q' \in \delta(q, a)$ it holds that
        $f'(q') \leq f(q)$, and
      \item
        $O' = \left\{\begin{array}{l}
            \delta(S, a) \setminus \oddof{f'} \text{ if } O = \emptyset, \\
            \delta(O, a) \setminus \oddof{f'} \text{ otherwise, and}
          \end{array}\right.$
        % \begin{itemize}
        %   \item  $O' = \delta(S, a) \setminus \oddof{f'}$ if $O = \emptyset$ or
        %   \item  $O' = \delta(O, a) \setminus \oddof{f'}$ if $O \neq \emptyset$, and
        % \end{itemize}
    \end{itemize}
  \item  $F' = 2^Q \times \{\emptyset\} \times \cR$.
\end{itemize}
\vspace{-2mm}

%\figkvexample
The macrostates $(S,O,f)$ of~$\algkv(\aut)$ are composed of three components.
The $S$~component tracks all runs of~$\aut$ over the input word in the same way
as determinization of an NFA.
The $O$~component, on the other hand, tracks all runs whose rank has been even
since the last cut-point (a~point where $O = \emptyset$).
The last component, $f$, assigns every state in~$S$ a~rank.
Note that the $f$~component is responsible for the nondeterminism of the
complement (and also for the content of the $O$~component).
A~run of $\algkv(\aut)$ is accepting if it manages to empty the $O$~component of
states occurring on the run infinitely often.
We often merge~$S$ and~$f$ components and use, e.g., $(\{r{:}4, s{:}4\},
\emptyset)$ to denote the macrostate $(\{r,s\}, \emptyset, \{r \mapsto 4, s
\mapsto 4\})$ (we also omit ranks of states not in~$S$).
See \cref{fig:kv_example} for a~part of $\algkv(\autex)$ that starts in
$(\{r{:}4, s{:}4\}, \emptyset)$ and keeps ranks as high as possible (the whole
automaton is prohibitively large to be shown here---the implementation of \algkv
in GOAL~\cite{goal} outputs a~BA with 98 states).
Note that in order to accept the word
$b^\omega$, the accepting run needs to nondeterministically decrease the rank
of the successor of~$s$ (the transition $(\{s{:}4, t{:}4\}, \{s,t\}) \ltr b
(\{s{:}3, t{:}4\}, \{t\})$).

In the worst case, \algkv constructs a~BA with approximately~$(6n)^n$ states~\cite{KupfermanV01}.

%*******************************************************************************
\vspace{-0.0mm}
\subsection{Complementation with Tight Rankings}\label{sec:label}
\vspace{-0.0mm}
%*******************************************************************************

% \newcommand{\figfkvexample}[0]{
% \begin{wrapfigure}[13]{r}{6.2cm}
% \vspace*{-13mm}
% \hspace*{-2mm}
% \begin{minipage}{6.3cm}
%   \centering
% \input{figs/fkv-ex.tikz}
% \end{minipage}
% \vspace*{-11mm}
% \caption{$\algfkv(\autex)$; the \emph{waiting} and the \emph{tight} parts are highlighted.
%   The 4~wobbly transitions and the macrostate $(\{s{:}1\},\emptyset)$ will be
%   removed by \algdelay (\cref{sec:delay}).}
%   % \vh{why the wobbly transition within the tight part?}}
% \label{fig:fkv_example}
% \end{wrapfigure}
% }

%\figfkvexample
Friedgut, Kupferman, and Vardi observed in~\cite{FriedgutKV06} that the
$\algkv$~construction generates macrostates with many rankings that are not
strictly necessary in the loop part of the lasso for an accepting run on a~word.
Their optimization, denoted as $\algfkv$, is based on composing the complement
automaton from two parts:
the first part (called by us the \emph{waiting} part) just tracks all runs
of~$\aut$ over the input word (in a~similar manner as in a~determinized NFA)
and the second part (the \emph{tight} part) in addition tracks the rank of each run in
a~similar manner as the $\algkv$ construction, with the difference that the
rankings are \emph{tight}.

For a~set of states $S \subseteq Q$, we call~$f$ to be $S$-\emph{tight} if
\begin{inparaenum}[(i)]
  \item  it has an odd rank~$r$,
  \item  $\{f(s) \mid s \in S\} \supseteq \{1, 3, \ldots, r\}$, and
  \item  $\{f(q) \mid q \notin S\} = \{0\}$.
\end{inparaenum}
A~ranking is \emph{tight} if it is $Q$-tight; we use $\cT$ to denote the set of
all tight rankings.

The \algfkv procedure constructs the BA~$\algfkv(\aut) = (Q',
\delta', I', F')$ whose components are defined as follows:

\vspace{-2mm}
\begin{itemize}
  \setlength{\itemsep}{0mm}
  \item  $Q' = Q_1 \cup Q_2$ where
    \begin{itemize}
    \setlength{\itemsep}{0mm}
      \item  $Q_1 = 2^Q$ and
      \item  $Q_2 = \{(S,O,f) \in 2^Q \times 2^Q \times \cT\mid
          f \text{ is $S$-tight}, O \subseteq S\}$,
    \end{itemize}
  \item  $I' = \{I\}$,
  \item  $\delta' = \tau_1 \cup \tau_2 \cup \tau_3$ where
    \begin{itemize}
      \item  $\tau_1\colon Q_1 \times \Sigma \to 2^{Q_1}$ such that $\tau_1(S, a) =
        \{\delta(S,a)\}$,
      \item  $\tau_2\colon Q_1 \times \Sigma \to 2^{Q_2}$ such that $\tau_2(S, a) =
        \{(S', \emptyset, f) \in Q_2 \mid S' = \delta(S, a), f \text{ is } S'\text{-tight}\}$,
      \item  $\tau_3\colon Q_2 \times \Sigma \to 2^{Q_2}$ such that $(S', O', f') \in
        \tau_3((S, O, f), a)$ iff
          \begin{itemize}
            \item  $S' = \delta(S, a)$,
            \item  for every $q \in S$ and $q' \in \delta(q, a)$ it holds that
              $f'(q') \leq f(q)$,
            \item  $\rankof f = \rankof{f'}$, and
            \item  $O' = \left\{\begin{array}{l}
              \delta(S, a) \setminus \oddof{f'} \text{ if } O = \emptyset, \\
              \delta(O, a) \setminus \oddof{f'} \text{ otherwise, and}
            \end{array}\right.$
            % \item  and
            %   \begin{itemize}[$\circ$]
            %     \item  $O' = \delta(S, a) \setminus \oddof{f'}$ if $O = \emptyset$ or
            %     \item  $O' = \delta(O, a) \setminus \oddof{f'}$ if $O \neq \emptyset$, and
            %   \end{itemize}
          \end{itemize}
    \end{itemize}
  \item  $F' = \{\emptyset\} \cup ((2^Q \times \{\emptyset\} \times \cT) \cap Q_2)$.
\end{itemize}
\vspace{-2mm}

%\figfkvexample
\noindent
We call the part of $\algfkv(\aut)$ with the states in~$Q_1$ the
\emph{waiting} part and the part with the states in~$Q_2$ the \emph{tight}
part (an accepting run in $\algfkv(\aut)$ simulates the run DAG of $\aut$
over a~word~$w$ by \emph{waiting} in~$Q_1$ until it can
generate \emph{tight rankings} only; then it moves to~$Q_2$).
See \cref{fig:fkv_example} for $\algfkv(\autex)$.
In order to accept the word
$b^\omega$, the accepting run needs to nondeterministically move from the waiting to the tight part.
Note that $\algfkv(\autex)$ is significantly smaller than
$\algkv(\autex)$ (which had 98 states).
In the worst case, \algfkv constructs a~BA with~$\bigOof{(0.96n)^n}$ states~\cite{FriedgutKV06}.

%*******************************************************************************
\vspace{-0.0mm}
\subsection{Optimal Rank-Based Complementation}
\vspace{-0.0mm}
%*******************************************************************************

\begin{figure}[t]
	% \begin{subfigure}{0.37\textwidth}
  %   \vspace*{4mm}
 	   \centering
		 \scalebox{0.9}{
	 	 \begin{tikzpicture}[->,>=stealth',shorten >=0pt,auto,node distance=1.8cm,
                    scale=0.8,transform shape,initial text={}]
  \tikzstyle{every state}=[inner sep=3pt,minimum size=5pt,rectangle,rounded corners=1mm]
  \tikzstyle{empty}=[]
  \tikzstyle{initstate}=[fill=yellow!30]
  \tikzstyle{wobbly}=[decorate, decoration={snake,amplitude=.2mm,segment length=2mm,post length=1mm}]

  \node[state,initial,initstate] (rs) {$\{r,s\}$};
  \node[state, below of=rs] (r) {$\{r\}$};
  \node[state, left of=rs,yshift=10mm] (st) {$\{s,t\}$};
  \node[state, below of=r] (s) {$\{s\}$};
  \node[state, left of=r] (t) {$\{t\}$};
  \node[state, accepting, left of=s] (em) {$\emptyset$};

  \node[state, accepting, right of=rs, node distance=30mm,yshift=10mm] (r1) {$\big(\{s{:}1, t{:}0\}, \emptyset,0\big)$};
  \node[state, below of=r1,xshift=-3mm] (r2) {$\big(\{s{:}1, t{:}0\}, \{t\},0\big)$};
  \node[state, fill=gray!40,accepting,below of=r2,yshift=-4mm] (r3) {$\big(\{s{:}1\},\emptyset,0\big)$};

  % \node[state, white, fill=white, below of=r2, xshift=-8mm, node distance=17mm, minimum width=20mm, minimum height=10mm] (sink) {};

  \path (rs) edge  node[right] {$a$} (r)
        (rs) edge  node[above,pos=0.2] {$b$} (st)
        (r) edge  node[right] {$b$} (s)
        (r) edge[wobbly]  node[below,pos=0.2] {$b$} (r3)
        (s) edge  node[xshift=1mm] {$b$} (st)
        % (s) edge[bend right=50]  node {$b$} (r1)
        (s) edge  node[above] {$a$} (em)
        (st) edge  node[left] {$a$} (t)
        (t) edge  node[left] {$b$} (em)
        (r) edge [loop right] node[pos=0.2,above] (lab2) {$a$} (r)
        (t) edge [loop left] node {$a$} (t)
        (em) edge [loop left] node (lab1) {$a,b$} (em)
        (st) edge [loop left] node {$b$} (st)
        (rs) edge[wobbly]  node[pos=0.3] {$b$} (r1)
        (st) edge  node {$b$} (r1)
        (r1) edge [bend left]  node {$b$} (r2)
        (r2) edge [bend left]  node {$b$} (r1)
        (r3) edge [wobbly]  node[right] {$b$} (r2);

  \draw[wobbly] (s) .. controls +(50mm,-10mm) and +(10mm,-20mm) .. node[pos=0.1,above,yshift=-0.5mm] {$b$} (r1.south east);

  \begin{pgfonlayer}{background}
  \node[draw,dashed,rectangle,fill=red!15,draw=black!70,rounded corners=8pt,inner sep=3pt,fit=(r1) (r2) (r3)] (tight) {};
  \node[draw,dashed,rectangle,fill=blue!15,draw=black!70,rounded corners=8pt,inner sep=3pt,fit=(st) (rs) (r) (t) (em) (s) (lab1) (lab2)] (waiting) {};
  \end{pgfonlayer}

  \node[above of=st,node distance=6.5mm,xshift=-8mm] {\emph{waiting}};
  \node[above of=r1,node distance=7mm,xshift=8mm] {\emph{tight}};

\end{tikzpicture}
		 }
	%  	 \label{fig:schewe_example}
	% \end{subfigure}
\caption{
    $\algschewe(\autex)$.
      The optimization \algdelay (\cref{sec:delay}) will remove the 4~wobbly transitions and macrostate $(\{s{:}1\},\emptyset,0)$.}
\label{fig:schewe_example}
\end{figure}

An optimal algorithm whose space complexity matches the theoretical lower bound
$\bigOof{(0.76n)^n}$ was given by Schewe in~\cite[Section 3.1]{Schewe09}.
We denote this algorithm as \algschewe.
Apart from the optimization of \algfkv, in \algschewe,
macrostates of the tight part contain one
additional component, i.e., a~macrostate has the form $\sofi$, where the last
component $i \in \{0, 2, \ldots, 2n - 2\}$, for $n = |Q|$, denotes the rank of
states that are in~$O$.
Then, at a~cut-point (when $O$ is being reset), $O$~is not filled with all states
having an even rank, but only those whose rank is~$i$ (at every
cut-point, $i$ changes to $i+2$ modulo the rank of~$f$).

Formally,
$\algschewe(\aut) = (Q', \delta', I', F')$ is constructed as
follows:

% The \CompSchewe procedure constructs the BA~$\BSven = (Q', \delta', I', F')$
% s.t.~$\langof \BSven = \overline{\langof \aut}$.  The components of~$\BSven$ are
% defined as follows:
%
%
% \ol{OLD:}
% Further, we consider the rank-based complementation procedure KV of Kupferman
% and Vardi~\cite{KupfermanV01}
% producing the automaton~$\BKV$ (with the state-language $\langkvof{S, O, f})$
% using the notation from the paper of Schewe~\cite{Schewe09}.
% \ol{say that a function $f: Q \to \{0, \ldots, 2n\}$ that maps all accepting
% states to even numbers is called a~\emph{ranking}}
%
% The construction constructs for a~BA $\aut = (Q, \delta, I, F)$ the complement BA
% $\B = (Q', \delta', I', F')$ where
%
\vspace{-2mm}
\begin{itemize}
  \setlength{\itemsep}{0mm}
  \item  $Q' = Q_1 \cup Q_2$ where
    \vspace{-1mm}
    \begin{itemize}
    \setlength{\itemsep}{0mm}
      \item  $Q_1 = 2^Q$ and
      \item  $Q_2 = \hspace*{-1mm}
        \begin{array}[t]{ll}
          \{(S,O,f, i) \in \hspace*{0mm}& 2^Q \times 2^Q \times \cT \times \{0, 2, \ldots, 2n - 2\} \mid f \text{ is $S$-tight},
          O \subseteq S \cap f^{-1}(i)\},
        \end{array}$
    \end{itemize}
    \vspace{-1mm}
  \item  $I' = \{I\}$,
  \item  $\delta' = \delta_1 \cup \delta_2 \cup \delta_3$ where
    \vspace{-1mm}
    \begin{itemize}
  \setlength{\itemsep}{0mm}
      \item  $\delta_1: Q_1 \times \Sigma \to 2^{Q_1}$ such that $\delta_1(S, a) =
        \{\delta(S,a)\}$,
      \item $\delta_2 : Q_1 \times \Sigma \to 2^{Q_2}$ such that $\delta_2(S, a) =
        \{(S', \emptyset, f, 0) \mid S' = \delta(S, a), f \text{ is } S'\text{-tight}\}$, and
      \item $\delta_3: Q_2 \times \Sigma \to 2^{Q_2}$ such that $(S', O', f', i') \in
        \delta_3((S, O, f, i), a)$ iff
          \begin{itemize}
  \setlength{\itemsep}{0mm}
            \item  $S' = \delta(S, a)$,
            \item  for every $q \in S$ and $q' \in \delta(q, a)$ it holds that
              $f'(q') \leq f(q)$,
            \item  $\rankof f = \rankof{f'}$,
            \item  and~~
              \begin{minipage}[t]{10cm}
              \begin{itemize}[$\circ$]
                \item  $i' = (i+2) \mod (\rankof{f'} + 1)$ and $O' = f'^{-1}(i')$ if
                  $O = \emptyset$ or
                \item  $i' = i$ and $O' = \delta(O, a) \cap f'^{-1}(i)$ if $O \neq
                  \emptyset$, and
              \end{itemize}
              \end{minipage}
          \end{itemize}
    \end{itemize}
    \vspace{-1mm}
  \item  $F' = \{\emptyset\} \cup ((2^Q \times \{\emptyset\} \times \cT \times
    \omega) \cap Q_2)$.
\end{itemize}
\vspace{-2mm}

% The macrostates $(S,O,f,i) \in Q_2$ of~$\algschewe(\aut)$ are composed of four components.
% The $S$~component tracks all runs of~$\aut$ over the input word in the same way
% as determinization of an NFA.
% The $O$~component, on the other hand, tracks all runs whose rank has been~$i$
% since the last cut-point (a~point where $O = \emptyset$).
% The last component, $f$, assigns every state in~$S$ a~rank (intuitively, $f$
% corresponds to the ranks assigned in the run DAG to the states on the given
% level).
% \ol{added the parenthesis}
% Note that the $f$~component is responsible for the nondeterminism of the
% complement (and also for the content of the $O$~component).
% A~run of $\algschewe(\aut)$ is accepting if it manages to empty the $O$~component of
% states occurring on the run infinitely often.
% We often merge~$S$ and~$f$ components and use, e.g., $(\{r{:}4, s{:}4\},
% \emptyset, 2)$ to denote the macrostate $(\{r,s\}, \emptyset, \{r \mapsto 4, s
% \mapsto 4\}, 2)$ (we also omit ranks of states not in~$S$).
See \cref{fig:schewe_example} for $\algschewe(\autex)$.
In this case, $\algschewe(\autex)$ is isomorphic to
$\algfkv(\autex)$; note this is just due to the simplicity of the example,
which does not contain higher ranks.
In general, $\algschewe(\autex)$ can indeed be much smaller.
% $b^\omega$, the accepting run needs to nondeterministically decrease the rank
% of the successor of~$s$ (the transition $(\{s{:}4, t{:}4\}, \{s,t\}) \ltr b
% (\{s{:}3, t{:}4\}, \{t\})$).

% The correctness of the construction is then given by the following theorem.
%
\vspace{0mm}
\begin{theorem}
 \textrm{\bf (\cite[Corollary~3.3]{Schewe09})}
	Let $\but = \algschewe(\aut)$.
  Then $\langof \but = \overline{\langof \aut}$.
	% Let $\aut$ be a BA and $\but = \algschewe(\aut)$. Then $\langof \but =
	% \overline{\langof \aut}$.
\end{theorem}

In the following, we assume that $\algschewe(\aut)$ contains only the states
and transitions reachable from $I'$.
We use \algschewe as the base algorithm in the rest of the
paper.
% We use \algschewe as the basis for further optimizations in the rest of the
% paper.

% See \cref{fig:fkv_example} for $\algfkv(\autex)$.
% Note that $\algfkv(\autex)$ is significantly smaller than
% $\algkv(\autex)$ (which had 98 states).

% %%%%%%%%%%%%%%%%%%%%%%%%%%%%%%%%%%%%%%%%%%%%%%%%%%%%%%%%%%%%%%%%%%%%%%%%%%%%%%%%
% \vspace{-0.0mm}
% \subsection{Optimal Rank-Based Complementation}\label{sec:schewe}
% \vspace{-0.0mm}
% %*******************************************************************************
%
% An optimal algorithm whose space complexity matches the theoretical lower bound
% $\bigOof{(0.76n)^n}$ was given by Schewe in~\cite[Section 3.1]{Schewe09}.
% We denote this algorithm as \algschewe.
% The difference from \algfkv is that in \algschewe, macrostates contain one
% additional component, i.e., a~macrostate has the form $\sofi$, where the last
% component $i \in \{0, 2, \ldots, 2n - 2\}$, for $n = |Q|$, denotes the rank of
% states that are in~$O$.
% Then, at a~cut-point (when $O$ is being reset), $O$~is not filled with all states
% having an even rank (as in \algfkv), but only those whose rank is~$i$ (at every
% cut-point, $i$ changes to $i+2$ modulo the rank of~$f$).
%
% For a~pair of rankings~$f$ and~$f'$, a~set $S \subseteq Q$, and a~symbol~$a \in
% \Sigma$, we use $f' \rankleq f$ iff for every $q \in S$ and $q' \in \delta(q,
% a)$ it holds that $f'(q') \leq f(q)$.
% \ol{change $\rankleq$ to sth else --- there is only one use, so maybe remove it}
%

%*******************************************************************************
\vspace{-0.0mm}
\subsection{Super-Tight Runs}\label{sec:super-tight-runs}
\vspace{-0.0mm}
%*******************************************************************************

Let $\but = \algschewe(\aut)$.
Each accepting run of $\but$ on $\word \in
\langof \but$ is \emph{tight}, i.e., the rankings of macrostates it
traverses in~$Q_2$ are tight (this follows from the definition of~$Q_2$).
In this section, we show that there exists a~\emph{super-tight run} of
$\but$ on~$\word$, which is, intuitively, a~run that uses as little
ranks as possible.
% The concept of super-tight runs is essential for our optimizations
% in~\cref{sec:optimizations}.
Our optimizations in~\cref{sec:optimizations} are based on
preserving super-tight runs of~$\but$.
% The concept of super-tight runs is essential for our optimizations
% in~\cref{sec:optimizations}.

Let
$
\rho = S_0 \ldots S_m
\sofiof{m+1} \sofiof{m+2}\ldots
$
be an accepting run of $\but$ over a~word~$\word \in
\Sigma^{\omega}$.
Given a~macrostate $\sofiof k$ for $k > m$, we define its \emph{rank} as
$\rankof{\sofiof k} = \rankof{f_k}$.
Further, we define the \emph{rank of the run~$\rho$} as $\rankof \rho =
\min\{\rankof{\sofiof k} \mid k > m \}$.
% \begin{lemma}
% Let $\word \in \langof{\but}$ and $\rho$ be an accepting run of $\but$
% on~$\word$ s.t.
% $$
% \rho = S_0 \ldots S_m
% \sofiof{m+1} \sofiof{m+2}\ldots
% $$
% Then there exists $\ell > m$ s.t. for all $k \geq \ell$, we have
% $\rankof{\sofiof k} = \rankof \rho$.
% \end{lemma}
%
% \begin{proof}
% The lemma follows directly from the definition of $\delta_3$ in the \algschewe
% procedure (a~successor has the same maximal rank as a~predecessor).
% \qed
% \end{proof}
Let~$\dagw$ be the run DAG of~$\aut$ over~$\word$ and~$\rank_\word$ be the
ranking of vertices in~$\dagw$.
We say that the run~$\rho$ is \emph{super-tight} if
for all~$k > m$ and all $q \in S_k$, it holds that
$f_k(q) = \rankofin{q, k} \word$.
Intuitively, super-tight runs correspond to runs whose ranking faithfully copies
the ranks assigned in~$\dagw$ (from some position~$m$ corresponding to the transition
from the waiting to the tight part of~$\but$).

% We define the rank of a~word $\word \in \langof \but$ in~$\but$ as
% $\rankofin \word \but = \min\{\rankof \rho \mid \rho \text{ is an accepting run of }
% \but \text{ on } \word\}$.
% We say that the run~$\rho$ is \emph{super-tight} if it holds that
% $\rankof{\sofiof{m+1}} = \rankofin \word \but$, i.e., when the run moves from
% the waiting to the tight part of~$\but$, it moves to a~macrostate whose rank is
% the same as the rank of the whole run (so all macrostates on the run have the
% same rank).

\begin{restatable}{lemma}{lemSuperTight}
\label{lem:super-tight-run}
Let $\word \in \langof{\but}$. Then there is a~super-tight accepting run~$\rho$
of $\but$ on~$\word$.
\end{restatable}
\begin{proof}
	This follows directly from the definition of a super-tight run and the
	\algschewe construction.
\end{proof}

Let
$
\rho = S_0 \ldots S_m
\sofiof{m+1} \sofiof{m+2}\ldots
$
be a~run and consider a~macrostate $\sofiof k$ for $k > m$.
We call a~set
$C_k \subseteq S_k$ a~\emph{tight core of a~ranking~$f_k$} if
$f_k(C_k) = \{ 1, 3, \dots, \rankof{f_k} \}$ and
$\restrof{f_k}{C_k}$ is injective (i.e., every state in the tight core has
a~unique odd rank).
Moreover, $C_k$ is a~\emph{tight core of a~macrostate~$\sofiof k$} if it is
a~tight core of $f_k$. We say that an infinite sequence $\tau = C_{m+1} C_{m+2}
\ldots$ is a~\emph{trunk} of run~$\rho$ if for all $k > m$ it holds that $C_k$
is a~tight core of $\rho(k)$ and there is a bijection $\theta: C_k \to C_{k+1}$
s.t. if $\theta(q_k) = q_{k+1}$ then $q_{k+1} \in\delta(q_k, \alpha_k)$.
%\vh{previous: $C_{k+1} \subseteq \delta(C_k, \wordof k)$.}
%
% \ol{}
% and~$C_k$ is a~\emph{tight core of the run~$\rho$ at
% position~$k$} if~$C_k$ is a~tight core of $\sofiof k$ and there is
% a~subset~$C_{k+1} \subseteq \delta(C_k, \wordof k)$ s.t.\ $C_{k+1}$ is a~tight
% core of run~$\rho$ at position~$k+1$
% (note that the previous definition is co-inductive \ol{is it true?}).
We will, in particular, be interested in trunks of super-tight runs.
In these runs, a~trunk (there can be several)
\emph{represents runs of~$\aut$ that keep the super-tight ranks of~$\rho$}.
The following lemma shows that every state in any tight core in a~trunk of such
a~run has at least one successor with the same rank.

\begin{restatable}{lemma}{lemTrunks}\label{lem:trunks}
Let
$
\rho = S_0 \ldots S_m
\sofiof{m+1}\ldots
$
be an~accepting super-tight run of~$\but$ on~$\word$. Then there is a trunk
$\tau = C_{m+1} C_{m+2} \ldots$ of~$\rho$ and, moreover,
for every~$k > m$ and all states $q_k \in C_k$, it holds that there is
a~state $q_{k+1} \in C_{k+1}$ such that $f_k(q_k) = f_{k+1}(q_{k+1})$.
\end{restatable}
\begin{proof}
	First we show how to inductively construct a trunk $\tau = C_{m+1} C_{m+2}
	\ldots$
	\begin{inparaenum}[(i)]
		\item As the base case, let $C_{m+1}$ be an arbitrary tight core of
			$f_{m+1}$.
		\item As the inductive step, consider a~tight core~$C_k$ from the trunk and
      let us construct~$C_{k+1}$. Since $\rho$ is a super-tight run, for each $q \in C_k$
		 	there is a state $q' \in S_{k+1}$ s.t. $f_k(q) = f_{k+1}(q')$. This follows
		 	from the run DAG ranking procedure. We put $q' \in C_{k+1}$. Such a
      constructed set is a~tight core of~$f_{k+1}$ and from the construction we
      have the property stated in the lemma.
	\end{inparaenum}
\end{proof}

% \begin{proof}
% 	Consider a
% 	tight cores $C_k$ and $C_{k+1}$ from the trunk $\tau$. Further, let $\theta$
% 	be the bijection between $C_k$ and $C_{k+1}$ given from the definition of
% 	trunk. We show by induction that each $q_k\in C_k$ there is a state
% 	$q_{k+1}\in C_{k+1}$ s.t. $f_k(q_{k}) = f_{k+1}(q_{k+1})$.
% 	\begin{itemize}
% 		\item Base case: Assume $q \in C_k$ s.t. $f_k(q) = \rankof{f_k}$. Then since
% 		$\tau$ is a trunk, $C_{k+1}$ is tight rank, there is a state $\theta(q) = q'
% 		\in C_{k+1}\cap\delta(C_k, \alpha_k)$ with $f_{k+1}(q') = \rankof{f_k}$.
%
% 		\item Inductive case: Assume a state  $q \in C_k$ with the rank $r =
% 		f_{k}(q)$ s.t. the induction holds true for all $p\in C_k$ with $f_k(p) >
% 		r$. Again since $\tau$ is a trunk, $C_{k+1}$ is tight rank, there is the
% 		state $\theta(q) = q' \in C_{k+1}$ and $f_{k+1}(q') = r$. From the induction
% 		hypothesis we have $f_{k}(p) = f_{k+1}(\theta(p))$ for each $p \in C_k$ s.t.
% 		$f_k(p) > r$. Since $C_{k+1}$ is a tight core it must be the case when
% 		$f_{k+1}(\theta(p)) = r$.
% 	\end{itemize}
% \qed
% \end{proof}

%%%%%%%%%%%%%%%%%%%%%%%%%%%%%%%%%%%%%%%%%%%%%%%%%%%%%%%%%%%%%%%%%%%%%%%%%%%%%%%%
\vspace{-0.0mm}
\section{Optimized Complement Construction}\label{sec:optimizations}
\vspace{-0.0mm}
%%%%%%%%%%%%%%%%%%%%%%%%%%%%%%%%%%%%%%%%%%%%%%%%%%%%%%%%%%%%%%%%%%%%%%%%%%%%%%%%

In this section, we introduce our optimizations of $\algschewe$ that are key to
producing small complement automata in practice.

\newcommand{\delayalgverb}[0]{
% \begin{algorithm}[!h]
\begin{algorithm}[t]
 \SetKwInput{KwInput}{Input}
 \SetKwInput{KwOutput}{Output}
 \KwInput{A B\"{u}chi automaton $\aut = (Q,I,\delta,F)$}
 \KwOutput{A B\"{u}chi automaton $\cut$ s.t.~$\langof \cut = \overline{\langof
  \aut}$}
 $\stack\gets \{I\}$, $Q_1 \gets \{I\}$, $\theta_2 \gets \emptyset$,
 $(\ignore, \delta_1 \cup \delta_2 \cup \delta_3, I', F') \gets
  \algschewe(\aut)$\;
 \While{$\stack \neq\emptyset$}{
 		Take a waiting-part macrostate $R \subseteq Q$ from $\stack$\;
		\ForEach{$a\in\Sigma$}{
			\If{$\exists T \in \delta_1(R,a)$ s.t. $R \ltr{a} T$ closes a~cycle in $Q_1$}{
        $\theta_2 \gets \theta_2 \cup \{R \ltr{a} U \mid U \in \delta_2(R,a)\}$\;
			}
      \ForEach{$T \in \delta_1(R,a)$ s.t. $T \notin Q_1$}{
        $\stack \gets \stack \cup \{T\}$\;
        $Q_1 \gets \states \cup \{T\}$\;
      }
		}
 }
 $Q_2 \gets \reach_{\delta_3}(\imgof{\theta_2})$\;
 \Return{$\cut = (Q_1 \cup Q_2, \delta_1\cup \theta_2 \cup \delta_3, I', F' \cap
  Q_2)$}\;
 \caption{The \algdelay construction}
 \label{alg:constr}
\end{algorithm}
}

%*******************************************************************************
\vspace{-0.0mm}
\subsection{Delaying the Transition from Waiting to Tight}\label{sec:delay}
\vspace{-0.0mm}
%*******************************************************************************

Our first optimization of the construction of the complement automaton reduces the
number of nondeterministic transitions between the \emph{waiting} and the
\emph{tight} part.
This optimization is inspired by the idea of \emph{partial order reduction} in
model checking~\cite{Godefroid90,Valmari91,Peled93}.
In particular, since in each state of the waiting part, it is possible to move
to the tight part, we can arbitrarily delay such a~transition (but need to take
it eventually) and, therefore, significantly reduce the number of transitions
(and, as our experiments later show, also significantly reduce the number of
reachable states in~$Q_2$).

\delayalgverb
Speaking in the terms of partial order reduction, when constructing the
waiting part of the complement BA, given a~macrostate $S \in Q_1$ and
a~symbol~$a \in \Sigma$, we can set
$\theta_2 \subseteq \delta_2$ such that $\theta_2(S, a) := \emptyset$ if the
\emph{cycle closing condition} holds and $\theta_2(S, a) := \delta_2(S,a)$
otherwise.
Informally, the \emph{cycle closing condition} (often denoted as \textbf{C3})
holds for~$S$ and~$a$ if the successor of~$S$ over~$a$ in the waiting part does
not close a~cycle where the transition to the tight part would be infinitely
often delayed.
Practically, it means that when constructing~$Q_1$, we need to check whether
successors of a~macrostate close a cycle in the so-far generated part of~$Q_1$.
% We refer to the construction as \algdelay\ol{ (it is formally described in Appendix~\ref{sec:delay-appendix}).
% remove ref. to appendix?}.
We give the construction in \cref{alg:constr} and refer to it as \algdelay.
Using this optimisation on the example in \cref{fig:schewe_example}, we would
remove the $b$-transitions from $\{r,s\}$ and $\{s\}$ to the macrostate
$(\{s{:}1, t{:}0\}, \emptyset, 0)$ and also the macrostate $(\{s{:}1\},
\emptyset, 0)$ (including the transitions incident with it).

\begin{restatable}{lemma}{lemDelay}
	% Let $\aut$ be a BA and $\cut = \algdelay(\aut)$. Then $\langof \cut =
	% \overline{\langof \aut}$.
	Let $\aut$ be a BA.
  Then $\langof{\algdelay(\aut)} = \langof{\algschewe(\aut)}$.
  Moreover, for every accepting super-tight run of~$\algschewe(\aut)$ on~$\word$,
  there is an accepting super-tight run of~$\algdelay(\aut)$ on~$\word$.
\end{restatable}

\begin{proof}
  Showing $\langof{\algdelay(\aut)} \subseteq \langof{\algschewe(\aut)}$ is
  trivial.
  In order to show $\langof{\algdelay(\aut)} \supseteq \langof{\algschewe(\aut)}$,
	consider some $\word \in\langof{\algschewe(\aut)}$.
  Then, there is an accepting run $\rho_m =
	S_0\dots S_m \sofiof {m+1}\dots$ on $\word$ in~$\algschewe(\aut)$.
  For each $\ell > m$ there is, however, also an accepting run $\rho_{\ell} =
  S_0\dots S_\ell \sofiof {\ell+1}\dots$ on $\word$ in $\algschewe(\aut)$.
  Note that each $\rho_\ell$ differs from $\rho_m$ on the point where the run
  switched from the waiting part to the tight part of~$\algschewe(\aut)$.
  Therefore, since the run $\rho_m$ managed to empty $O$ infinitely often,
  $\rho_\ell$ will also be able to do so and, therefore, it will also be
  accepting.

	From the \algdelay construction, we have the fact that at least one
  macrostate $\sofiof{k+1}$ where
	$k>m$ is in~$Q_2$. If this were not true, there would be a~closed cycle with no state
	in~$Q_2$, which is a~contradiction. From the previous reasoning we have that
	the run $\rho_k = S_1\dots S_k (S_{k+1}, O_{k+1}, f_{k+1}, i'_{k+1})\dots$ on
  $\word$ is present in $\algdelay(\aut)$. Moreover, this run is accepting both
  in $\algschewe(\aut)$ and $\algdelay(\aut)$, which concludes the proof.
\end{proof}

Since \algdelay does not affect the rankings in the macrostates and only delays
the transition from the waiting to the tight part, we can freely use it as the
base algorithm instead of \algschewe in all following optimizations.

%*******************************************************************************
\vspace*{-0.0mm}
\subsection{Successor Rankings}
\vspace*{-0.0mm}
%*******************************************************************************

% \newcommand{\figsuccrankcoarse}[0]{
% \begin{figure}[t]
% \includegraphics[width=5cm,keepaspectratio]{figs/succ_rank_coarse_aut.png}
% \includegraphics[width=5cm,keepaspectratio]{figs/succ_rank_coarse_compl.png}
% \caption{\ol{SuccRank coarse}}
% \label{fig:succ_rank_coarse}
% \end{figure}
% }

\newcommand{\figsuccrankcoarse}[0]{
\begin{wrapfigure}[17]{r}{5.0cm}
\vspace*{-12mm}
\hspace*{0mm}
\begin{minipage}{5.0cm}
  \centering
  \begin{tikzpicture}[->,>=stealth',shorten >=0pt,auto,node distance=1.2cm,
                    scale=0.8,transform shape,initial text={}]
  \tikzstyle{every state}=[inner sep=3pt,minimum size=5pt]
  \tikzstyle{empty}=[]
  \tikzstyle{initstate}=[fill=yellow!30]

  \node[state,initial,initstate] (q) {$q$};
  \node[state,right of=q] (r) {$r$};
  \node[state,initial,initstate,right of=r, xshift=8mm] (s) {$s$};
  \node[state, right of=s] (t) {$t$};

  \path (q) edge  node {$a$} (r)
        (r) edge[loop right]  node[above,pos=0.3] {$a$} (r)
        (s) edge  node {$a$} (t);

\end{tikzpicture}

  \vspace{3mm}
  \begin{tikzpicture}[>=stealth',shorten >=0pt,auto,node distance=1.2cm,
                    scale=0.8,transform shape,initial text={}]
  \tikzstyle{every state}=[inner sep=3pt,minimum size=5pt,rectangle,rounded corners=1mm]
  \tikzstyle{empty}=[]
  \tikzstyle{initstate}=[fill=yellow!30]

  \node[state,initial,initstate] (qs) {$\{q,s\}$};
  \node[state, below of=qs] (rt) {$\{r,t\}$};
  \node[state, below of=rt] (r) {$\{r\}$};

  \node[state, accepting, above of=qs,xshift=4mm] (r1) {$\big(\{r{:}3, t{:}1\}, \emptyset, 0\big)$};
  \node[state, accepting, right of=r1, node distance=32mm] (r3) {$\big(\{r{:}1, t{:}3\}, \emptyset, 0\big)$};
  \node[state, accepting, below of=r3] (r2) {$\big(\{r{:}1, t{:}1\}, \emptyset, 0\big)$};
  \node[state, accepting, below of=r2] (r4) {$\big(\{r{:}1, t{:}0\}, \emptyset, 0\big)$};
  \node[state, accepting, below of=r4] (r5) {$\big(\{r{:}0, t{:}1\}, \emptyset, 0\big)$};
  \node[state, accepting, below of=r5] (r6) {$\big(\{r{:}1\}, \emptyset, 0\big)$};

  % \node[state, thick, red, cross out, right of=q, node distance=35mm, minimum size=8mm] (cr1) {};
  % \node[state, thick, red, cross out, below of=r1, minimum size=8mm] (cr2) {};

  \begin{pgfonlayer}{background}
  \draw (r1.south west) edge[line width=2mm, draw=red!40] (r1.north east);
  \draw (r1.north west) edge[line width=2mm, draw=red!40] (r1.south east);

  \draw (r3.south west) edge[line width=2mm, draw=red!40] (r3.north east);
  \draw (r3.north west) edge[line width=2mm, draw=red!40] (r3.south east);
  \end{pgfonlayer}

  \path[->] (qs) edge  node [left] {$a$} (rt)
        (rt) edge  node [left] {$a$} (r)
        (r) edge[loop left]  node[below,pos=0.2] {$a$} (r)
        (qs) edge  node[left] {$a$} (r1)
        (qs) edge [above]  node {$a$} (r2)
        (qs) edge[bend right=15]  node[above,pos=0.5] {$a$} (r4)
        (qs) edge[bend right=20]  node[above,pos=0.5] {$a$} (r5)
        (qs) edge  node[above,pos=0.3] {$a$} (r3)
        (rt) edge[bend right=10]  node [above] {$a$} (r6)
        (r) edge[bend right=10]  node [above] {$a$} (r6);

  % \begin{pgfonlayer}{background}
  %   \node[state, white, fill=white!20, below of=r2, node distance=14mm, xshift=-5mm, minimum width=35mm, minimum height=8mm] (sink) {};
  %   \path[dashed,gray] (q)  edge [bend right]  node {} (sink)
  %     (rs) edge [bend right] node {} (sink)
  %     (t) edge [bend right] node {} (sink)
  %     (r1) edge [bend left]  node {} (sink)
  %     (r2) edge  node {} (sink)
  %     (r3) edge [bend left]  node {} (sink);
  % \end{pgfonlayer}
  %

\end{tikzpicture}
\end{minipage}
\vspace*{-5mm}
\caption{Illustration of \succrankred reduction ($\condcoarse$), focusing on
  the transitions from the waiting to the tight part.}
\label{fig:succ_rank_coarse}
\end{wrapfigure}
}

\newcommand{\figsuccrankfine}[0]{
\begin{wrapfigure}[15]{r}{5.0cm}
\vspace*{-10mm}
\hspace*{0mm}
\begin{minipage}{5.0cm}
  \centering
  \begin{tikzpicture}[->,>=stealth',shorten >=0pt,auto,node distance=1.2cm,
                    scale=0.8,transform shape,initial text={}]
  \tikzstyle{every state}=[inner sep=3pt,minimum size=5pt]
  \tikzstyle{empty}=[]
  \tikzstyle{initstate}=[fill=yellow!30]

  \node[state,initial,initstate] (q) {$q$};
  \node[state,right of=q, yshift=6mm] (r) {$r$};
  \node[state,below of=r] (s) {$s$};
  \node[state,right of=r] (t) {$t$};

  \path (q) edge  node {$a$} (r)
        (q) edge  node {$a$} (s)
        (r) edge  node {$a$} (t)
        (s) edge [loop right]  node {$a$} (s)
        (r) edge[loop above]  node {$a$} (r);

\end{tikzpicture}
  \begin{tikzpicture}[>=stealth',shorten >=0pt,auto,node distance=1.2cm,
                    scale=0.8,transform shape,initial text={}]
  \tikzstyle{every state}=[inner sep=3pt,minimum size=5pt,rectangle,rounded corners=1mm]
  \tikzstyle{empty}=[]
  \tikzstyle{initstate}=[fill=yellow!30]

  \node[state,initial,initstate] (q) {$\{q\}$};
  \node[empty, right of=q] (qright) {};
  \node[state, below of=q] (rs) {$\{r,s\}$};
  \node[empty, below right of=rs,xshift=8mm,yshift=0mm] (rsright) {};
  \node[state, below of=rs] (rst) {$\{r,s,t\}$};
  \node[empty, right of=rst,node distance=16mm] (rstright) {};

  \node[state, accepting, right of=rs, node distance=35mm] (r1) {$\big(\{r{:}1, s{:}5, t{:}3\}, \emptyset, 0\big)$};
  \node[empty, below of=r1] (r1below) {};
  % \node[state, accepting, below of=r1] (r3) {$\big(\{r\mapsto 1, s\mapsto 3\}, \emptyset, 0\big)$};
  % \node[state, accepting, below of=r3] (r2) {$\big(\{r\mapsto 1, s\mapsto 1\}, \emptyset, 0\big)$};

  % \node[state, thick, red, cross out, right of=rs, node distance=35mm, minimum size=8mm] (cr1) {};
  %\node[state, thick, red, cross out, below of=r1, minimum size=8mm] (cr2) {};

  % \node[state, white, fill=white!20, below of=r1, node distance=16mm, minimum width=40mm, minimum height=12mm] (sink) {};

  \path[->] (q) edge  node [left] {$a$} (rs)
        (rs) edge  node [left] {$a$} (rst)
        (rst) edge[loop below]  node {$a$} (rst)
        (rs) edge  node {$a$} (r1);

  \begin{pgfonlayer}{background}
  \draw (r1.south west) edge[line width=2mm, draw=red!40] (r1.north east);
  \draw (r1.north west) edge[line width=2mm, draw=red!40] (r1.south east);
  \end{pgfonlayer}

  \begin{pgfonlayer}{background}
    \path[->,dashed,gray] (q)  edge node {$a$} (qright)
      (rs) edge node {$a$} (rsright)
      (rst) edge node {$a$} (rstright)
      (r1) edge  node {$a$} (r1below);
    % \path[dashed,gray] (q)  edge [bend right]  node {} (sink)
    %   (rs) edge [bend right] node {} (sink)
    %   (rst) edge [bend right] node {} (sink)
    %   (r1) edge  node {} (sink);
  \end{pgfonlayer}

\end{tikzpicture}
\end{minipage}
\vspace{-8mm}
\caption{Illustration of \succrankred reduction ($\condfine$), focusing on one particular macrostate.}
\label{fig:succ_rank_fine}
\end{wrapfigure}
}

\newcommand{\figsuccrank}[0]{
\begin{figure}[t]
  \hfill
	\begin{subfigure}[b]{0.3\textwidth}
		\scalebox{0.90}{
    \hspace*{-5mm}
	  \begin{tikzpicture}[->,>=stealth',shorten >=0pt,auto,node distance=1.2cm,
                    scale=0.8,transform shape,initial text={}]
  \tikzstyle{every state}=[inner sep=3pt,minimum size=5pt]
  \tikzstyle{empty}=[]
  \tikzstyle{initstate}=[fill=yellow!30]

  \node[state,initial,initstate] (q) {$q$};
  \node[state,right of=q] (r) {$r$};
  \node[state,initial,initstate,right of=r, xshift=8mm] (s) {$s$};
  \node[state, right of=s] (t) {$t$};

  \path (q) edge  node {$a$} (r)
        (r) edge[loop right]  node[above,pos=0.3] {$a$} (r)
        (s) edge  node {$a$} (t);

\end{tikzpicture}
		}

	  \vspace{7mm}
		\scalebox{0.90}{
    \hspace*{-5mm}
	  \begin{tikzpicture}[>=stealth',shorten >=0pt,auto,node distance=1.2cm,
                    scale=0.8,transform shape,initial text={}]
  \tikzstyle{every state}=[inner sep=3pt,minimum size=5pt,rectangle,rounded corners=1mm]
  \tikzstyle{empty}=[]
  \tikzstyle{initstate}=[fill=yellow!30]

  \node[state,initial,initstate] (qs) {$\{q,s\}$};
  \node[state, below of=qs] (rt) {$\{r,t\}$};
  \node[state, below of=rt] (r) {$\{r\}$};

  \node[state, accepting, above of=qs,xshift=4mm] (r1) {$\big(\{r{:}3, t{:}1\}, \emptyset, 0\big)$};
  \node[state, accepting, right of=r1, node distance=32mm] (r3) {$\big(\{r{:}1, t{:}3\}, \emptyset, 0\big)$};
  \node[state, accepting, below of=r3] (r2) {$\big(\{r{:}1, t{:}1\}, \emptyset, 0\big)$};
  \node[state, accepting, below of=r2] (r4) {$\big(\{r{:}1, t{:}0\}, \emptyset, 0\big)$};
  \node[state, accepting, below of=r4] (r5) {$\big(\{r{:}0, t{:}1\}, \emptyset, 0\big)$};
  \node[state, accepting, below of=r5] (r6) {$\big(\{r{:}1\}, \emptyset, 0\big)$};

  % \node[state, thick, red, cross out, right of=q, node distance=35mm, minimum size=8mm] (cr1) {};
  % \node[state, thick, red, cross out, below of=r1, minimum size=8mm] (cr2) {};

  \begin{pgfonlayer}{background}
  \draw (r1.south west) edge[line width=2mm, draw=red!40] (r1.north east);
  \draw (r1.north west) edge[line width=2mm, draw=red!40] (r1.south east);

  \draw (r3.south west) edge[line width=2mm, draw=red!40] (r3.north east);
  \draw (r3.north west) edge[line width=2mm, draw=red!40] (r3.south east);
  \end{pgfonlayer}

  \path[->] (qs) edge  node [left] {$a$} (rt)
        (rt) edge  node [left] {$a$} (r)
        (r) edge[loop left]  node[below,pos=0.2] {$a$} (r)
        (qs) edge  node[left] {$a$} (r1)
        (qs) edge [above]  node {$a$} (r2)
        (qs) edge[bend right=15]  node[above,pos=0.5] {$a$} (r4)
        (qs) edge[bend right=20]  node[above,pos=0.5] {$a$} (r5)
        (qs) edge  node[above,pos=0.3] {$a$} (r3)
        (rt) edge[bend right=10]  node [above] {$a$} (r6)
        (r) edge[bend right=10]  node [above] {$a$} (r6);

  % \begin{pgfonlayer}{background}
  %   \node[state, white, fill=white!20, below of=r2, node distance=14mm, xshift=-5mm, minimum width=35mm, minimum height=8mm] (sink) {};
  %   \path[dashed,gray] (q)  edge [bend right]  node {} (sink)
  %     (rs) edge [bend right] node {} (sink)
  %     (t) edge [bend right] node {} (sink)
  %     (r1) edge [bend left]  node {} (sink)
  %     (r2) edge  node {} (sink)
  %     (r3) edge [bend left]  node {} (sink);
  % \end{pgfonlayer}
  %

\end{tikzpicture}
		}
		\caption{\centering}
		\label{fig:succ_rank_coarse}
	\end{subfigure}
  \hfill
  \hfill
	\begin{subfigure}[b]{0.3\textwidth}
		\scalebox{0.90}{
    \hspace*{-5mm}
	  \begin{tikzpicture}[->,>=stealth',shorten >=0pt,auto,node distance=1.2cm,
                    scale=0.8,transform shape,initial text={}]
  \tikzstyle{every state}=[inner sep=3pt,minimum size=5pt]
  \tikzstyle{empty}=[]
  \tikzstyle{initstate}=[fill=yellow!30]

  \node[state,initial,initstate] (q) {$q$};
  \node[state,right of=q, yshift=6mm] (r) {$r$};
  \node[state,below of=r] (s) {$s$};
  \node[state,right of=r] (t) {$t$};

  \path (q) edge  node {$a$} (r)
        (q) edge  node {$a$} (s)
        (r) edge  node {$a$} (t)
        (s) edge [loop right]  node {$a$} (s)
        (r) edge[loop above]  node {$a$} (r);

\end{tikzpicture}
		}

	  \vspace{7mm}
		\scalebox{0.90}{
    \hspace*{-5mm}
	  \begin{tikzpicture}[>=stealth',shorten >=0pt,auto,node distance=1.2cm,
                    scale=0.8,transform shape,initial text={}]
  \tikzstyle{every state}=[inner sep=3pt,minimum size=5pt,rectangle,rounded corners=1mm]
  \tikzstyle{empty}=[]
  \tikzstyle{initstate}=[fill=yellow!30]

  \node[state,initial,initstate] (q) {$\{q\}$};
  \node[empty, right of=q] (qright) {};
  \node[state, below of=q] (rs) {$\{r,s\}$};
  \node[empty, below right of=rs,xshift=8mm,yshift=0mm] (rsright) {};
  \node[state, below of=rs] (rst) {$\{r,s,t\}$};
  \node[empty, right of=rst,node distance=16mm] (rstright) {};

  \node[state, accepting, right of=rs, node distance=35mm] (r1) {$\big(\{r{:}1, s{:}5, t{:}3\}, \emptyset, 0\big)$};
  \node[empty, below of=r1] (r1below) {};
  % \node[state, accepting, below of=r1] (r3) {$\big(\{r\mapsto 1, s\mapsto 3\}, \emptyset, 0\big)$};
  % \node[state, accepting, below of=r3] (r2) {$\big(\{r\mapsto 1, s\mapsto 1\}, \emptyset, 0\big)$};

  % \node[state, thick, red, cross out, right of=rs, node distance=35mm, minimum size=8mm] (cr1) {};
  %\node[state, thick, red, cross out, below of=r1, minimum size=8mm] (cr2) {};

  % \node[state, white, fill=white!20, below of=r1, node distance=16mm, minimum width=40mm, minimum height=12mm] (sink) {};

  \path[->] (q) edge  node [left] {$a$} (rs)
        (rs) edge  node [left] {$a$} (rst)
        (rst) edge[loop below]  node {$a$} (rst)
        (rs) edge  node {$a$} (r1);

  \begin{pgfonlayer}{background}
  \draw (r1.south west) edge[line width=2mm, draw=red!40] (r1.north east);
  \draw (r1.north west) edge[line width=2mm, draw=red!40] (r1.south east);
  \end{pgfonlayer}

  \begin{pgfonlayer}{background}
    \path[->,dashed,gray] (q)  edge node {$a$} (qright)
      (rs) edge node {$a$} (rsright)
      (rst) edge node {$a$} (rstright)
      (r1) edge  node {$a$} (r1below);
    % \path[dashed,gray] (q)  edge [bend right]  node {} (sink)
    %   (rs) edge [bend right] node {} (sink)
    %   (rst) edge [bend right] node {} (sink)
    %   (r1) edge  node {} (sink);
  \end{pgfonlayer}

\end{tikzpicture}
		}
		\caption{\centering}
		\label{fig:succ_rank_fine}
	\end{subfigure}
  \hspace*{1cm}
\caption{
  (\subref{fig:succ_rank_coarse}) Illustration of \succrankred reduction
  ($\condcoarse$), focusing on the transitions from the waiting to the tight
  part.
  (\subref{fig:succ_rank_fine}) Illustration of \succrankred reduction
  ($\condfine$), focusing on one particular macrostate.
}
\end{figure}
}

\begin{figure}[t]
  \hfill
	\begin{subfigure}[b]{0.3\textwidth}
		\scalebox{0.90}{
    \hspace*{-5mm}
	  \begin{tikzpicture}[->,>=stealth',shorten >=0pt,auto,node distance=1.2cm,
                    scale=0.8,transform shape,initial text={}]
  \tikzstyle{every state}=[inner sep=3pt,minimum size=5pt]
  \tikzstyle{empty}=[]
  \tikzstyle{initstate}=[fill=yellow!30]

  \node[state,initial,initstate] (q) {$q$};
  \node[state,right of=q] (r) {$r$};
  \node[state,initial,initstate,right of=r, xshift=8mm] (s) {$s$};
  \node[state, right of=s] (t) {$t$};

  \path (q) edge  node {$a$} (r)
        (r) edge[loop right]  node[above,pos=0.3] {$a$} (r)
        (s) edge  node {$a$} (t);

\end{tikzpicture}
		}

	  \vspace{7mm}
		\scalebox{0.90}{
    \hspace*{-5mm}
	  \begin{tikzpicture}[>=stealth',shorten >=0pt,auto,node distance=1.2cm,
                    scale=0.8,transform shape,initial text={}]
  \tikzstyle{every state}=[inner sep=3pt,minimum size=5pt,rectangle,rounded corners=1mm]
  \tikzstyle{empty}=[]
  \tikzstyle{initstate}=[fill=yellow!30]

  \node[state,initial,initstate] (qs) {$\{q,s\}$};
  \node[state, below of=qs] (rt) {$\{r,t\}$};
  \node[state, below of=rt] (r) {$\{r\}$};

  \node[state, accepting, above of=qs,xshift=4mm] (r1) {$\big(\{r{:}3, t{:}1\}, \emptyset, 0\big)$};
  \node[state, accepting, right of=r1, node distance=32mm] (r3) {$\big(\{r{:}1, t{:}3\}, \emptyset, 0\big)$};
  \node[state, accepting, below of=r3] (r2) {$\big(\{r{:}1, t{:}1\}, \emptyset, 0\big)$};
  \node[state, accepting, below of=r2] (r4) {$\big(\{r{:}1, t{:}0\}, \emptyset, 0\big)$};
  \node[state, accepting, below of=r4] (r5) {$\big(\{r{:}0, t{:}1\}, \emptyset, 0\big)$};
  \node[state, accepting, below of=r5] (r6) {$\big(\{r{:}1\}, \emptyset, 0\big)$};

  % \node[state, thick, red, cross out, right of=q, node distance=35mm, minimum size=8mm] (cr1) {};
  % \node[state, thick, red, cross out, below of=r1, minimum size=8mm] (cr2) {};

  \begin{pgfonlayer}{background}
  \draw (r1.south west) edge[line width=2mm, draw=red!40] (r1.north east);
  \draw (r1.north west) edge[line width=2mm, draw=red!40] (r1.south east);

  \draw (r3.south west) edge[line width=2mm, draw=red!40] (r3.north east);
  \draw (r3.north west) edge[line width=2mm, draw=red!40] (r3.south east);
  \end{pgfonlayer}

  \path[->] (qs) edge  node [left] {$a$} (rt)
        (rt) edge  node [left] {$a$} (r)
        (r) edge[loop left]  node[below,pos=0.2] {$a$} (r)
        (qs) edge  node[left] {$a$} (r1)
        (qs) edge [above]  node {$a$} (r2)
        (qs) edge[bend right=15]  node[above,pos=0.5] {$a$} (r4)
        (qs) edge[bend right=20]  node[above,pos=0.5] {$a$} (r5)
        (qs) edge  node[above,pos=0.3] {$a$} (r3)
        (rt) edge[bend right=10]  node [above] {$a$} (r6)
        (r) edge[bend right=10]  node [above] {$a$} (r6);

  % \begin{pgfonlayer}{background}
  %   \node[state, white, fill=white!20, below of=r2, node distance=14mm, xshift=-5mm, minimum width=35mm, minimum height=8mm] (sink) {};
  %   \path[dashed,gray] (q)  edge [bend right]  node {} (sink)
  %     (rs) edge [bend right] node {} (sink)
  %     (t) edge [bend right] node {} (sink)
  %     (r1) edge [bend left]  node {} (sink)
  %     (r2) edge  node {} (sink)
  %     (r3) edge [bend left]  node {} (sink);
  % \end{pgfonlayer}
  %

\end{tikzpicture}
		}
		\caption{\centering}
		\label{fig:succ_rank_coarse}
	\end{subfigure}
  \hfill
  \hfill
	\begin{subfigure}[b]{0.3\textwidth}
		\scalebox{0.90}{
    \hspace*{-5mm}
	  \begin{tikzpicture}[->,>=stealth',shorten >=0pt,auto,node distance=1.2cm,
                    scale=0.8,transform shape,initial text={}]
  \tikzstyle{every state}=[inner sep=3pt,minimum size=5pt]
  \tikzstyle{empty}=[]
  \tikzstyle{initstate}=[fill=yellow!30]

  \node[state,initial,initstate] (q) {$q$};
  \node[state,right of=q, yshift=6mm] (r) {$r$};
  \node[state,below of=r] (s) {$s$};
  \node[state,right of=r] (t) {$t$};

  \path (q) edge  node {$a$} (r)
        (q) edge  node {$a$} (s)
        (r) edge  node {$a$} (t)
        (s) edge [loop right]  node {$a$} (s)
        (r) edge[loop above]  node {$a$} (r);

\end{tikzpicture}
		}

	  \vspace{7mm}
		\scalebox{0.90}{
    \hspace*{-5mm}
	  \begin{tikzpicture}[>=stealth',shorten >=0pt,auto,node distance=1.2cm,
                    scale=0.8,transform shape,initial text={}]
  \tikzstyle{every state}=[inner sep=3pt,minimum size=5pt,rectangle,rounded corners=1mm]
  \tikzstyle{empty}=[]
  \tikzstyle{initstate}=[fill=yellow!30]

  \node[state,initial,initstate] (q) {$\{q\}$};
  \node[empty, right of=q] (qright) {};
  \node[state, below of=q] (rs) {$\{r,s\}$};
  \node[empty, below right of=rs,xshift=8mm,yshift=0mm] (rsright) {};
  \node[state, below of=rs] (rst) {$\{r,s,t\}$};
  \node[empty, right of=rst,node distance=16mm] (rstright) {};

  \node[state, accepting, right of=rs, node distance=35mm] (r1) {$\big(\{r{:}1, s{:}5, t{:}3\}, \emptyset, 0\big)$};
  \node[empty, below of=r1] (r1below) {};
  % \node[state, accepting, below of=r1] (r3) {$\big(\{r\mapsto 1, s\mapsto 3\}, \emptyset, 0\big)$};
  % \node[state, accepting, below of=r3] (r2) {$\big(\{r\mapsto 1, s\mapsto 1\}, \emptyset, 0\big)$};

  % \node[state, thick, red, cross out, right of=rs, node distance=35mm, minimum size=8mm] (cr1) {};
  %\node[state, thick, red, cross out, below of=r1, minimum size=8mm] (cr2) {};

  % \node[state, white, fill=white!20, below of=r1, node distance=16mm, minimum width=40mm, minimum height=12mm] (sink) {};

  \path[->] (q) edge  node [left] {$a$} (rs)
        (rs) edge  node [left] {$a$} (rst)
        (rst) edge[loop below]  node {$a$} (rst)
        (rs) edge  node {$a$} (r1);

  \begin{pgfonlayer}{background}
  \draw (r1.south west) edge[line width=2mm, draw=red!40] (r1.north east);
  \draw (r1.north west) edge[line width=2mm, draw=red!40] (r1.south east);
  \end{pgfonlayer}

  \begin{pgfonlayer}{background}
    \path[->,dashed,gray] (q)  edge node {$a$} (qright)
      (rs) edge node {$a$} (rsright)
      (rst) edge node {$a$} (rstright)
      (r1) edge  node {$a$} (r1below);
    % \path[dashed,gray] (q)  edge [bend right]  node {} (sink)
    %   (rs) edge [bend right] node {} (sink)
    %   (rst) edge [bend right] node {} (sink)
    %   (r1) edge  node {} (sink);
  \end{pgfonlayer}

\end{tikzpicture}
		}
		\caption{\centering}
		\label{fig:succ_rank_fine}
	\end{subfigure}
  \hspace*{1cm}
\caption{
  (\subref{fig:succ_rank_coarse}) Illustration of \succrankred reduction
  ($\condcoarse$), focusing on the transitions from the waiting to the tight
  part.
  (\subref{fig:succ_rank_fine}) Illustration of \succrankred reduction
  ($\condfine$), focusing on one particular macrostate.
}
\end{figure}

Our next optimization is used to reduce the maximum considered ranking of
a~macrostate in the tight part of~$\but = \algschewe(\aut)$.
For a~given macrostate, the number of tight rankings that can occur
within the macrostate rises combinatorially with the macrostate's maximum rank
(in particular, the number of tight rankings for a~given set of states
corresponds to the Stirling number of the second kind of the maximum
rank~\cite{FriedgutKV06}).
It is hence desirable to reduce the maximum considered rank as much as
possible.

%\figsuccrankcoarse
The idea of our optimization called \succrankred is the following.
Suppose we have a~macrostate $\sofi$ from the tight part of~$\but$.
Further, assume that the maximum number of non-accepting states in the $S$-component of a~macrostate that
is infinitely often reachable from $\sofi$ is~$\maxinfreachof S$.
Then, we know that a~super-tight accepting run that goes through~$\sofi$ will
never need a~rank higher than $2 \maxinfreachof S - 1$ (any accepting state will
be assigned an even rank, so we can omit them).
Therefore, if the rank of~$f$ is higher than $2 \maxinfreachof S - 1$, we can
safely discard~$\sofi$ (since there will be a~super-tight accepting run that
goes over $(S,O',f',i')$ with $f' < f$).
This part of the optimization is called \emph{coarse}.

Moreover, let~$q \in S$ and let $\mininfreachof{\{q\}}$ be the smallest size of
a~set of states (again without accepting states) reachable from~$q$ over some (infinite) word infinitely often.
Then, we know that those states will have a~rank bounded by the rank of~$f(q)$,
so there are only (at most) $\maxinfreachof S - \mininfreachof{\{q\}}$
states whose rank can be higher than~$f(q)$.
Therefore, the rank of~$f$, which
is tight, can be at most $f(q) + 2(\maxinfreachof S - \mininfreachof{\{q\}})$.
We call this part of the optimization \emph{fine}.

%!!!!!!!!!!!!!!!!!!!!
% FORCED NEW PAGE
%!!!!!!!!!!!!!!!!!!!!
%\newpage

We now formalize the intuition.
Let us fix a~BA $\aut = (Q, \delta, I, F)$.
Then, let us consider a~BA $R_\aut = (2^Q, \delta_R, \emptyset, \emptyset)$,
with $\delta_R = \{R \ltr{a} S \mid S = \delta(R, a)\}$, which is tracking
\emph{reachability} between set of all states of~$\aut$ (we only focus on its
structure and not the language).
Note that $R_\aut$ is deterministic and complete.
Further, given $S \subseteq Q$, let us use $\sccof S \subseteq
2^{2^{Q}}$ to denote the set of all \emph{strongly connected components
reachable from~$S$ in~$R_\aut$}.
We will use $\infreachof S$ to denote the set of states~$\bigcup \sccof S$,
i.e., the set of states such that there is an infinite path in~$R_\aut$
starting in~$S$ that passes through a~given state infinitely many times.
%
% Then, let $\graphof \aut = \langle V_\aut,
% E_\aut \rangle$ be a~directed graph where the set of vertices is given as
% $V_\aut = 2^Q$ and the set of edges as $E_\aut = \{(S, S') \mid S \subseteq Q
% \land \exists a\in \Sigma: S' = \delta(S, a)\}$.
% Further, given $S \in V_\aut$, let us use $\sccof S \subseteq
% 2^{V_\aut}$ to denote the set of all \emph{strongly connected components
% reachable from~$S$ in~$\graphof \aut$} (over edges in~$E_\aut$).
% We will use $\infreachof S$ to denote the set of vertices~$\bigcup \sccof S$,
% i.e., the set of vertices such that there is an infinite path in~$\graphof
% \aut$ starting in~$S$ that passes through a~given vertex infinitely many times.
% For~$S \subseteq Q$,
We define the maximum and minimum sizes of macrostates
reachable infinitely often from~$S$:
%
% \vspace{-2mm}
\begin{align*}
  \hspace*{-3mm}\maxinfreachof S & {}= \max\{ |R\setminus F| : R \in \infreachof S\}\qquad\text{and}\qquad
  \mininfreachof S {}= \min\{ |R\setminus F| : R \in \infreachof S\}.
\end{align*}%\\[-8mm]

\noindent
For a~macrostate $\sofi$, we define
%
%\vspace{-1mm}
$
\condcoarse(\sofi) \defiff \rankof f \leq 2\maxinfreachof S - 1.
$
If $\sofi$ does not satisfy $\condcoarse$, we can omit it from
the output of~$\algschewe(\aut)$ (as allowed by
\cref{lem:suckrank}).
See \cref{fig:succ_rank_coarse} for an example of such a~macrostates.
For instance, macrostate $(\{r{:}3, t{:}1\}, \emptyset, 0)$ can be removed since
its rank is 3 and $\maxinfreachof{\{r,t\}} = 1$, so $3 \not\leq 2
\maxinfreachof{\{r,t\}}-1$.

Moreover, we also define the condition
%
%\vspace{-1mm}
\begin{equation}
  \condfine(\sofi) \defiff
  \rankof f \leq \min\{f(q) + 2(\maxinfreachof S - \mininfreachof{\{q\}}) \mid
  {q \in S} \}.
\end{equation}

%\figsuccrankfine
\noindent
% Again, if $\sofi$ does not satisfy~$\condfine$, it can be omitted.
Again, we can omit $\sofi$ if it does not satisfy~$\condfine$.
See \cref{fig:succ_rank_fine} for an example of such a~macrostate.
Note that the rank of $(\{r{:}1,s{:}5,t{:}3\}, \emptyset, 0)$
is 5,
$\maxinfreachof{\{r,s,t\}} = 3$ and
$\mininfreachof{\{r\}} = 2$,
$\mininfreachof{\{s\}} = 1$,
$\mininfreachof{\{t\}} = 0$.
Then, $\min\{f(r) + 2(3-2), f(s) + 2(3-1), f(t) + 2(3-0)\} = \min\{1+2,5+4,
3+6\} = 3$, so the macrostate does not satisfy~$\condfine$
and can be removed.

We emphasize that~$\condcoarse$ and~$\condfine$ are incomparable.
For example, the macrostates removed due to~$\condcoarse$ in
\cref{fig:succ_rank_coarse} satisfy~$\condfine$ (since, e.g., $3 \leq \min\{3 +
2(1-1), 1+2(1-0)\}$) and the macrostate removed due to~$\condfine$ in
\cref{fig:succ_rank_fine} satisfies~$\condcoarse$ (since $5 \leq 2 \cdot 3 - 1$).

Putting the conditions together, we define the predicate
%
%\vspace{-2mm}
\begin{equation}\label{eq:succrankred}
  \succrankred(\sofi) \defiff \condcoarse(\sofi) \land \condfine(\sofi) .
\end{equation}\\[-9mm]

\noindent
We abuse notation and use $\succrankred(\aut)$ to denote the output of
$\algschewe(\aut) = (Q', \delta', I', F')$ where the states from the tight part
of~$Q'$ are restricted to those that satisfy $\succrankred$.

\begin{restatable}{lemma}{lemSuccRank}\label{lem:suckrank}
	Let $\aut$ be a BA. Then $\langof{\succrankred(\aut)} = \langof{\algschewe(\aut)}$.
\end{restatable}
\begin{proof}
	The inclusion $\langof{\succrankred(\aut)} \subseteq
	\langof{\algschewe(\aut)}$ is clear. Now we look at the other direction.
	Consider some $\word \in\langof{\algschewe(\aut)}$. Then, there is an
	accepting super-tight run $\rho = S_0\dots S_m \sofiof{m+1}\dots$
	of $\algschewe(\aut)$ over~$\word$.
  Consider $k > m$ and a~macrostate $\sofiof k$.
  The maximum rank of this macrostate is bounded by $2\maxinfreachof{S_k} - 1$
  because~$\maxinfreachof{S_k}$ is the largest size of the $S$-component (without final states) of
  a~macrostate reachable from~$S_k$
  and, therefore, removing macrostates that do not satisfy~$\condcoarse$
  from~$\algschewe(\aut)$ will not affect this run.

  Next, we prove the correctness of removing states from
  $\algschewe(\aut)$ using~$\condfine$.
  Consider a~set of states~$T \subseteq Q$; we will use $\rho_T$ to denote the
  run $\rho_T = T_0 T_1 T_2 \ldots$ of~$R_\aut$ from~$T$ ($= T_0$) over the
  word~$\wordof{k:\omega}$.
  Since~$R_\aut$ is deterministic and complete, there is exactly one such run.
  Given a~state~$q \in S_k$,
  let~$a$ be the smallest size of a~set of states (without final states) that occurs
  infinitely often in~$\rho_{\{q\}}$ and~$b$ be the largest size of a~set of
  states that occurs infinitely often in~$\rho_{S_k}$ (again without final states).
  From the definition of~$\mininfreachof{\cdot }$ and~$\maxinfreachof{\cdot }$, it
  holds that
  \begin{equation}
    \textstyle
    \mininfreachof{\{q\}} ~\leq~
    a ~\leq~
    b ~\leq~
    \maxinfreachof{S_k}.
  \end{equation}
  Since we can reach~$a$ different states from~$q$, the ranks of these states
  need to be less or equal to~$f_k(q)$ (no successor of $q$ can be given a~rank
  higher than~$f_k(q)$).
  Further, since we can reach at most~$b$ different states from~$S_k$, there
  are at some infinitely often occurring macrostate of~$\rho$ in the worst case
  only $b-a$ states that can have and odd rank greater than~$f_k(q)$.
  Due to the tightness of all macrostates in the tight part of~$\rho$, we can
  conclude that the maximal rank of $f_k$ can be bounded by
  $f_k(q) + 2(\maxinfreachof{S_k} - \mininfreachof{\{q\}})$.
  Therefore, a~macrostate where $\condfine$ does not hold will not be in
  a~super-tight run, so removing those macrostates does not affect the language
  of~$\algschewe(\aut)$.\qedhere
\end{proof}

%*******************************************************************************
\vspace{-0.0mm}
\subsection{Rank Simulation}\label{sec:ranksim}
\vspace{-0.0mm}
%*******************************************************************************

\newcommand{\figranksim}[0]{
\begin{figure}[t]
\begin{subfigure}[b]{0.6\textwidth}
  \begin{minipage}{4cm}
    \vspace*{-10mm}
		\scalebox{0.90}{
	  \begin{tikzpicture}[->,>=stealth',shorten >=0pt,auto,node distance=1.2cm,
                    scale=0.8,transform shape,initial text={}]
  \tikzstyle{every state}=[inner sep=3pt,minimum size=5pt]
  \tikzstyle{empty}=[]
  \tikzstyle{initstate}=[fill=yellow!30]

  \node[state,initial,initstate] (qini) {$q_0$};
  \node[state,right of=qini] (q1) {$q_1$};
  \node[state,right of=q1, yshift=6mm] (q2) {$q_2$};
  \node[state,accepting,below of=q2] (q3) {$q_3$};
  \node[state,accepting,right of=q2] (q4) {$q_4$};

  \path (q1) edge  node {$a$} (q2)
        (qini) edge  node {$a$} (q1)
        (q1) edge  node {$a$} (q3)
        (q2) edge  node {$a$} (q4)
        (q4) edge [loop right]  node {$a$} (q4)
        (q3) edge[loop right]  node {$a$} (q3);

  \node[state,initial,initstate,below of=qini, node distance=1.8cm] (rini) {$r_0$};
  \node[state, right of=rini] (r1) {$r_1$};
  \node[state,right of=r1] (r2) {$r_2$};
  \node[state,accepting,right of=r2] (r3) {$r_3$};

  \path (r1) edge  node {$a$} (r2)
        (rini) edge  node {$a$} (r1)
        (r2) edge  node {$a$} (r3)
        (r3) edge [loop right]  node {$a$} (r3);

\end{tikzpicture}}
	  % \vspace{7mm}
  \end{minipage}
  \begin{minipage}{4cm}
		\scalebox{0.90}{
    % \hspace*{-5mm}
	  \begin{tikzpicture}[>=stealth',shorten >=0pt,auto,node distance=1.2cm,
                    scale=0.8,transform shape,initial text={}]
  \tikzstyle{every state}=[inner sep=3pt,minimum size=5pt,rectangle,rounded corners=1mm]
  \tikzstyle{empty}=[]
  \tikzstyle{initstate}=[fill=yellow!30]

  \node[state,initial,initstate] (q0r0) {$\{q_0, r_0\}$};
  \node[state,below of=q0r0] (q1r1) {$\{q_1, r_1\}$};
  \node[state, below of=q1r1] (q2q3r2) {$\{q_2,q_3,r_2\}$};
  \node[state, below of=q2q3r2] (q3q4r3) {$\{q_3,q_4,r_3\}$};

  \node[state, accepting, right of=q0r0, node distance=35mm] (m1) {$\big(\{q_1{:}1, r_1{:}3\}, \emptyset, 0\big)$};
  \node[state, accepting, below of=m1] (m2) {$\big(\{q_1{:}3, r_1{:}1\}, \emptyset, 0\big)$};
  \node[state, accepting, below of=m2] (m3) {$\big(\{q_1{:}1, r_1{:}1\}, \emptyset, 0\big)$};
  \node[empty,minimum size=30pt,xshift=-5.5mm] at(q3q4r3.east) (r345inv) {};

  \node[state, white, fill=white!20, below of=m3, node distance=16mm, minimum width=40mm, minimum height=12mm] (sink) {};

  \path[->] (q1r1) edge  node [left] {$a$} (q2q3r2)
        (q0r0) edge  node [left] {$a$} (q1r1)
        (q2q3r2) edge  node [left] {$a$} (q3q4r3)
        (q0r0) edge [] node {$a$} (m1)
        (q0r0) edge [bend right=10] node {$a$} (m2)
        (q0r0) edge [bend right=16] node[yshift=-1mm] {$a$} (m3)
        (r345inv) edge[loop right]  node {$a$} (r345inv);

  % \begin{pgfonlayer}{background}
  %   \path[dashed,gray] (q1r1)  edge [bend right]  node {} (sink)
  %     (q2q3r2) edge [bend right] node {} (sink)
  %     (q3q4r3) edge [bend right] node {} (sink)
  %     (m1) edge  node {} (sink)
  %     (q0r0) edge  node {} (sink)
  %     (q1r1) edge  node {} (sink);
  % \end{pgfonlayer}

  \begin{pgfonlayer}{background}
    \draw (m1.south west) edge[line width=2mm, draw=red!40] (m1.north east);
    \draw (m1.north west) edge[line width=2mm, draw=red!40] (m1.south east);

    \draw (m2.south west) edge[line width=2mm, draw=red!40] (m2.north east);
    \draw (m2.north west) edge[line width=2mm, draw=red!40] (m2.south east);
  \end{pgfonlayer}

\end{tikzpicture}
		}
  \end{minipage}
	\vspace*{-3mm}
	\caption{\centering}
	\label{fig:ranksim}
\end{subfigure}
\hfill
\begin{subfigure}[b]{0.2\textwidth}
	% \vspace*{-26mm}
	% \hspace*{-2mm}
		\scalebox{0.90}{
    \hspace*{-0mm}
	  \begin{tikzpicture}[->,>=stealth',shorten >=0pt,auto,node distance=1.2cm,
                    scale=0.8,transform shape,initial text={}]
  \tikzstyle{every state}=[inner sep=3pt,minimum size=5pt]
  \tikzstyle{empty}=[]
  \tikzstyle{initstate}=[fill=yellow!30]

  \node[state,initial,initstate] (q) {$q$};
  \node[state,initial,initstate,right of=q,node distance=14mm] (r) {$r$};

  \path (q) edge[loop above]  node {$a$} (q)
        (r) edge[loop above]  node {$a$} (r);

\end{tikzpicture}
		}

	  \vspace{7mm}
		\scalebox{0.90}{
    \hspace*{-0mm}
	  \begin{tikzpicture}[>=stealth',shorten >=0pt,auto,node distance=1.2cm,
                    scale=0.8,transform shape,initial text={}]
  \tikzstyle{every state}=[inner sep=3pt,minimum size=5pt,rectangle,rounded corners=1mm]
  \tikzstyle{empty}=[]
  \tikzstyle{initstate}=[fill=yellow!30]
  \tikzstyle{wobbly}=[decorate, decoration={snake,amplitude=.2mm,segment length=2mm,post length=1mm}]

  \node[state,initial,initstate] (qr) {$\{q,r\}$};
  \node[state, accepting, below of=qr] (r1) {$\big(\{q{:}3, r{:}1\}, \emptyset, 0\big)$};
  \node[state, accepting, below of=r1] (r2) {$\big(\{q{:}1, r{:}1\}, \emptyset, 0\big)$};

  \node[empty,right of=qr,xshift=10mm,yshift=-5mm] (b1) {};
  \node[empty,right of=r1,xshift=10mm,yshift=-5mm] (b2) {};

  % \node[state, thick, red, cross out, below of=r1, minimum size=8mm] (cr1) {};

  \node[state, white, fill=white!20, right of=r2, node distance=23mm, minimum width=5mm] (sink) {};

  \node[empty,minimum size=30pt,xshift=-5.5mm] at(r1.east) (r1inv) {};

  \path[->] (qr) edge [loop right]  node {$a$} (rs)
        (qr) edge  node [left] {$a$} (r1)
        (r1inv) edge[loop right]  node {$a$} (r1inv)
        (r1) edge[wobbly]  node {$a$} (r2);

  \begin{pgfonlayer}{background}
    \draw (r2.south west) edge[line width=2mm, draw=red!40] (r2.north east);
    \draw (r2.north west) edge[line width=2mm, draw=red!40] (r2.south east);
  \end{pgfonlayer}

  \begin{pgfonlayer}{background}
    \path[->,dashed,gray] (qr)  edge [bend right=20]  node {} (b1)
      (r1) edge [bend right=6] node {} (b2)
      (r2) edge node {} (sink);
  \end{pgfonlayer}

\end{tikzpicture}
		}
	\caption{\centering}
	\label{fig:rank_restr}
\end{subfigure}

\caption{
  (\subref{fig:ranksim}) Illustration of $\ranksimred'$.
  (\subref{fig:rank_restr}) Illustration of \rankrestr. }
  \vspace{-2mm}
\end{figure}
}

The next optimization \ranksimred is a~modification of optimization $\purgedi$
% from our previous work in~\cite{ChenHL19}.
from~\cite{ChenHL19}.
Intuitively, $\purgedi$ is based on the fact that if a~state~$p$ is directly
simulated by a~state~$r$, i.e., $p \dirsimby r$, then any macrostate $\sofi$ where
$f(p) > f(r)$ can be safely removed
(intuitively, any run from~$p$ can be simulated by a~run from~$r$, where the run
from~$r$ may contain more accepting states and so needs to decrease its
rank more times).
$\purgedi$ is compatible with \algschewe but,
unfortunately, it is incompatible with the \algmaxrank construction
(one of our further optimizations introduced in \cref{sec:maxrank}) since in \algmaxrank, several runs are represented by one
\emph{maximal} run (w.r.t.\ the ranks) and removing such a~run would also remove
the smaller runs. % (see \cref{sec:maxrank} for details).
We, however, change the condition and obtain a new reduction, which is
incomparable with $\purgedi$ but compatible with~\algmaxrank.

Consider the following relation of \emph{odd-rank simulation} on~$Q$
defined such that $p \ors r$ iff
%
%\vspace{-3mm}
\begin{equation}
    \forall\word\in\Sigma^{\omega}, \forall i \geq 0\colon
    (\rankofin{p,i} \word \text{ is odd} \land \rankofin{r,i} \word \text{ is odd})
    \limpl
		\rankofin{p,i} \word \leq \rankofin{r,i} \word .
\end{equation}\\[-4mm]
% \begin{equation}
% 	\begin{split}
% 		p \ors r \defiff &\mbox{ if $ \forall\word\in\Sigma^{\omega}, i \in \omega:
% 		\rank_\word(p,i)$ and $\rank_\word(r,i)$ are both odd}\\
% 		&\mbox{ then }\rank_\word(p,i) \leq \rank_\word(r,i)
% 	\end{split}
% \end{equation}
%
%
Intuitively,
if $p \ors r$ holds, then in any super-tight run and a~macrostate
$\sofi$ in such a~run, if $p,r \in S$ and both~$f(p)$ and~$f(r)$ are odd, then
it needs to hold that $f(p) \leq f(r)$.
Such a~reasoning can also be applied transitively ($\ors$
is by itself not transitive):
if, in addition, $t \in S$, the rank~$f(t)$ is odd, and $r \ors t$, then it also
needs to hold that $f(p) \leq f(t)$.

% \ol{tranzitivita - jen na lichych stavech}
% This relation is not transitive. However, for a given set of odd states (states
% that have definitely odd rank) $T$ we can employ transitive relation between the
% states from $T$ (we can iteratively compute a transitive closure of states in
% $T$). Using this relation we can restrict the maximal rank in a macrostate $(S,
% O, f, i)$. For each subset $T\subseteq S$ we compute rank simulation (with
% transitive states $T$). From this relation we can get maximal rank (maximal
% number of distinct classes of $\preceq$). Finally we take maximum among these
% values.

% This relation is not transitive. However, for a given set of odd states (states
% that have definitely odd rank) $T$ we can employ transitive relation between the
% states from $T$ (we can iteratively compute a transitive closure of states in
% $T$). Using this relation we can restrict the maximal rank in a macrostate $(S,
% O, f, i)$. For each subset $T\subseteq S$ we compute rank simulation (with
% transitive states $T$). From this relation we can get maximal rank (maximal
% number of distinct classes of $\preceq$). Finally we take maximum among these
% values.

Formally, given a~ranking~$f$,
let~$\orsf$ be a~modification of~$\ors$ defined as
\vspace{-1mm}
\begin{equation}\label{eq:orsf}
  p \orsf r \defiff f(p) \text{ is odd} \land f(r) \text{ is odd} \land p \ors r
\end{equation}\\[-5mm]
and $\orsfT$ be its transitive closure.
We use~$\orsfT$ to define the following condition:
%
% \begin{equation}
%   \begin{split}
%     \ranksimred(\sofi) \defiff {} &\forall \{q_1, \ldots, q_n\} \subseteq  S:
%     \big( (\forall q_i : f(q_i) \text{ is odd}) \land {} \\
%     & q_1 \ors q_2 \ors \cdots \ors q_n \big) \limpl f(q_1) \leq f(q_n).
%   \end{split}
% \end{equation}
\vspace{-1mm}
\begin{equation}\label{eq:ranksimred}
  \begin{split}
    \ranksimred(\sofi) \defiff {} &\forall p, r \in S\colon p \orsfT r \limpl f(p) \leq f(r).
  \end{split}
\end{equation}\\[-5mm]
Abusing the notation, let $\ranksimred(\aut)$ denotes the output of
$\algschewe(\aut) = (Q', \delta', I', F')$ where states from the tight part
of~$Q'$ are restricted to those that satisfy $\ranksimred$.

%\vspace{-1mm}
\begin{restatable}{lemma}{lemRankSimRed}
\label{lem:lemRankSimRed}
Let~$\aut$ be a~BA.
Then $\langof{\ranksimred(\aut)} = \langof{\algschewe(\aut)}$.
\end{restatable}
%\vspace{-1mm}
\begin{proof}
	The inclusion $\langof{\ranksimred(\aut)} \subseteq \langof{\algschewe(\aut)}$
	is clear.
  For the reverse direction, let $\alpha \in \langof{\algschewe(\aut)}$.
  Then, there is a super-tight run $\rho = S_1\dots S_m \sofiof {m+1} \dots $,
  i.e., for each $k > m$ and each $q\in S_k$ we have $f_k(q) =
  \rank_\word(q,k)$.
  Clearly, each macrostate of $\rho$ satisfies \ranksimred.
\end{proof}

From the definition of~$\ors$, it is not immediate how to compute it,
since it is defined over all infinite runs of~$\aut$ over all infinite words.
The computation of a~rich under-approximation of~$\ors$ will be the topic of the
rest of this section.
We first note that ${\dirsimby} \subseteq {\ors}$, which is a~consequence of the
following lemma.

%\vspace{-1mm}
\begin{lemma}[Lemma~7 in \cite{ChenHL19}]
  \label{lem:aplas}
Let~$p,r \in Q$ be such that~$p \dirsimby r$ and $\graphof \word = (V, E)$ be
the run DAG of~$\aut$ over~$\word$.
For all $i \geq 0$,  $((p,i) \in V \land (r, i) \in V) \limpl
\rankofin{p, i} \word \leq \rankofin{r,i} \word$.
\end{lemma}
%\vspace{-1mm}

\noindent
We extend $\dirsimby$ into a~relation~$\orsr$, which is computed
statically on~$\aut$, and then show that ${\orsr} \subseteq {\ors}$.
The relation~$\orsr$ is defined recursively as the smallest binary relation over~$Q$
s.t.
\begin{inparaenum}[(i)]
  \item  ${\dirsimby} \subseteq {\orsr}$ and
  \item  for $p,r \in Q$, if~$\forall a \in \Sigma: (\delta(p,a)\setminus F)
    \orsraa (\delta(r,a)\setminus F)$, then $p \orsr r$.
\end{inparaenum}
Here, $S_1 \orsraa S_2$ holds iff $\forall x\in S_1, \forall y\in
S_2: x\orsr y$.
The relation~$\orsr$ can then be computed using a~standard \emph{worklist}
algorithm, starting from~$\dirsimby$ and adding pairs of states for which condition~2
holds until a~fixpoint is reached.

\begin{figure}[t]
\begin{subfigure}[b]{0.6\textwidth}
  \begin{minipage}{4cm}
    \vspace*{-10mm}
		\scalebox{0.90}{
	  \begin{tikzpicture}[->,>=stealth',shorten >=0pt,auto,node distance=1.2cm,
                    scale=0.8,transform shape,initial text={}]
  \tikzstyle{every state}=[inner sep=3pt,minimum size=5pt]
  \tikzstyle{empty}=[]
  \tikzstyle{initstate}=[fill=yellow!30]

  \node[state,initial,initstate] (qini) {$q_0$};
  \node[state,right of=qini] (q1) {$q_1$};
  \node[state,right of=q1, yshift=6mm] (q2) {$q_2$};
  \node[state,accepting,below of=q2] (q3) {$q_3$};
  \node[state,accepting,right of=q2] (q4) {$q_4$};

  \path (q1) edge  node {$a$} (q2)
        (qini) edge  node {$a$} (q1)
        (q1) edge  node {$a$} (q3)
        (q2) edge  node {$a$} (q4)
        (q4) edge [loop right]  node {$a$} (q4)
        (q3) edge[loop right]  node {$a$} (q3);

  \node[state,initial,initstate,below of=qini, node distance=1.8cm] (rini) {$r_0$};
  \node[state, right of=rini] (r1) {$r_1$};
  \node[state,right of=r1] (r2) {$r_2$};
  \node[state,accepting,right of=r2] (r3) {$r_3$};

  \path (r1) edge  node {$a$} (r2)
        (rini) edge  node {$a$} (r1)
        (r2) edge  node {$a$} (r3)
        (r3) edge [loop right]  node {$a$} (r3);

\end{tikzpicture}}
	  % \vspace{7mm}
  \end{minipage}
  \begin{minipage}{4cm}
		\scalebox{0.90}{
    % \hspace*{-5mm}
	  \begin{tikzpicture}[>=stealth',shorten >=0pt,auto,node distance=1.2cm,
                    scale=0.8,transform shape,initial text={}]
  \tikzstyle{every state}=[inner sep=3pt,minimum size=5pt,rectangle,rounded corners=1mm]
  \tikzstyle{empty}=[]
  \tikzstyle{initstate}=[fill=yellow!30]

  \node[state,initial,initstate] (q0r0) {$\{q_0, r_0\}$};
  \node[state,below of=q0r0] (q1r1) {$\{q_1, r_1\}$};
  \node[state, below of=q1r1] (q2q3r2) {$\{q_2,q_3,r_2\}$};
  \node[state, below of=q2q3r2] (q3q4r3) {$\{q_3,q_4,r_3\}$};

  \node[state, accepting, right of=q0r0, node distance=35mm] (m1) {$\big(\{q_1{:}1, r_1{:}3\}, \emptyset, 0\big)$};
  \node[state, accepting, below of=m1] (m2) {$\big(\{q_1{:}3, r_1{:}1\}, \emptyset, 0\big)$};
  \node[state, accepting, below of=m2] (m3) {$\big(\{q_1{:}1, r_1{:}1\}, \emptyset, 0\big)$};
  \node[empty,minimum size=30pt,xshift=-5.5mm] at(q3q4r3.east) (r345inv) {};

  \node[state, white, fill=white!20, below of=m3, node distance=16mm, minimum width=40mm, minimum height=12mm] (sink) {};

  \path[->] (q1r1) edge  node [left] {$a$} (q2q3r2)
        (q0r0) edge  node [left] {$a$} (q1r1)
        (q2q3r2) edge  node [left] {$a$} (q3q4r3)
        (q0r0) edge [] node {$a$} (m1)
        (q0r0) edge [bend right=10] node {$a$} (m2)
        (q0r0) edge [bend right=16] node[yshift=-1mm] {$a$} (m3)
        (r345inv) edge[loop right]  node {$a$} (r345inv);

  % \begin{pgfonlayer}{background}
  %   \path[dashed,gray] (q1r1)  edge [bend right]  node {} (sink)
  %     (q2q3r2) edge [bend right] node {} (sink)
  %     (q3q4r3) edge [bend right] node {} (sink)
  %     (m1) edge  node {} (sink)
  %     (q0r0) edge  node {} (sink)
  %     (q1r1) edge  node {} (sink);
  % \end{pgfonlayer}

  \begin{pgfonlayer}{background}
    \draw (m1.south west) edge[line width=2mm, draw=red!40] (m1.north east);
    \draw (m1.north west) edge[line width=2mm, draw=red!40] (m1.south east);

    \draw (m2.south west) edge[line width=2mm, draw=red!40] (m2.north east);
    \draw (m2.north west) edge[line width=2mm, draw=red!40] (m2.south east);
  \end{pgfonlayer}

\end{tikzpicture}
		}
  \end{minipage}
	\vspace*{-3mm}
	\caption{\centering}
	\label{fig:ranksim}
\end{subfigure}
\hfill
\begin{subfigure}[b]{0.2\textwidth}
	% \vspace*{-26mm}
	% \hspace*{-2mm}
		\scalebox{0.90}{
    \hspace*{-0mm}
	  \begin{tikzpicture}[->,>=stealth',shorten >=0pt,auto,node distance=1.2cm,
                    scale=0.8,transform shape,initial text={}]
  \tikzstyle{every state}=[inner sep=3pt,minimum size=5pt]
  \tikzstyle{empty}=[]
  \tikzstyle{initstate}=[fill=yellow!30]

  \node[state,initial,initstate] (q) {$q$};
  \node[state,initial,initstate,right of=q,node distance=14mm] (r) {$r$};

  \path (q) edge[loop above]  node {$a$} (q)
        (r) edge[loop above]  node {$a$} (r);

\end{tikzpicture}
		}

	  \vspace{7mm}
		\scalebox{0.90}{
    \hspace*{-0mm}
	  \begin{tikzpicture}[>=stealth',shorten >=0pt,auto,node distance=1.2cm,
                    scale=0.8,transform shape,initial text={}]
  \tikzstyle{every state}=[inner sep=3pt,minimum size=5pt,rectangle,rounded corners=1mm]
  \tikzstyle{empty}=[]
  \tikzstyle{initstate}=[fill=yellow!30]
  \tikzstyle{wobbly}=[decorate, decoration={snake,amplitude=.2mm,segment length=2mm,post length=1mm}]

  \node[state,initial,initstate] (qr) {$\{q,r\}$};
  \node[state, accepting, below of=qr] (r1) {$\big(\{q{:}3, r{:}1\}, \emptyset, 0\big)$};
  \node[state, accepting, below of=r1] (r2) {$\big(\{q{:}1, r{:}1\}, \emptyset, 0\big)$};

  \node[empty,right of=qr,xshift=10mm,yshift=-5mm] (b1) {};
  \node[empty,right of=r1,xshift=10mm,yshift=-5mm] (b2) {};

  % \node[state, thick, red, cross out, below of=r1, minimum size=8mm] (cr1) {};

  \node[state, white, fill=white!20, right of=r2, node distance=23mm, minimum width=5mm] (sink) {};

  \node[empty,minimum size=30pt,xshift=-5.5mm] at(r1.east) (r1inv) {};

  \path[->] (qr) edge [loop right]  node {$a$} (rs)
        (qr) edge  node [left] {$a$} (r1)
        (r1inv) edge[loop right]  node {$a$} (r1inv)
        (r1) edge[wobbly]  node {$a$} (r2);

  \begin{pgfonlayer}{background}
    \draw (r2.south west) edge[line width=2mm, draw=red!40] (r2.north east);
    \draw (r2.north west) edge[line width=2mm, draw=red!40] (r2.south east);
  \end{pgfonlayer}

  \begin{pgfonlayer}{background}
    \path[->,dashed,gray] (qr)  edge [bend right=20]  node {} (b1)
      (r1) edge [bend right=6] node {} (b2)
      (r2) edge node {} (sink);
  \end{pgfonlayer}

\end{tikzpicture}
		}
	\caption{\centering}
	\label{fig:rank_restr}
\end{subfigure}

\caption{
  (\subref{fig:ranksim}) Illustration of $\ranksimred'$.
  (\subref{fig:rank_restr}) Illustration of \rankrestr. }
  \vspace{-2mm}
\end{figure}

\begin{restatable}{lemma}{lemOrs}
We have ${\orsr} \subseteq {\ors}$.
\end{restatable}
\begin{proof}
  The base case ${\simdi} \subseteq {\ors}$ follows directly from
  \cref{lem:aplas}.
  For the induction step, let $p, r \in Q$ be such that
  $\forall a \in \Sigma: (\delta(p,a)\setminus F) \orsraa (\delta(r,a)\setminus
  F)$.
  Our induction hypothesis is that for every $a \in \Sigma, x \in
  (\delta(p,a)\setminus F)$, and $y \in
  (\delta(r,a)\setminus F)$, it holds that
  for all
  $\word \in \Sigma^\omega$ and for all $i \geq 0$, if
  $\rank_\word(p,i)$ is odd and $\rank_\word(r,i)$ is odd, then
   $\rank_\word(p,i) \leq \rank_\word(r,i)$.
   Let us fix an~$a \in \Sigma$ and a~word~$\word \in \Sigma^\omega$ that has
   $a$ at its $i$-th position.
   If $\rankofin{p,i} \word$ or $\rankofin{r,i} \word$ are even, the condition
   holds trivially.

   Assume now that $\rankofin{p,i} \word$ and $\rankofin{r,i} \word$ are odd.
   From the construction of the run DAG~$\dagw$ in \cref{sec:rundag}, it follows
   that there exist infinite paths from $(p,i)$ and $(r, i)$ in~$\dagw$ such
   that all vertices on these paths are assigned the same (odd) ranks as $(p,i)$
   and $(r,i)$, respectively.
   In particular, there are direct successors $(p', i+1)$ of $(p, i)$ and $(r',
   i+1)$ of $(r,i)$ whose ranks match the ranks of their predecessors.
   From the induction hypothesis, it holds that $\rankofin{p',i+1} \word \leq
   \rankofin{r',i+1} \word$ and so $\rankofin{p, i} \word \leq \rankofin{r, i}
   \word$ and the lemma follows.
   (Note that in the previous reasoning, it is essential that $(p, i)$ and $(r,
   i)$ have an odd ranking; if a~node has an even ranking in~$\dagw$, then the
   condition that there needs to be a successor with the same ranking does not
   hold in general.)
\end{proof}

Putting it all together, we modify~\eqref{eq:ranksimred} by substituting
$\orsfT$ with $\orsrfT$, which denotes the transitive closure of~$\orsrf$, where $\orsrf$ is a~relation defined (by modifying~\eqref{eq:orsf}) as
\begin{equation}
  \begin{split}
    p \orsrf r \defiff {}& f(p) \text{ is odd} \land f(r) \text{ is odd} \land {}
                         p \orsr r .
  \end{split}
\end{equation}

\noindent
Because ${\orsr} \subseteq {\ors}$, \cref{lem:lemRankSimRed} still holds.
We denote the modification of \ranksimred that uses $\orsrfT$ instead of~$\orsfT$
as $\ranksimred'$.

% \begin{proof}
% 	Base case $\simdi = \preceq$ follows directly from our APLAS paper.
%   Inductive
% 	case: Assume states $p$ and $q$ having an odd rank (in a run DAG level $\ell$
% 	for some word $\word$). Since both states have odd rank, there is a successor
% 	of each state having the same even rank (follows from the definition of run
% 	DAG). Lets say $p'$ and $q'$. Further assume that the condition in
% 	\ref{eq:odd-rank-ind} holds true. From that follows that $\rank_\word(p',
% 	\ell+1) \leq \rank_\word(q', \ell+1)$ which implies $\rank_\word(p, \ell)
% 	\leq \rank_\word(q, \ell)$ (and hence $p \preceq q$). \qed
% \end{proof}

% {\makeatletter
% \let\par\@@par
% \par\parshape0
% \everypar{}\figranksim
% \par}%
\begin{example}
Consider the BA $\aut$ (top) and the part of $\algschewe(\aut)$ (bottom) in
\cref{fig:ranksim}.
Note that $r_2\dirsimby q_2$ and $q_2\dirsimby r_2$ so
$r_2 \orsr q_2$ and $q_2 \orsr r_2$.
From the definition of $\orsr$, we can deduce that
$r_1 \orsr q_1$ (since $\{r_2\} \orsraa \{q_2\}$) and
$q_1 \orsr r_1$ (since $\{q_2\} \orsraa \{r_2\}$).
Note that $q_1 \not\dirsimby r_1$).
As a~consequence and due to the odd ranks of~$q_1$ and~$r_1$, we can eliminate
the macrostates $(\{q_1{:}1, r_1{:}3\}, \emptyset, 0)$ and
$(\{q_1{:}3, r_1{:}1\}, \emptyset, 0)$.
\end{example}

%*******************************************************************************
\vspace{-0.0mm}
\subsection{Ranking Restriction}
\vspace{-0.0mm}
%*******************************************************************************

\newcommand{\figrankrestr}[0]{
\begin{wrapfigure}[8]{r}{2.8cm}
\vspace*{-26mm}
\hspace*{-2mm}
\begin{minipage}{3.2cm}
  \centering
  \begin{tikzpicture}[->,>=stealth',shorten >=0pt,auto,node distance=1.2cm,
                    scale=0.8,transform shape,initial text={}]
  \tikzstyle{every state}=[inner sep=3pt,minimum size=5pt]
  \tikzstyle{empty}=[]
  \tikzstyle{initstate}=[fill=yellow!30]

  \node[state,initial,initstate] (q) {$q$};
  \node[state,initial,initstate,right of=q,node distance=14mm] (r) {$r$};

  \path (q) edge[loop above]  node {$a$} (q)
        (r) edge[loop above]  node {$a$} (r);

\end{tikzpicture}

  \vspace{2mm}
  \begin{tikzpicture}[>=stealth',shorten >=0pt,auto,node distance=1.2cm,
                    scale=0.8,transform shape,initial text={}]
  \tikzstyle{every state}=[inner sep=3pt,minimum size=5pt,rectangle,rounded corners=1mm]
  \tikzstyle{empty}=[]
  \tikzstyle{initstate}=[fill=yellow!30]
  \tikzstyle{wobbly}=[decorate, decoration={snake,amplitude=.2mm,segment length=2mm,post length=1mm}]

  \node[state,initial,initstate] (qr) {$\{q,r\}$};
  \node[state, accepting, below of=qr] (r1) {$\big(\{q{:}3, r{:}1\}, \emptyset, 0\big)$};
  \node[state, accepting, below of=r1] (r2) {$\big(\{q{:}1, r{:}1\}, \emptyset, 0\big)$};

  \node[empty,right of=qr,xshift=10mm,yshift=-5mm] (b1) {};
  \node[empty,right of=r1,xshift=10mm,yshift=-5mm] (b2) {};

  % \node[state, thick, red, cross out, below of=r1, minimum size=8mm] (cr1) {};

  \node[state, white, fill=white!20, right of=r2, node distance=23mm, minimum width=5mm] (sink) {};

  \node[empty,minimum size=30pt,xshift=-5.5mm] at(r1.east) (r1inv) {};

  \path[->] (qr) edge [loop right]  node {$a$} (rs)
        (qr) edge  node [left] {$a$} (r1)
        (r1inv) edge[loop right]  node {$a$} (r1inv)
        (r1) edge[wobbly]  node {$a$} (r2);

  \begin{pgfonlayer}{background}
    \draw (r2.south west) edge[line width=2mm, draw=red!40] (r2.north east);
    \draw (r2.north west) edge[line width=2mm, draw=red!40] (r2.south east);
  \end{pgfonlayer}

  \begin{pgfonlayer}{background}
    \path[->,dashed,gray] (qr)  edge [bend right=20]  node {} (b1)
      (r1) edge [bend right=6] node {} (b2)
      (r2) edge node {} (sink);
  \end{pgfonlayer}

\end{tikzpicture}
\end{minipage}
\vspace*{-7mm}
\caption{Illustration of \rankrestr}
\label{fig:rank_restr}
\end{wrapfigure}
}

%
% In order to reduce number of choices of the ranking function we modify the run
% DAG ranking procedure to the following. We assign ranks to vertices of run DAGs
% as follows: Let $\dagwiof 0 = \dagw$ and~$j = 0$. Repeat the following steps
% until the~fixpoint or for at most $2n + 2$ steps, where $n$~is the number of
% states of~$\aut$.
% %
% \begin{itemize}
% \vspace{-1.5mm}
%   \item  Set $\rankwof{p,i} := j$ for all finite vertices $(p,i)$
% 		of~$\dagwiof j$ s.t. there is a finite vertice $(u,j)$ where $j < i$ and
% 		$u\rightsquigarrow v$ with $u \in F$. Let $\dagwiof{j+1}$ be $\dagwiof{j}$ minus the
% 		vertices with the rank~$j$.
%   \item  Set $\rankwof{p,i} := j+1$ for all endangered vertices $(p,i)$
%     of~$\dagwiof {j+1}$ and let $\dagwiof{j+2}$ be $\dagwiof{j+1}$ minus the
%     vertices with the rank~$j+1$.
%   \item  Set $j := j + 2$.
% \vspace{-1.5mm}
% \end{itemize}
% %
% The modification consists only in assigning rank to finite vertices. The
% vertices having odd rank in the original procedure will have the same rank also
% in this modification. Note that the results concerning tight rankings holds also
% for this modified procedure (some vertices with even rank $k$ in the original
% procedure now get odd rank $k + 1$). This modified procedure reduce number of
% guesses -- the rank is deterministically decreased when a final state is
% occurred and nondeterministically when you leave a final state.

%named \rankrestr

%\figrankrestr
Another optimization, called \rankrestr, restricts ranks of successors
of states with an odd rank.
In particular, in a~super-tight run, every odd-ranked state has a~successor
with the same rank (this follows from the construction of the run DAG).
Let $\aut$ be a BA and $\but = \algschewe(\aut) = (Q, \delta_1 \cup \delta_2
\cup \delta_3, I, F)$.
We define the following restriction on transitions:
\begin{equation}
	\begin{split}
	\rankrestr&(\sofi \transover{a} \sofiprime) \defiff\\& \forall q\in
	S: f(q)\mbox{ is odd } \limpl (\exists q' \in\delta(q,a): f'(q') = f(q)).
	\end{split}
\end{equation}
We abuse notation and use $\rankrestr(\aut)$ to denote $\but$ with
transitions from~$\delta_3$ restricted to those that satisfy $\rankrestr$.
See \cref{fig:rank_restr} for an example of a~transition (and a~newly unreachable macrostate) removed using $\rankrestr$.
%
% Then the restricted transition
% function is denoted as $\R$ and it is the smallest relation s.t.
% \begin{inparaenum}[(i)]
% 	\item $\delta_1 \cup \delta_2 \subseteq \R$
% 	\item $\tau = \sofi \transover{a} \sofiprime \in \delta_3 \wedge \forall q\in
% 	S: f(q)\mbox{ is odd } \limpl (\exists q' \in\delta(q,a): f'(q') = f(q))$
% 	implies $\tau \in \R$.
% \end{inparaenum}
% Then $\rankrestr(\aut) = (Q, \R, I, F)$.

% \begin{lemma}
%   Let $\but = \algschewe(\aut) = (Q, \delta, \ignore, \ignore)$ and let $\sofi \in Q$.
%   Then for every $a \in \Sigma$, it holds that there is a~state~$\sofiprime \in Q$
%   such that $\sofi \transover{a} \sofiprime \in \delta$ and $\forall q\in S:
%   f(q)\mbox{ is odd } \limpl (\exists q' \in\delta(q,a): f'(q') = f(q))$.
% \end{lemma}

\begin{restatable}{lemma}{lemRankRestr}
	Let $\aut$ be a BA. Then $\langof{\rankrestr(\aut)} = \langof{\algschewe(\aut)}$.
\end{restatable}
\begin{proof}
	The inclusion $\langof{\rankrestr(\aut)} \subseteq \langof{\algschewe(\aut)}$
	is clear. We now look at the reverse direction. Consider a word $\word \in
	\langof{\algschewe(\aut)}$. Let $\rho = S_0\dots S_m(S_{m+1}, O_{m+1},
	f_{m+1}, i_{m+1})\dots$ be an accepting super-tight run on $\word$. Now
	consider a macrostate $(S_j, O_j, f_j, i_j)$ where $j > m$ and some state $q
	\in S_j$. Since $\rho$ is super-tight, i.e., represents
  the run DAG of~$\aut$ over~$\word$, it holds that if $f_j(q)$ is odd, then
  there is a~state $q'\in S_{j+1}$ s.t.  $f_{j+1}(q') = f_j(q)$, which satisfies
  $\rankrestr$.
\end{proof}

\vspace{-0.0mm}
\subsection{Maximum Rank Construction}\label{sec:maxrank}
\vspace{-0.0mm}
%*******************************************************************************

Our next optimization, named \algmaxrank, has the biggest
practical effect.  We introduce it as the last one because it depends on our
previous optimizations (in particular \succrankred and $\ranksimred'$).
It is a~modified version of Schewe's ``Reduced Average Outdegree''
construction~\cite[Section~4]{Schewe09}, named $\algschewerao$, which may omit some runs,
the so-called \emph{max-rank} runs, that are essential for our other
optimizations (we discuss the particular issue later).%
\footnote{We believe that this property was not originally intended by the
author, since it is not addressed in the proof.
As far as we can tell, the construction is correct, although the original
argument of the proof in~\cite{Schewe09} needs to be corrected.}

% It is a~corrected version of Schewe's ``Reduced Average Outdegree''
% construction~\cite[Section~4]{Schewe09}, named $\algschewerao$, which contains
% a~small but fatal bug (we discuss the particular issue later).
% \ol{change the claim}

The main idea of \algmaxrank is that a~set of runs of $\but = \algschewe(\aut)$
(including super-tight runs) that assign different ranks to non-trunk states is
represented by a~single, ``maximal,'' not necessarily super-tight (but having
the same rank), run in $\cut = \algmaxrank(\aut)$.
We call such runs \emph{max-rank runs}.
More concretely, when moving from the waiting to the tight part, $\cut$~needs to
correctly guess a~rank that is needed on an accepting run and the first tight
core of a~trunk of the run.
The ranks of the rest of states are made maximal.
Then, the tight part of~$\cut$ contains for each macrostate and symbol at
most two successors: one via $\eta_3$ and one via~$\eta_4$.
Loosely speaking, the~$\eta_3$-successor keeps all ranks as high as possible, while
the~$\eta_4$-successor decreases the rankings of all non-accepting states in~$O$
(and can therefore help emptying~$O$, which is necessary for an accepting run).

Before we give the construction, let us first provide some needed notation.
We now use $(S, O, f, i) \leq (S, O, g, i)$ to denote that $f \leq g$ and
similarly for~$<$ (note that non-ranking components of the macrostates need to
match).
% In the definitions, given a~set of macrostates~$R$ from~$Q_2$, we use
% $\max_{f}\{\sofi~\in~R \}$ to denote the set of maximal elements of
% the partial order $\leq$ on macrostates in~$R$.
% For instance, if
% $R = \{ (S, O, \{\rnk q 3, \rnk s 1\}, i),
%         (S, O, \{\rnk q 1, \rnk s 1\}, i)$,
%         $(S, O, \{\rnk q 1, \rnk s 3\}, i) \},$
% then
% $\max_{f}\{\sofi \in R\} = \{
%   (S, O, \{\rnk q 3, \rnk s 1\}, i)$,
%   $(S, O, \{\rnk q 1, \rnk s 3\}, i) \}.$
% \ol{introduce the notation $\rnk s 1$}

The construction is then formally defined as
$\algmaxrank(\aut) = (Q_1 \cup Q_2, \eta, I', F')$
with
$\eta = \delta_1 \cup \eta_2 \cup \eta_3 \cup \eta_4$
such that $Q_1, Q_2, I', F', \delta_1$ are the same as in \algschewe.
Let $\but = \algdelay(\aut) = (\ignore, \delta_1 \cup \theta_2 \cup \delta_3, \ignore, \ignore)$
where $\delta_1, \theta_2$, and $\delta_3$ are defined as in \algdelay.
% We define an auxiliary transition function that keeps macrostates satisfying conditions
% $\ranksimred'$ and \succrankred as follows:
We define an auxiliary transition function that uses our previous optimizations
as follows:
\vspace{-1mm}
\begin{equation}
  \auxdelta(q, a) = \{q' \mid q' \in \theta_2(q, a) \land \ranksimred'(q') \land \succrankred(q'))\} .
\vspace{-1mm}
\end{equation}
(We note that~$q$ is from the waiting and~$q'$ is from the tight part of~$\but$.)
Given a~macrostate $\sofi$ and $a \in \Sigma$, we define the maximal successor ranking
$\maxrankof{\sofi, a}$ to be the ranking~$\fpmax$ given as follows.
Consider $q' \in \delta(S, a)$ and the rank $r = \min\{f(s) \mid s \in \delta^{-1}(q',
a) \cap S\}$.
Then
\vspace{-0mm}
\begin{itemize}
  \setlength{\itemsep}{0mm}
  \item  $\fpmax(q') := r-1$ if~$r$ is odd and $q' \in F$ and
  \item  $\fpmax(q') := r$ otherwise.
\end{itemize}
\vspace{-0mm}

Let $\delta_3$ be the transition function of the tight part of
$\algschewe(\aut)$.
We can now proceed to the definition of the missing components of
$\algmaxrank(\aut)$:
\vspace{-0mm}
\begin{itemize}
  \setlength{\itemsep}{0mm}
  % \item $\eta_2(S, a) := \max_{f'}\{(S', \emptyset, f', 0) \in \auxdelta(S, a) \}$.
  \item $\eta_2(S, a) := \{(S', \emptyset, f', 0) \in \auxdelta(S, a) \mid (S',
    \emptyset, f', 0) \text{ is a~maximal element of } {\leq} \text{ in }
    \auxdelta(S,a)\}$.

	\item $\eta_3(\sofi, a)$:
    Let~$\fpmax = \maxrankof{\sofi, a}$.
    Then, we set
    \vspace{-1mm}
    \begin{itemize}
      \setlength{\itemsep}{0mm}
      \item  $\eta_3(\sofi, a) := \{(S', O', \fpmax, i')\}$ when $(S', O',
        \fpmax, i') \in \delta_3(\sofi, a)$ (i.e., if $\fpmax$ is tight; note that, in general, the result of $\maxrank$ may not be tight) and
      \item  $\eta_3(\sofi, a) := \emptyset$ otherwise.
    \end{itemize}
    \vspace{-0mm}

    % \ol{OLD}
    %
    % $= \max_{f'}\{\sofiprime \in \auxdelta(\sofi, a) \}$.
    % We emphasize that from the definition of~$\auxdelta$, it holds that
    % $\eta_3(\sofi, a)$ contains at most one element---there is indeed
    % a~unique ranking that is the supremum, but it is not necessarily tight
    % \ol{should this be proved?} \ol{give an example}), \ol{maybe lemma along
    % the lines of preserving the max. tight ranking through the pruning
    % techniques} \ol{example: if $\auxdelta(\sofi, a) = \{\{q \mapsto 3, r
    % \mapsto 1\}, \{q \mapsto 1, r \mapsto 3\}\}$, then $\eta_3(\sofi, a) =
    % \emptyset$}
    % \ol{change definition!! (we compute maximum of KV and check its tight}
	% \item $\sofiprime \in \eta_4(\sofi, a)$ if $i'' \neq 0$, $(S'',
	% O'', f'', i'')\in\eta_3(\sofi, a)$, $i' = i''$, $S' = S''$ $O' = F\cap
	% O''\cap (f'')^{-1}(i'')$, $f'' = f' \lhd \{ u\mapsto f'(u)-1 \mid u\in
	% O'\setminus F \}$.
  \item $\eta_4(\sofi, a)$: Let
    $\eta_3(\sofi, a) = \{(S', P', h', i')\}$ and let
    \vspace{-1mm}
    \begin{itemize}
      \setlength{\itemsep}{0mm}
      \item  $f' = \variant{h'}{\{ u\mapsto h'(u)-1 \mid u\in P'\setminus F \}}$ and
      \item  $O' = P'\cap f'^{-1}(i')$.
    \end{itemize}
    \vspace{-1mm}
    Then, if $i' \neq 0$, we set $\eta_4(\sofi, a) := \{\sofiprime\}$, else
    we set \mbox{$\eta_4(\sofi, a) := \emptyset$}.
  % \item $\sofiprime \in \eta_4(\sofi, a)$ if
  %   $\eta_3(\sofi, a) = \{(S', P', h', i')\}$ s.t.
  %   \begin{itemize}
  %     \item  $i' \neq 0$,
  %     \item  $f' = \variant{h'}{\{ u\mapsto h'(u)-1 \mid u\in P'\setminus F \}}$, and
  %     \item  $O' = P'\cap f'^{-1}(i')$.
  %   \end{itemize}
  %   We emphasize that $|\eta_4(\sofi, a)| \leq 1$.
  %   \ol{example}
\end{itemize}
\vspace{-0mm}
Note that~$\eta_3$ and~$\eta_4$ are deterministic (though not complete), so we
will sometimes use the notation $\sofiprime = \eta_3(\sofi, a)$.

\algmaxrank differs from $\algschewerao$ in the definition of~$\eta_2$ and~$\eta_4$.
In particular, in the $\eta_4$ of $\algschewerao$ (named $\gamma_4$ therein),
the condition that only non-accepting states
($u \in P' \setminus F$) decrease rank is omitted.
Instead, the rank of all states in~$P'$ is decreased by one, which might create
% a~``ranking'' that is actually not a~ranking according to the definition (since an
a~``false ranking'' (not an actual ranking since an
accepting state is given an odd rank), so the target macrostate is omitted
from the complement.
Due to this, some max-rank runs may also be removed (see
\cref{sec:schewe_bug_example} for an example).
Our construction preserves max-rank runs, which makes the proof of the
theorem significantly more involved.

\begin{restatable}{theorem}{theMaxRank}
  Let $\aut$ be a~BA and $\cut = \algmaxrank(\aut)$.
  Then $\langof \cut = \overline{\langof \aut}$.
\end{restatable}
\begin{proof}
  Wlog.\ assume that~$\aut$ is complete.
  Let $\but = \algschewe(\aut)$.
  Showing $\langof \cut \subseteq \langof \but$ is easy (the transitions
  of~$\cut$ are contained in the transitions of~$\but$).
  Next, we show that $\langof \cut \supseteq \langof \but$.

  Let $\word\in \langof \but$ and~$\rho$ be a~super-tight run (from
  \cref{lem:super-tight-run}, we know that a~super-tight run exists) of $\but$
  over $\word = \wordof 0 \wordof 1 \wordof 2 \ldots$ s.t.
  \begin{equation*}
    \rho = S_0\dots S_m(S_{m+1}, O_{m+1}, f_{m+1}, i_{m+1})(S_{m+2}, O_{m+2}, f_{m+2}, i_{m+2})\dots
  \end{equation*}
  Let $\tau = C_{m+1} C_{m+2} \ldots$ be a~trunk of~$\rho$.
  We will construct the run
  \begin{equation*}
    \rho' = S_0\dots S_m(S_{m+1}, O_{m+1}', f_{m+1}', i_{m+1}')(S_{m+2}, O_{m+2}', f_{m+2}', i_{m+2}')\dots
  \end{equation*}
  of~$\cut$ on~$\word$
  in the following way (note that the $S$-components of the macrostates
  traversed by~$\rho$ and~$\rho'$ are the same):
  \begin{enumerate}
    \item  For the transition from the waiting to the tight part, we set
      $O_{m+1}' := \emptyset$ and $i_{m+1}' := 0$.
      The ranking $f_{m+1}'$ is set as follows.
      % Let $C_{m+1}$ be a~tight core (if there are more, we arbitrarily pick
      % one) of the run~$\rho$ at position~$m+1$
      % (cf.~\cref{sec:super-tight-runs}) and $r$ be the rank of~$\rho$.
      Let $r$ be the rank of~$\rho$ (remember that~$\rho$ is super-tight).
      We first construct an auxiliary (tight) ranking
      \begin{equation}\label{eq:first-ranking}
        g = \variant{f_{m+1}}{\{u \mapsto \max\{r-1, f_{m+1}(u)\}
        \mid u \in S_{m+1} \setminus C_{m+1}\}}.
      \end{equation}
      Note that~$C_{m+1}$ is also a~tight core of~$g$.
      There are now two possible cases:
      \begin{enumerate}
        \item  $(S_{m+1}, \emptyset, g, 0) \in \eta_2(S_m, \wordof m)$: If this
          holds, we set $f_{m+1}' := g$.
        \item  Otherwise, $\eta_2(S_m, \wordof m)$ contains at least one
          macrostate with a ranking~$h$ s.t.~$g < h$.
          We pick an arbitrary such ranking~$h$ from $\eta_2(S_m, \wordof m)$ and set
          $f_{m+1}' := h$.
          Note that~$C_{m+1}$ is also a~tight core of~$h$.
      \end{enumerate}
      Note that $\eta_2(S_m, \wordof m)$ contains at least one macrostate
      $(S_{m+1}, \emptyset, h, 0)$ with the rank~$r$ such that $g \leq h$.
      This follows from the fact that the reductions $\succrankred$ and $\ranksimred$
      only remove macrostates that do not occur on super-tight runs and that
      $(S_{m+1}, \emptyset, g, 0)$ is not removed using the reductions.
      The latter follows from~\eqref{eq:succrankred} and~\eqref{eq:ranksimred}.
      \ol{some discussion might be nice, but let's leave it for a journal
      submission}

    \item  Let $k > m$ and $i_*$ be such that $(\ignore, \ignore, f_*, i_*)
      = \eta_4(\sofiprimeof k, \wordof k)$.
      Then,
      \begin{itemize}
        \item  we set $(S_{k+1},
          O_{k+1}', f_{k+1}', i_{k+1}') := \eta_4((S_{k}, O_{k}', f_{k}',
          i_k'), \wordof k)$ if $O_{k+1} = \emptyset$, $i_* = i_{k+1}$, and
          $f_* \geq f_{k+1}$,
        \item otherwise, we set $(S_{k+1}, O_{k+1}', f_{k+1}', i_{k+1}') :=
          \eta_3((S_{k}, O_{k}', f_{k}', i_k'), \wordof k)$.
        % \item if $(S_{k+1}, \emptyset, f_{k+1}, i_{k+1})\in\pi$ and $i_{k+1} =
        %   i_{k+1}'$ then $(S_{k+1}, O_{k+1}', f_{k+1}',
        %   i_{k+1}')\in\eta_4((S_{k}, O_{k}', f_{k}', i_k'))$ ($i_{k+1}'$ does
        %   not depend on $\eta_3$ or $\eta_4$),
        % \item $(S_{k+1}, O_{k+1}', f_{k+1}', i_{k+1}')\in\eta_3((S_{k}, O_{k}',
        %   f_{k}', i_k'))$ otherwise.
      \end{itemize}
    % \item  Let $k > m$ and $(S, O_*, f_*, i_*) = \eta_4((S_k, O_k', f_k', i_k'),
    %   \wordof k)$.
    %   Then,
    %   %
    %   \begin{itemize}
    %     \item   if $O_{k+1} = \emptyset$ and $i_* = i_{k+1}$, we set $(S_{k+1},
    %       O_{k+1}', f_{k+1}', i_{k+1}') := (S, O_*, f_*, i_*)$,
    %       i.e., we use the output of $\eta_4$,
    %     \item otherwise, we set $(S_{k+1}, O_{k+1}', f_{k+1}', i_{k+1}') := \eta_3((S_{k}, O_{k}',
    %       f_{k}', i_k'), \wordof k)$.
    %     % \item if $(S_{k+1}, \emptyset, f_{k+1}, i_{k+1})\in\pi$ and $i_{k+1} =
    %     %   i_{k+1}'$ then $(S_{k+1}, O_{k+1}', f_{k+1}',
    %     %   i_{k+1}')\in\eta_4((S_{k}, O_{k}', f_{k}', i_k'))$ ($i_{k+1}'$ does
    %     %   not depend on $\eta_3$ or $\eta_4$),
    %     % \item $(S_{k+1}, O_{k+1}', f_{k+1}', i_{k+1}')\in\eta_3((S_{k}, O_{k}',
    %     %   f_{k}', i_k'))$ otherwise.
    %   \end{itemize}
  \end{enumerate}
  Intuitively, $\rho'$ simulates the super-tight run~$\rho$ of~$\but$ with the
  difference that
  \begin{inparaenum}[(i)]
    \item  the transition from the waiting to the tight part sets the ranks of
      all non-core states to $r-1$,
    \item  in the tight part, $\rho'$~keeps taking the maximizing~$\eta_3$
      transitions until it happens that $\rho'$ is stuck with emptying some~$O$,
      in which case, the ranks of all non-accepting states in~$O$ are decreased
      (the~$\eta_4$ transition).
  \end{inparaenum}

  First, we prove that the run~$\rho'$ constructed according to the procedure
  above is infinite.
  Intuitively, there are two possibilities how the construction of~$\rho'$ can break:
  \begin{inparaenum}[(i)]
    \item  the macrostate $(S_{m+1}, \emptyset, f_{m+1}', 0)$ is not in
      $\eta_2(S_m, \wordof m)$,
    \item  $\eta_3(\sofiof m, \wordof m) = \emptyset$, or
    \item  $\eta_4(\sofiof m, \wordof m) = \emptyset$.
  \end{inparaenum}

	\begin{claim}
	\label{claim:prop_run}
    For every $k > m$ the following conditions hold:
    \begin{enumerate}[(i)]
      \item  the macrostate $\rho'(k)$ is well defined,
      \item  $f'_{k} \geq f_{k}$, and
      \item  $C_k$ is a~tight core of~$\rho'(k)$.
    \end{enumerate}
  \end{claim}
  \begin{claimproof}
    By induction on the position $k>m$ in~$\rho'$.
    \begin{itemize}
      \item  $k = m+1$:
        \begin{enumerate}[(i)]
          \item  Proving $(S_{m+1}, \emptyset, f_{m+1}', 0) \in \eta_2(S_m, \wordof m)$:
            This easily follows from the construction of $\sofiof{m+1}$ given
            above.
          \item  Proving $f'_{m+1} \geq f_{m+1}$:
            This, again, easily follows from the construction of $\sofiof{m+1}$.
            In particular, the $g$ constructed in~\eqref{eq:first-ranking}
            satisfied the property $g \geq f_{m+1}$ and the~$f'_{m+1}$
            constructed from it satisfies $f_{m+1}' \geq g$.
          \item We have that $C_{m+1}$ is a tight core of $\sofiof{m+1}$. From
						the definition of $f_{m+1}'$ we directly obtain that  $C_{m+1}$ is
						also a tight core of $\sofiprimeof{m+1}$.
        \end{enumerate}

      \item  $k+1$: Suppose the claim holds for~$k$.
        \begin{enumerate}[(i)]
          \item  (and (iii)) For proving $\rho'(k+1)$ is well-defined, from the
            construction, we need to prove the following:
            \begin{itemize}
              \item  \emph{If $\rho'(k+1)$ is the $\eta_3$-successor of $\rho'(k)$,
                we need to show that $\eta_3(\sofiprimeof k, \wordof k) \neq
                \emptyset$.}

                This condition can be proved by showing that the ranking $\fpmax =
                \maxrankof{\sofiprimeof k, \wordof k}$ is tight.
                From the induction hypotheses
                (``$f'_k \geq f_k$'' and ``$C_k$ is a~tight core of $f'_k$''),
                we know that $f_k$ and $f'_k$ coincide on states from~$C_k$.
                Further, from \cref{lem:trunks}, it holds that for every state
                $q_k \in C_k$ there is a state $q_{k+1} \in C_{k+1}$ such that
                $f_k(q_k) = f_{k+1}(q_{k+1})$.
                From the construction of~$\fpmax$, we can conclude that it also
                holds that $f'_k(q_k) = \fpmax(q_{k+1})$ (which proves~(iii)).
                Using induction hypothesis one more time (``$f'_k$ is tight''),
                we can conclude that $\fpmax$ is also tight.

              \item  \emph{If $O_{k+1} = \emptyset$, $i_{k+1} = i_{k+1}'$, and
                $f'_{k+1} \geq f_{k+1}$ hold at the same time (i.e., $\rho'(k+1)$ is the
                $\eta_4$-successor of $\rho'(k)$), we need to show that
                $\eta_4(\sofiprimeof k, \wordof k) \neq~\emptyset$.}

                Above, we have already shown that $\eta_3(\sofiprimeof k,
                \wordof k) \neq~\emptyset$.
                From the definition,
                in order for $\eta_4(\sofiprimeof k, \wordof k) = \emptyset$, it
                would need to hold that $i_{k+1}' = 0$.
                Since our assumption is that~$\aut$ is complete and we know
                that~$\rho$ is accepting, it needs to hold that at every
                position $j > m$, we have $f_j(q) > 0$ for any state~$q \in
                S_j$ (otherwise, if $q$ appeared in the~$O$-component of some
                macrostate in~$\rho$, it would never disappear and so~$\rho$
                could not be accepting).

                The proof of (iii) easily follows from the previous step
                for~$\eta_3$, since the ranking function of the result
                of~$\eta_4$ differs from the one for~$\eta_3$ only on states
                from the $O$-component, which are even and therefore, by
                definition, not in a~tight core.
            \end{itemize}

          \item  Proving $f'_{k+1} \geq f_{k+1}$ assuming the induction
            hypothesis $f'_k \geq f_k$:
            \begin{itemize}
              \item  \emph{If $\rho'(k+1)$ is the $\eta_3$-successor of
                $\rho'(k)$},
                $f'_{k+1} \geq f_{k+1}$ follows immediately from the fact that
                the~$\eta_3$ transition function yields the maximal successor
                ranking.

              \item  \emph{If $O_{k+1} = \emptyset$, $i_{k+1} = i_{k+1}'$, and
                $f'_{k+1} \geq f_{k+1}$ hold at the same time (i.e., $\rho'(k+1)$ is the
                $\eta_4$-successor of $\rho'(k)$)},
                $f'_{k+1} \geq f_{k+1}$ is already a~condition for~$\eta_4$ to be taken.
                % \claimqed
                \claimqedhere
            \end{itemize}
        \end{enumerate}
    \end{itemize}

  \end{claimproof}

  Next, we will show that $\rho'$ is accepting, i.e., that $O$-component of
  macrostates in~$\rho'$ is emptied infinitely many times.
  For the sake of contradiction, assume that $\rho'$ is not accepting, i.e.,
  for some $\ell > m$, it happens that for all $k \geq \ell$ it holds that
  $O'_k \neq \emptyset$ and $i'_k = i'_\ell$ (the run is ``stuck'' at
  some~$i'$ and cannot empty~$O'$).
  We will show that if~$\rho'$ is ``stuck'' at some $i'$, it will contain
  infinitely many macrostates obtained using an~$\eta_4$ transition.
  An~$\eta_4$ transition decreases ranks of all non-final states in~$O'$ and, as
  a~consequence, removes such tracked runs from~$O'$.
  Therefore, if the rank of some run in~$O'$ is infinitely often not decreased
  by~$\eta_4$, there needs to be a~corresponding run of~$\aut$ with infinitely
  many occurrences of an accepting state, so it would need to hold that $\word
  \in \aut$, leading to a~contradiction.

  Let us now prove the previous reasoning more formally.
  First, we show that~$\rho'$ needs to contain infinitely many occurrences
  of~$\eta_4$-obtained macrostates.
  Since~$\rho$ is accepting, it satisfies infinitely often the condition
  that~$O$ is empty and $i = i'_\ell$.
  To satisfy the condition for executing an~$\eta_4$ transition,
  we need to show that for infinitely many~$k$ it in addition holds
  that $ f_k' = \variant{\fpmax}{\{ q\mapsto \fpmax(q)-1 \mid u\in O_k'\setminus
  F \}} \geq f_k$ where~$\fpmax$ is as in the definition of~$\eta_3$.
  In particular, we will show that for infinitely many~$k$, we will have $i_k =
  i'_\ell$, $O_k = \emptyset$, and $\forall q \in O'_k \setminus F: \fpmax(q) >
  f(q)$ (from \cref{claim:prop_run} we already know that $\fpmax \geq
  f_k$).

  Let~$p > \ell$ be a~position such that
  $i_{p-1} \neq i'_\ell$, $O_{p-1} = \emptyset$, and $i_p = i'_\ell$, i.e.,
  a~position at which run~$\rho$ started emptying runs with rank~$i'_\ell$.
  Because~$\rho$ is accepting, there is a~position $k \geq p$ such that
  $\rho(k) = (S_{k}, \emptyset, f_{k}, i'_\ell)$, therefore, the ranks of
  all runs tracked in~$O_p$ were decreased (otherwise, $O_{k}$~could not be
  empty).
  Consider the following claim.

  \begin{claim}\label{claim:emptying}
    $\forall q \in O'_k: f_k(q) < f'_k(q)$
  \end{claim}

  \begin{claimproof}
  The weaker property $f_k \leq f'_k$ follows from
  \cref{claim:prop_run}.
  We prove the strict inequality for states in $O'_k$ by contradiction.
  Assume that $f_k(q) = f'_k(q)$ for some $q \in O'_k$.
  Then there needs to be a~predecessor~$s$ of~$q$ in~$S_p$ such that $f_p(s) =
  i'_\ell$ and so $s \in O_p$.
  But since $q \notin O_k$, then somewhere between $p$ and $k$, the rank of the
  run in~$\aut$ must have been decreased.
  Therefore, $f_k(q) < f'_k(q)$.
  \end{claimproof}

  From \cref{claim:emptying} and the fact that $f'_k(q) \leq \fpmax$ it follows
  that $\forall q \in O'_k \setminus F: \fpmax(q) > f(q)$, so an~$\eta_4$
  transition was taken infinitely often in~$\rho'$.

  The last thing to show is that when $\eta_4$ is taken infinitely often,
  $O'$~will be eventually empty.
  The condition does not hold only in the case when the rank of a~run of~$\aut$
  tracked in~$O'$ is infinitely ofren not decreased because it is represented
  in~$O'$ by a~final state $q \in O' \cap F$.
  But then~$\aut$ contains a~run over~$\word$ that touches an accepting state
  infinitely often, so $\word \in \langof \aut$, which is a~contradiction.
  % Since~$\rho$ is accepting and, therefore, infinitely often emptying~$O$ and
  % rotating the $i$-component, there will be a~position~$p > \ell$ such that
  % $i_{p-1} \neq i'_\ell$, $O_{p-1} = \emptyset$, and $i_p = i'_\ell$.
  %
  % Because~$\rho$ is accepting, there is a~position $k \geq p$ such that
  % $\rho(k) = (S_{k}, \emptyset, f_{k}, i'_\ell)$, therefore, the ranks of
  % all runs tracked in~$O_p$ were decreased (otherwise, $O_{k}$~could not be
  % empty).
  %
	% 	\begin{itemize}
	% 	\item Further observe that after trigering $\eta_4$ transition, we
	% 		$O_{k+1}'$ contains only finite states (or it is an emptyset). This is
	% 		given by the construction.
	% 	\begin{itemize}
	% 		\item We need to show that $\rho'$ contains infinite many final states (i.e.,
	% 			we flush $O'$ infinitely many times).
	% 		\item Assume the opposite, $O$s in $\rho$ flush infinitely many times, but
	% 			$O'$ in $\rho'$ don't.
	% 		\item Since $\rho$ flush infinitely many times, we know that the transition
	% 			$\eta_4$ was trigered infinite many times.
	% 		\item Since $\rho'$ is not accepting, there is some infinite path $\xi$
	% 			taken from $\rho'$ having all states from some position equaly even, i.e.
	% 			$f_\ell'(\xi_\ell)$ is even (here we mean by the path sequence of states
	% 			of the original uncompleted BA).
	% 		\item The only way this happen is that $\xi$ contains infinitely many
	% 			final states, which is contradiction (recall that after transition
	% 			$\eta_4$, the $O'$ part contains final states only).
	% 	\end{itemize}
	% \end{itemize}
\end{proof}

Note that \algmaxrank is incompatible with \rankrestr since \rankrestr optimizes
the transitions in the tight part of the complement BA, which
are abstracted in \algmaxrank.

\vspace{-0.0mm}
\subsection{Backing Off}
\vspace{-0.0mm}
%*******************************************************************************

Our final optimization, called \algbackoff, is a~strategy for guessing when our
optimized rank-based construction is likely (despite the optimizations) to
generate too many states and when it might be helpful to give up and use
a~different complementation procedure instead.
We evaluate this after the initial phase of \algschewe, which
constructs~$\delta_2$
($\eta_2$~in \algmaxrank, $\theta_2$ in \algdelay; we will just
use~$\delta_2$ now), finishes.
We provide a~set of pairs $\{(\StateSize_j, \RankMax_j)\}_{j \in
\J}$ for an index set~$\J$ (obtained experimentally)
and check (after~$\delta_2$ is constructed) that for all
$\sofi \in \imgof{\delta_2}$ and all $j \in \J$ it holds that
either $|S| < \StateSize_j$ or $\rankof f < \RankMax_j$.
If for some $\sofi$ and~$j$ the condition does not hold, we terminate the
construction and execute a~different, \emph{surrogate}, procedure.

%%%%%%%%%%%%%%%%%%%%%%%%%%%%%%%%%%%%%%%%%%%%%%%%%%%%%%%%%%%%%%%%%%%%%%%%%%%%%%
%%%%%%%%%%%%%%%%%%%%%%%%%%                  %%%%%%%%%%%%%%%%%%%%%%%%%%%%%%%%%%
%%%%%%%%%%%%%%%%%%%%%%%%%% TABLES AND PLOTS %%%%%%%%%%%%%%%%%%%%%%%%%%%%%%%%%%
%%%%%%%%%%%%%%%%%%%%%%%%%%                  %%%%%%%%%%%%%%%%%%%%%%%%%%%%%%%%%%
%%%%%%%%%%%%%%%%%%%%%%%%%%%%%%%%%%%%%%%%%%%%%%%%%%%%%%%%%%%%%%%%%%%%%%%%%%%%%%
\newcommand{\figother}[0]{
\begin{figure}[t]
\begin{subfigure}[b]{0.49\linewidth}
	\centering
\includegraphics[width=5.5cm,keepaspectratio]{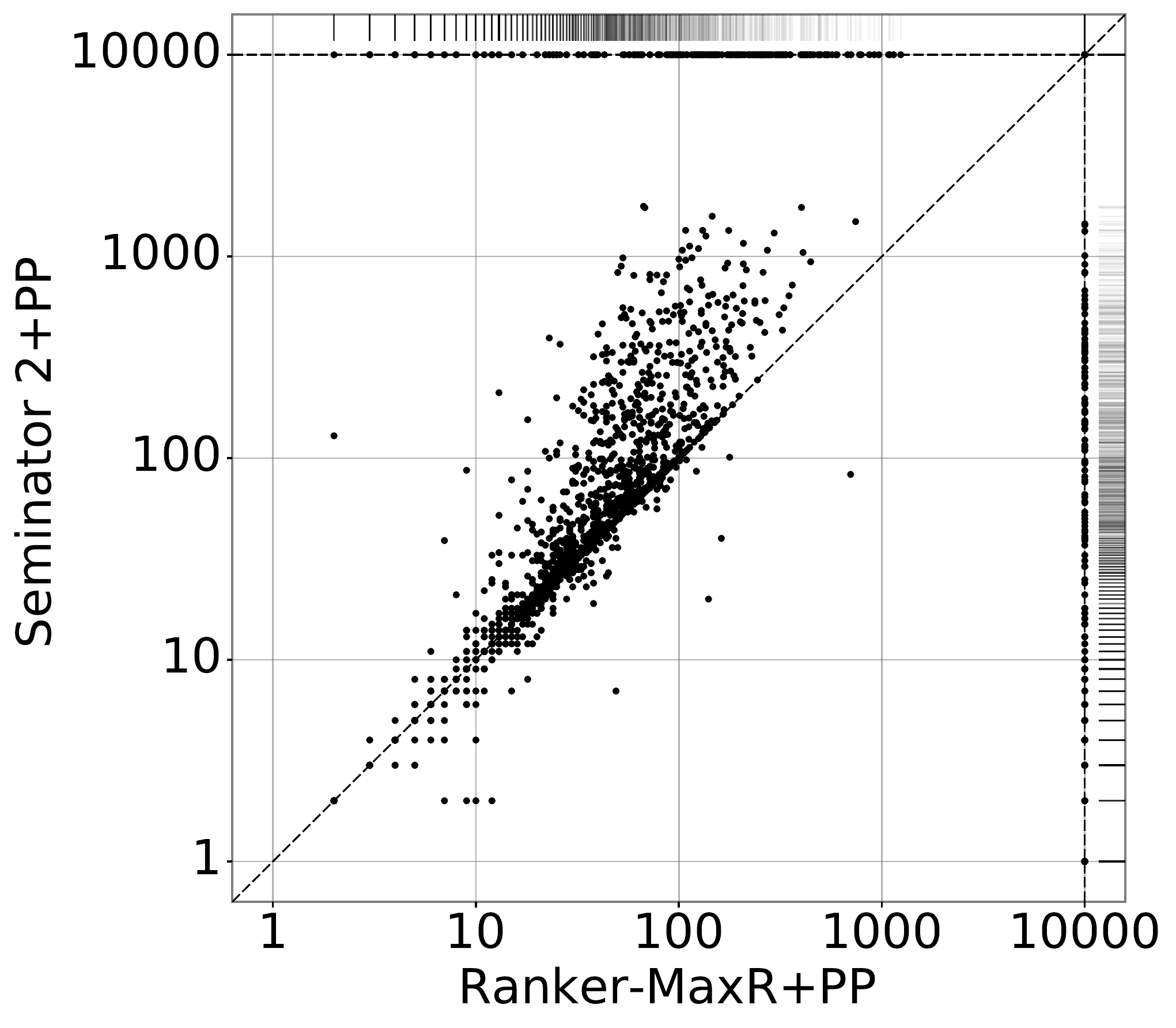}
\vspace*{-2mm}
\caption{\centering\footnotesize $\rankermaxrank$ vs \seminator}
\label{fig:ranker_vs_seminator}
\end{subfigure}
\begin{subfigure}[b]{0.49\linewidth}
	\centering
\includegraphics[width=5.5cm,keepaspectratio]{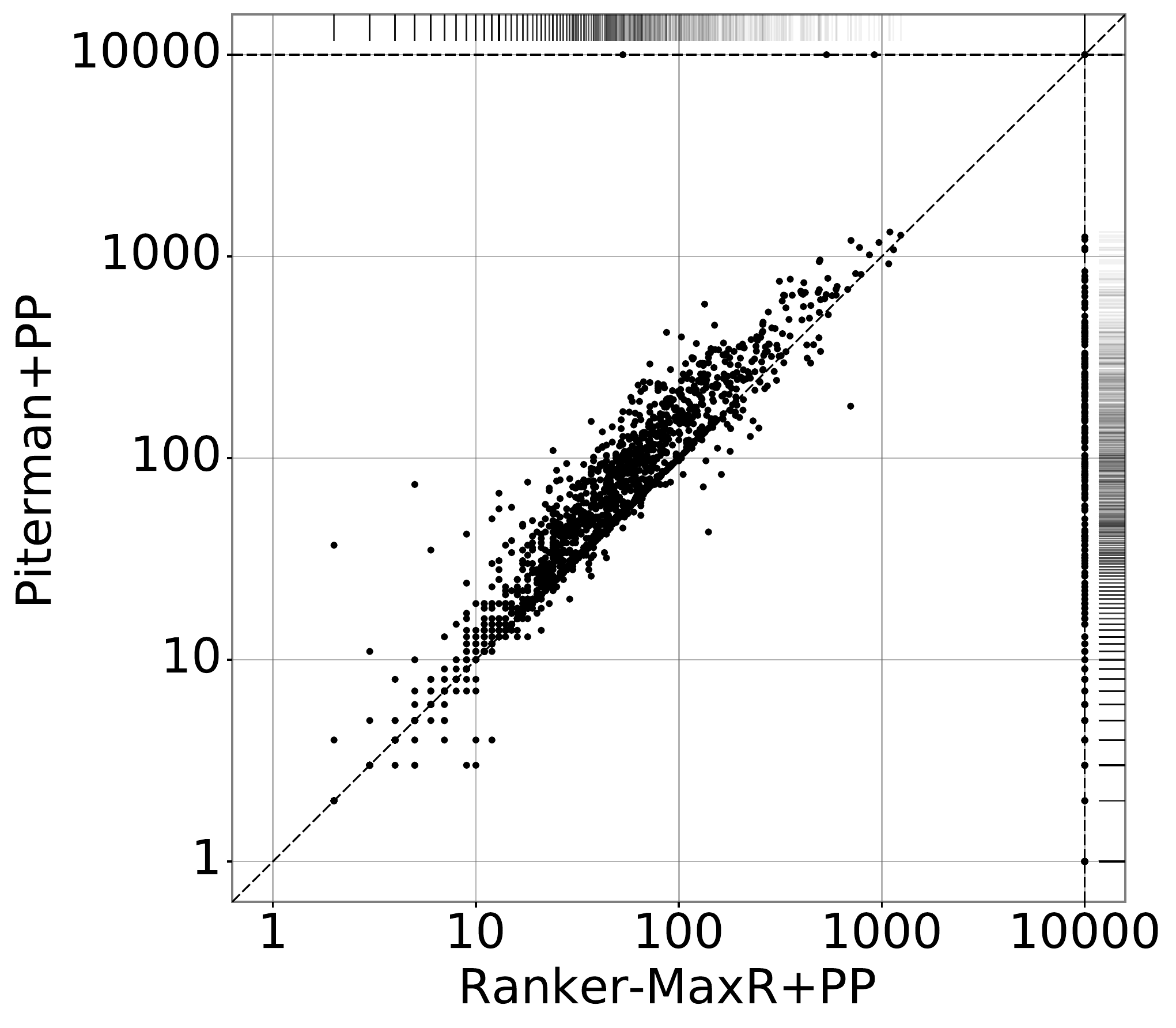}
\vspace*{-2mm}
\caption{\centering\footnotesize $\rankermaxrank$ vs \piterman}
\label{fig:ranker_vs_piterman}
\end{subfigure}
\vspace{4mm}

\begin{subfigure}[b]{0.49\linewidth}
	\centering
\includegraphics[width=5.5cm,keepaspectratio]{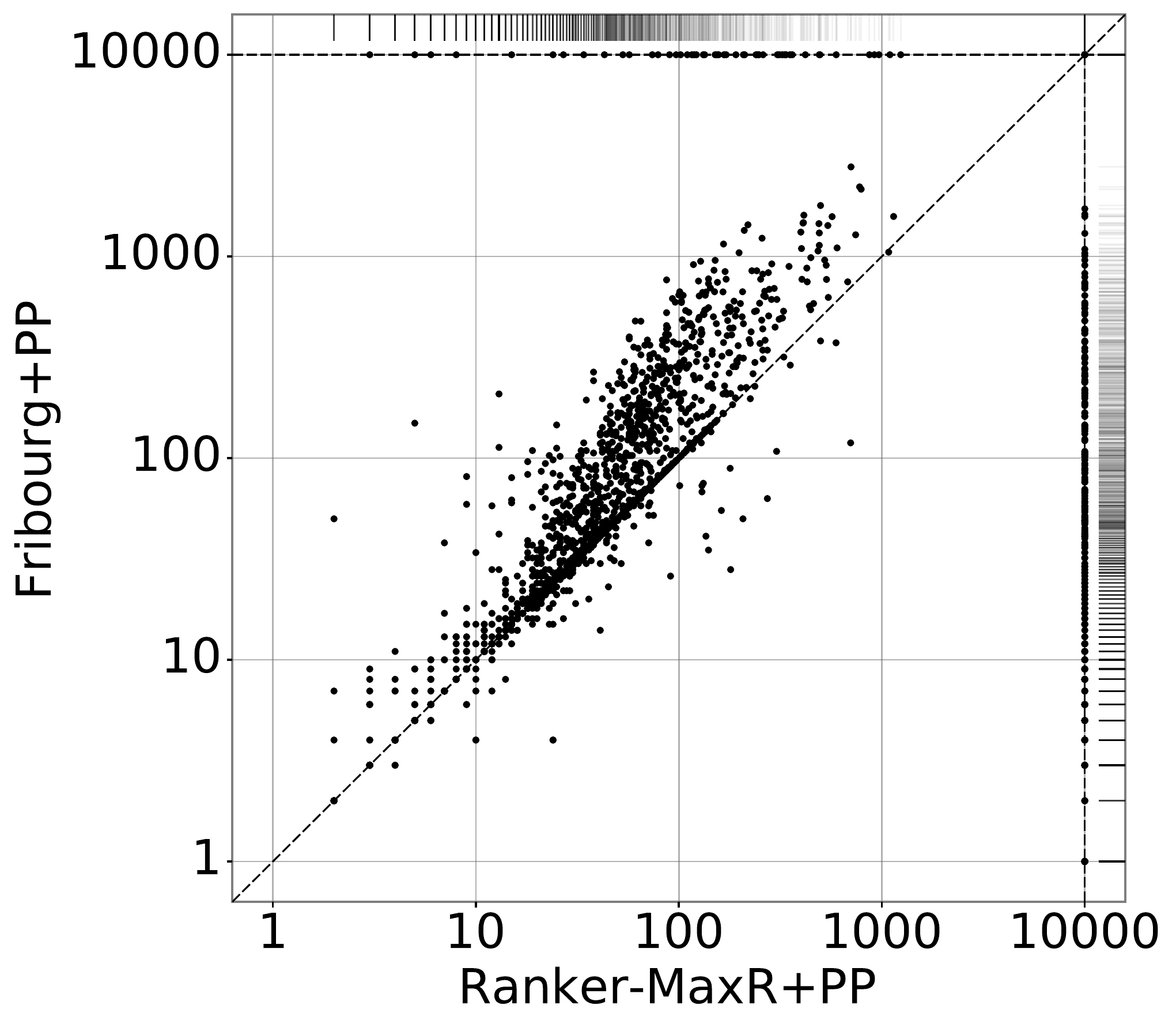}
\vspace*{-2mm}
\caption{\centering\footnotesize $\rankermaxrank$ vs \fribourg}
\label{fig:ranker_vs_fribourg}
\end{subfigure}
\begin{subfigure}[b]{0.49\linewidth}
	\centering
\includegraphics[width=5.5cm,keepaspectratio]{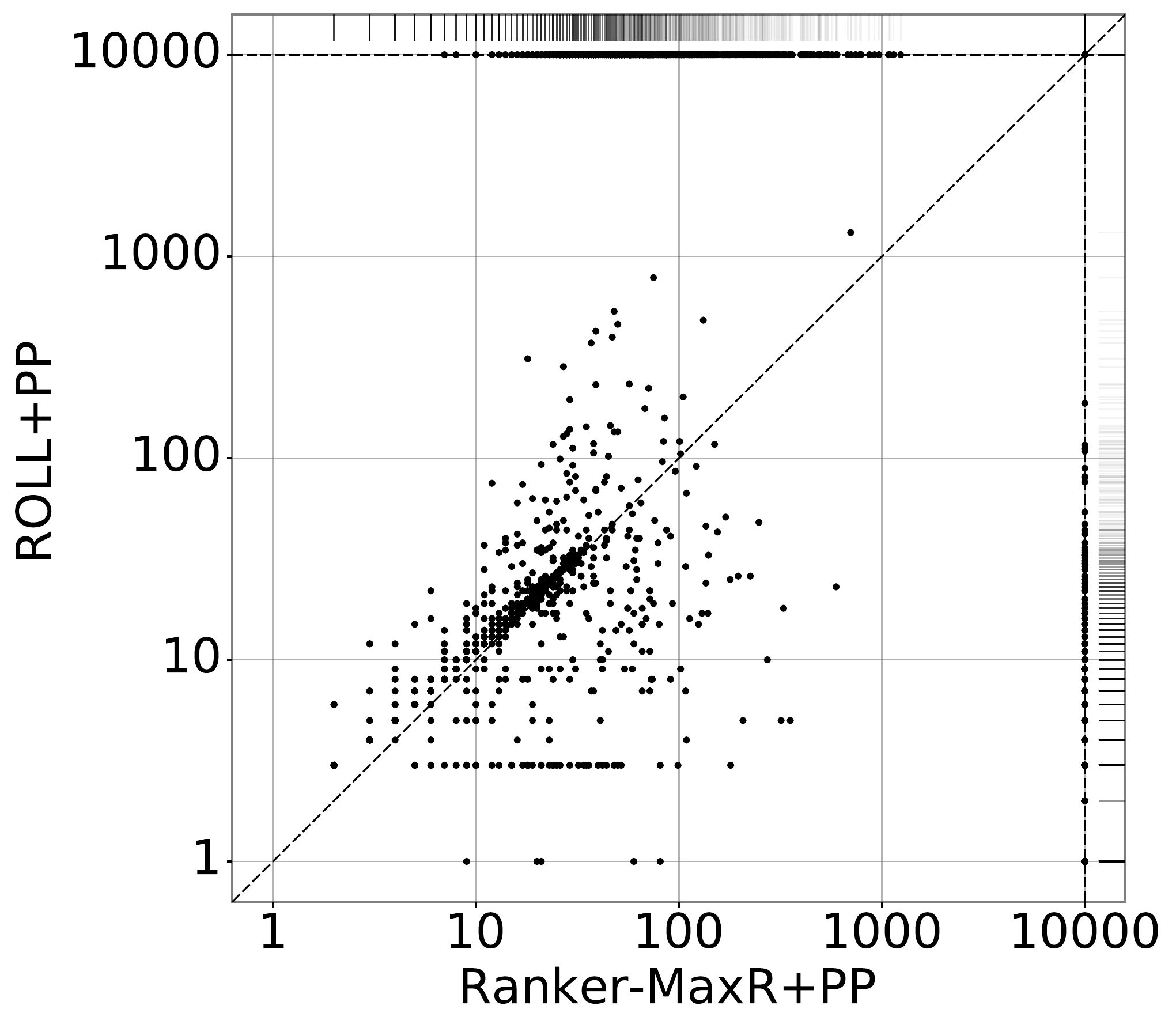}
\vspace*{-2mm}
\caption{\centering\footnotesize $\rankermaxrank$ vs \roll}
\label{fig:ranker_vs_roll}
\end{subfigure}
\caption{\footnotesize Comparison of the sizes of the BAs constructed using our optimized
  rank-based construction and other approaches.
  Timeouts are on the dashed lines.
  % The output of all tools was processed by \autfilt (simplification level
  % \texttt{--high}) from \spot.
  }
\label{fig:comparison_other}
\vspace*{-2mm}
\end{figure}
}

\newcommand{\tableresults}[0]{
\begin{table}[t]
\caption{
  \footnotesize
  Statistics for our experiments.
  The upper part compares different optimizations of the rank-based procedure
  (no postprocessing).
  The lower part compares our approach with other methods (with postprocessing).
  ``BO'' denotes the \algbackoff optimization.
  In the left-hand side of the table,
  the column ``\textbf{med.}'' contains the median,
  ``\textbf{std.\ dev}'' contains the standard deviation, and
  ``\textbf{TO}'' contains the number of timeouts (5\,mins).
  In the right-hand side of the table, we provide the number of cases where our
  tool ($\rankermaxrank$ without postprocessing in the upper part and with
  postprocessing in the lower part) was strictly better (``\textbf{wins}'') or
  worse (``\textbf{losses}'').
  The ``\textbf{(TO)}'' column gives the number of times this was because
  of the timeout of the loser.
  Approaches implemented in \goal are labelled with~\goalpic.
  }
\vspace{-5mm}
\label{tab:summary}
\begin{center}
	\scalebox{0.9}{
	\footnotesize
  \hspace*{-2mm}%
	\newcolumntype{d}[1]{D{.}{.}{#1}}
\begin{tabular}{lrd{4.2}rd{5.2}r@{\hskip 1mm}|rrrr}
\toprule
 \multicolumn{1}{c}{\bf method}                     &
  \multicolumn{1}{c}{\bf ~~max~~}&      \multicolumn{1}{c}{\bf ~~~mean~~~~} &
  \multicolumn{1}{c}{\bf med.} &
  \multicolumn{1}{c}{\bf ~~std.\ dev~~} &   \multicolumn{1}{c|}{\bf TO} &
  \multicolumn{1}{c}{\bf wins} &
  \multicolumn{1}{c}{\bf (TO)} &
  \multicolumn{1}{c}{\bf losses} &
  \multicolumn{1}{c}{\bf (TO)} \\
\midrule
 \emphcell  $\rankermaxrank$ & 319\,119 & 8\,051.58   &      185 & 28\,891.4    &        360 & \centercell{---} & \centercell{---} & \centercell{---} & \centercell{---} \\
 $\rankerrankrestr$          & 330\,608 & 9\,652.67   &      222 & 32\,072.6    &        854 &   1810 &             (495) &     109 &                (1)  \\
 $\algschewerao$~\goalpic             &  67\,780 & 5\,227.3    &      723 & 10\,493.8    &        844 &   2030 &             (486) &       3 &                (2)  \\
\midrule
\emphcell $\rankermaxrank$   &   1\,239 &   61.83 &       32 &   103.18 &        360 & \centercell{---} & \centercell{---} & \centercell{---} & \centercell{---} \\
 $\rankermaxrank$+BO         &   1\,706 &   73.65 &       33 &   126.8  &         17 & \centercell{---} & \centercell{---} & \centercell{---} & \centercell{---} \\
 \piterman~\goalpic                 &   1\,322 &   88.30 &       40 &   142.19 &         12 &   1\,069 &              (3) &     469 &              (351) \\
 \safra~\goalpic                    &   1\,648 &   99.22 &       42 &   170.18 &        158 &   1\,171 &            (117) &     440 &              (319) \\
 \spot                     &   2\,028 &   91.95 &       38 &   158.13 &         13 &    907 &              (6) &     585 &              (353) \\
\fribourg~\goalpic                 &   2\,779 &  113.03 &       36 &   221.91 &         78 &    996 &             (51) &     472 &              (333) \\
 $\ltldstar$                 &   1\,850 &   88.76 &       41 &   144.09 &        128 &   1\,156 &             (99) &     475 &              (331) \\
 $\seminator$                &   1\,772 &   98.63 &       33 &   191.56 &        345 &   1\,081 &            (226) &     428 &              (241) \\
 $\roll$                     &   1\,313 &   21.50 &       11 &    57.67 &       1\,106 &   1\,781 &           (1\,041) &     522 &              (295) \\
\bottomrule
\end{tabular}

}
\end{center}
\vspace*{-7mm}
\end{table}
}

\newcommand{\figrank}[0]{
\begin{figure}[t]
\begin{subfigure}[b]{0.49\linewidth}
	\centering
\includegraphics[width=5.5cm,keepaspectratio]{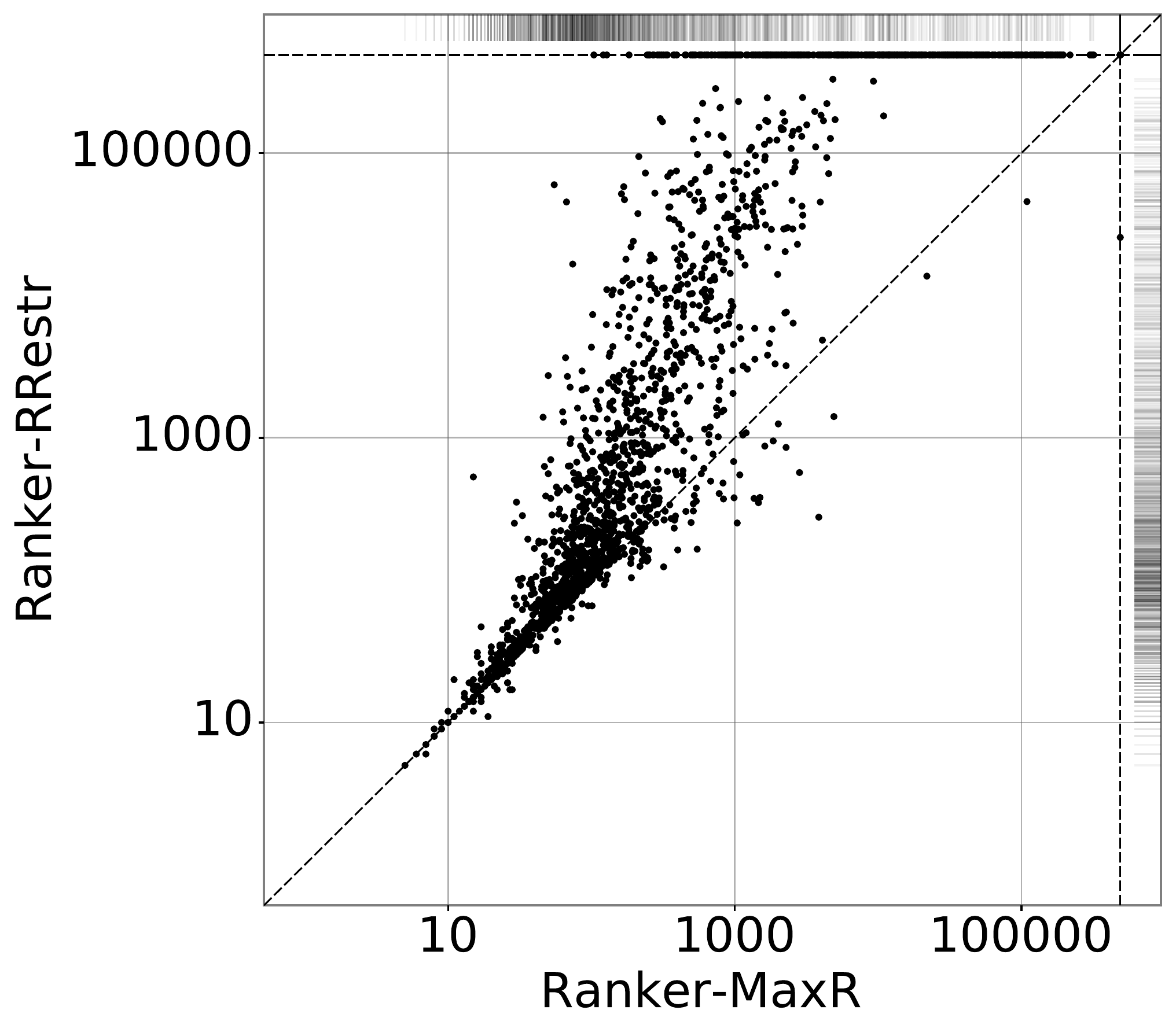}
\caption{\centering\footnotesize $\rankermaxrank$ vs $\rankerrankrestr$}
\label{fig:maxrank_vs_rankrestr}
\end{subfigure}
\begin{subfigure}[b]{0.49\linewidth}
	\centering
\includegraphics[width=5.5cm,keepaspectratio]{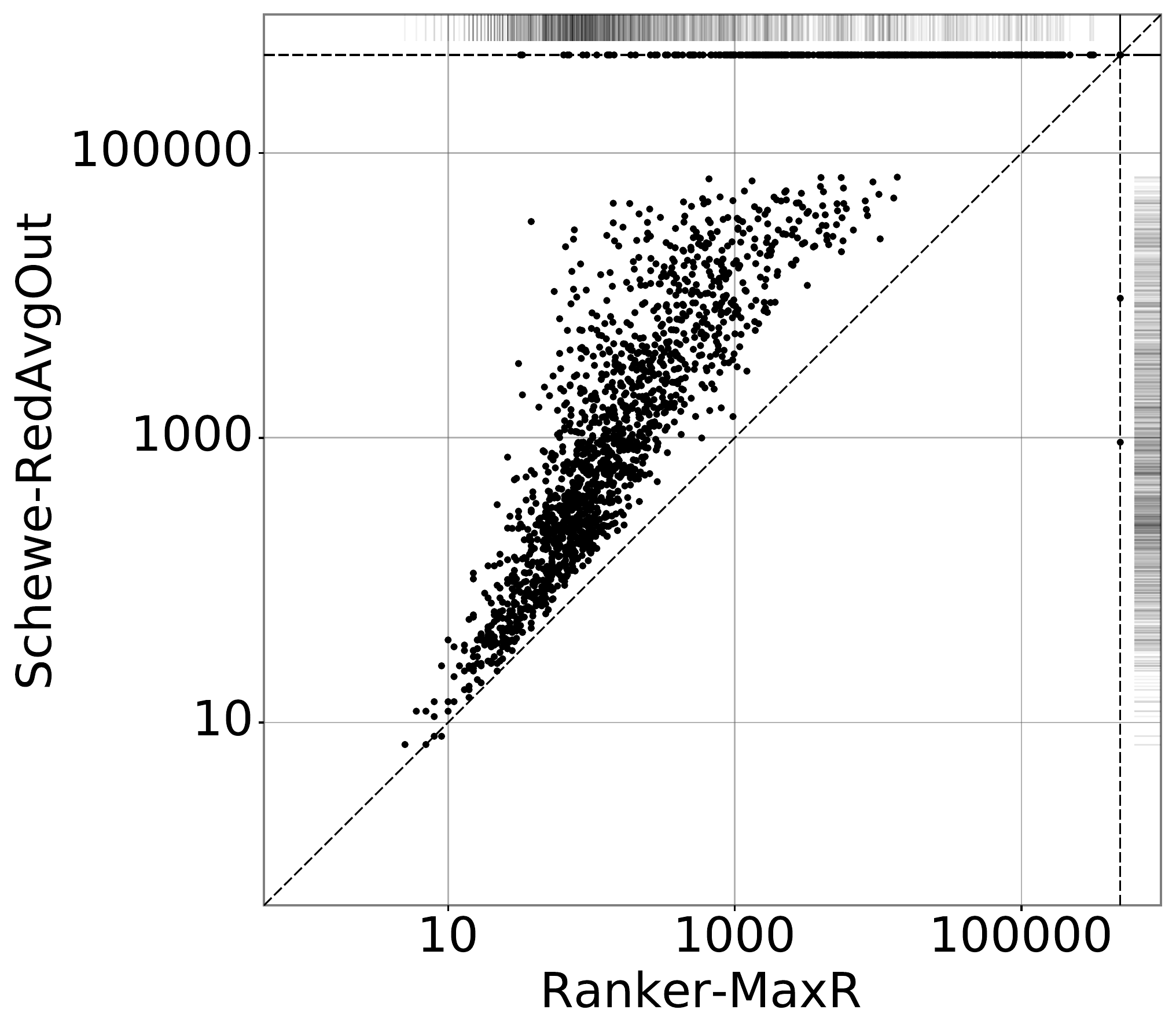}
\caption{\centering\footnotesize $\rankermaxrank$ vs $\algschewerao$}
\label{fig:ranker_vs_schewe}
\end{subfigure}
\caption{\footnotesize Evaluation of the effectiveness of our optimizations on the generated
  state space (axes are logarithmic).
  The horizontal and vertical dashed lines represent timeouts.}
\label{fig:rank-based}
\vspace*{-2mm}
\end{figure}
}

\newcommand{\tabletimes}[0]{
\begin{table}[t]
	\begin{minipage}{0.4\textwidth}
		\caption{Run times of the tools [s]}
    \vspace{-2mm}
		\label{tab:times}
		\scalebox{0.9}{
		\footnotesize
		\begin{tabular}{lrrr}
\toprule
 \centercell{\bf method}                 &
  \centercell{\bf mean} &
  \centercell{\bf ~med.~} &
  \centercell{\bf ~std.\ dev} \\
\midrule
 $\rankermaxrank$     & 10.21 &    0.84  &   28.43  \\
 $\rankermaxrank$+BO  &  9.40 &    3.03  &   16.00  \\
 \piterman~\goalpic            &  7.47 &    6.03  &    8.46  \\
 \safra~\goalpic               & 15.49 &    7.03  &   35.59  \\
 \spot                &  1.07 &    0.02  &    8.94  \\
 \fribourg~\goalpic            & 19.43 &   10.01  &   32.76  \\
 \ltldstar            &  4.17 &    0.06  &   22.19  \\
 \seminator           & 11.41 &    0.37  &   34.97  \\
 \roll                & 42.63 &   14.92  &   67.31  \\
\bottomrule
\end{tabular}

		}
	\end{minipage}
	\hfill
	\begin{minipage}{0.5\textwidth}
		\caption{Wins and losses for $\rankermaxrank$+BO}
    \vspace{-2mm}
		\label{tab:winners_bo}
		\scalebox{0.9}{
		\footnotesize
		\newcolumntype{d}[1]{D{.}{.}{#1}}
\begin{tabular}{lrrrr}
\toprule
 \multicolumn{1}{c}{\bf method}                     &
  \multicolumn{1}{c}{\bf wins} &
  \multicolumn{1}{c}{\bf (TO)} &
  \multicolumn{1}{c}{\bf losses} &
  \multicolumn{1}{c}{\bf (TO)} \\
\midrule
 \piterman~\goalpic      &   1\,160 &              (4) &     112 &               (9) \\
 \safra~\goalpic         &   1\,255 &            (147) &     222 &               (6) \\
 \spot          &      985 &              (8) &     328 &              (12) \\
 \fribourg~\goalpic      &   1\,076 &             (71) &     287 &              (10) \\
 \ltldstar      &   1\,208 &            (118) &     272 &               (7) \\
 \seminator     &   1\,236 &            (333) &     253 &               (5) \\
 \roll          &   1\,923 &         (1\,096) &     360 &               (7) \\
\hline
\end{tabular}

		}
		\vspace*{1.7em}
	\end{minipage}
  \vspace{-5mm}
\end{table}
}

\newcommand{\tablebo}[0]{
\begin{wraptable}[12]{r}{6.0cm}
\vspace{-7mm}
\caption{Wins and losses for $\rankermaxrank$+BO}
\vspace{-3mm}
\hspace*{-3mm}
\newcolumntype{d}[1]{D{.}{.}{#1}}
\begin{tabular}{lrrrr}
\toprule
 \multicolumn{1}{c}{\bf method}                     &
  \multicolumn{1}{c}{\bf wins} &
  \multicolumn{1}{c}{\bf (TO)} &
  \multicolumn{1}{c}{\bf losses} &
  \multicolumn{1}{c}{\bf (TO)} \\
\midrule
 \piterman~\goalpic      &   1\,160 &              (4) &     112 &               (9) \\
 \safra~\goalpic         &   1\,255 &            (147) &     222 &               (6) \\
 \spot          &      985 &              (8) &     328 &              (12) \\
 \fribourg~\goalpic      &   1\,076 &             (71) &     287 &              (10) \\
 \ltldstar      &   1\,208 &            (118) &     272 &               (7) \\
 \seminator     &   1\,236 &            (333) &     253 &               (5) \\
 \roll          &   1\,923 &         (1\,096) &     360 &               (7) \\
\hline
\end{tabular}

\label{tab:winners_bo}
\end{wraptable}
}

%
%*******************************************************************************
\vspace{-0.0mm}
\section{Experimental Evaluation}\label{sec:experiments}
\vspace{-0.0mm}
%*******************************************************************************

%------------------------------------------------------------------------------
\subparagraph*{Used tools and evaluation environment.}
We implemented the optimizations described in the previous sections in a~tool
called \ranker~\cite{ranker}
% \footnote{\ranker is available at
% \url{https://github.com/vhavlena/ba-inclusion}}
in C++ (we tested the correctness of our implementation using \spot's
\texttt{autcross} on all BAs in our benchmark).
%; in many cases it did not finish, but for those it finished, it never reported an error).
We compared our
complementation approach with other state-of-the-art tools, namely,
\goal~\cite{goal} (including the \fribourg plugin~\cite{fribourg}),
\spot~2.9.3~\cite{spot}, \seminator~\cite{seminator},
\ltldstar~0.5.4~\cite{KleinB07}, and \roll~\cite{roll}.
All tools were set to the mode where they output an automaton with the standard
state-based B\"{u}chi acceptance condition.
We note that some of the tools are aimed at
complementing more general flavours of $\omega$-automata, such as \seminator
focusing on generalized transition-based B\"{u}chi automata.
The experimental evaluation was performed on a~64-bit \textsc{GNU/Linux Debian}
workstation with an Intel(R) Xeon(R) CPU E5-2620 running at 2.40\,GHz with
32\,GiB of RAM.
The timeout was set to 5\,minutes.

\figrank

%------------------------------------------------------------------------------
\vspace{-0mm}
\subparagraph*{Dataset.}
The source of our main benchmark are the 11,000 BAs
used in~\cite{tsai-compl}, which were randomly generated using the
Tabakov-Vardi approach~\cite{TabakovV05} over a~two letter alphabet,
starting from 15 states and with various different parameters (see~\cite{tsai-compl} for
more details).
In preprocessing, the automata were reduced using a~combination of
\rabit~\cite{MayrC13} and \spot's \autfilt (using the \texttt{--high}
simplification level) and converted to the HOA format~\cite{BabiakBDKKM0S15}.
From this set, we removed automata that are
\begin{inparaenum}[(i)]
  \item  semi-deterministic,
  \item  inherently weak, or
  \item  unambiguous,
\end{inparaenum}
since for these kinds of automata there exist more efficient complementation
procedures than for unrestricted
BAs~\cite{BlahoudekHSST16,seminator,BoigelotJW01,li-unambigous}.
Moreover, we removed BAs with an empty language or empty language of complement.
We were left with \textbf{2,393} \emph{hard} automata.
(In~\cref{sec:expr_ltl} we also present additional results on a less challenging
benchmark used in~\cite{seminator}, containing BAs obtained by translation from LTL
formulae.)
% (In~\cite{techrep} we also present results on the benchmark from~\cite{seminator}.)
% benchmark used in~\cite{seminator}, containing BAs obtained by translation from LTL
% formulae.)

%------------------------------------------------------------------------------
\vspace{-0mm}
\subparagraph*{Selection of Optimizations.}
We use two settings of \ranker with different optimizations turned on.
Since the \rankrestr and \algmaxrank optimizations are incompatible, the main
difference between the settings is which one of those two they use.
The particular optimizations used in the settings are the following:
\vspace{-0mm}
{\small
\begin{align*}
  \rankermaxrank ={}&
  \algdelay + \succrankred + \ranksimred' + \algmaxrank \\
  \rankerrankrestr ={}&
  \algdelay + \succrankred + \ranksimred' + \rankrestr + \purgedi\\[-12mm]
\end{align*}}

\noindent
(The $\purgedi$ optimization is from~\cite{ChenHL19}.)
Note that the two settings include all optimizations compatible with \algmaxrank
and \rankrestr respectively.
Due to space constraints, we cannot give a~detailed analysis of the effect of
individual optimizations on the size of the obtained complement automaton.
Let us, at least, give a~bird's-eye view.
The biggest effect has \algmaxrank, followed by \algdelay---their use is key
to obtaining a~small state space.
The rest of the optimizations are less effective, but they still remove
a~significant number of states.

%*******************************************************************************
\vspace{-0.0mm}
\subsection{Comparison of Rank-Based Procedures}
\vspace{-0.0mm}
%*******************************************************************************

First, we evaluated how our optimizations reduce the generated state space,
i.e., we compared the sizes of complemented BAs with no postprocessing.
Such a~use case represents applications like testing inclusion or equivalence
of BAs, where postprocessing the output is irrelevant.

\begin{table}[t]
\caption{
  \footnotesize
  Statistics for our experiments.
  The upper part compares different optimizations of the rank-based procedure
  (no postprocessing).
  The lower part compares our approach with other methods (with postprocessing).
  ``BO'' denotes the \algbackoff optimization.
  In the left-hand side of the table,
  the column ``\textbf{med.}'' contains the median,
  ``\textbf{std.\ dev}'' contains the standard deviation, and
  ``\textbf{TO}'' contains the number of timeouts (5\,mins).
  In the right-hand side of the table, we provide the number of cases where our
  tool ($\rankermaxrank$ without postprocessing in the upper part and with
  postprocessing in the lower part) was strictly better (``\textbf{wins}'') or
  worse (``\textbf{losses}'').
  The ``\textbf{(TO)}'' column gives the number of times this was because
  of the timeout of the loser.
  Approaches implemented in \goal are labelled with~\goalpic.
  }
\vspace{-5mm}
\label{tab:summary}
\begin{center}
	\scalebox{0.9}{
	\footnotesize
  \hspace*{-2mm}%
	\newcolumntype{d}[1]{D{.}{.}{#1}}
\begin{tabular}{lrd{4.2}rd{5.2}r@{\hskip 1mm}|rrrr}
\toprule
 \multicolumn{1}{c}{\bf method}                     &
  \multicolumn{1}{c}{\bf ~~max~~}&      \multicolumn{1}{c}{\bf ~~~mean~~~~} &
  \multicolumn{1}{c}{\bf med.} &
  \multicolumn{1}{c}{\bf ~~std.\ dev~~} &   \multicolumn{1}{c|}{\bf TO} &
  \multicolumn{1}{c}{\bf wins} &
  \multicolumn{1}{c}{\bf (TO)} &
  \multicolumn{1}{c}{\bf losses} &
  \multicolumn{1}{c}{\bf (TO)} \\
\midrule
 \emphcell  $\rankermaxrank$ & 319\,119 & 8\,051.58   &      185 & 28\,891.4    &        360 & \centercell{---} & \centercell{---} & \centercell{---} & \centercell{---} \\
 $\rankerrankrestr$          & 330\,608 & 9\,652.67   &      222 & 32\,072.6    &        854 &   1810 &             (495) &     109 &                (1)  \\
 $\algschewerao$~\goalpic             &  67\,780 & 5\,227.3    &      723 & 10\,493.8    &        844 &   2030 &             (486) &       3 &                (2)  \\
\midrule
\emphcell $\rankermaxrank$   &   1\,239 &   61.83 &       32 &   103.18 &        360 & \centercell{---} & \centercell{---} & \centercell{---} & \centercell{---} \\
 $\rankermaxrank$+BO         &   1\,706 &   73.65 &       33 &   126.8  &         17 & \centercell{---} & \centercell{---} & \centercell{---} & \centercell{---} \\
 \piterman~\goalpic                 &   1\,322 &   88.30 &       40 &   142.19 &         12 &   1\,069 &              (3) &     469 &              (351) \\
 \safra~\goalpic                    &   1\,648 &   99.22 &       42 &   170.18 &        158 &   1\,171 &            (117) &     440 &              (319) \\
 \spot                     &   2\,028 &   91.95 &       38 &   158.13 &         13 &    907 &              (6) &     585 &              (353) \\
\fribourg~\goalpic                 &   2\,779 &  113.03 &       36 &   221.91 &         78 &    996 &             (51) &     472 &              (333) \\
 $\ltldstar$                 &   1\,850 &   88.76 &       41 &   144.09 &        128 &   1\,156 &             (99) &     475 &              (331) \\
 $\seminator$                &   1\,772 &   98.63 &       33 &   191.56 &        345 &   1\,081 &            (226) &     428 &              (241) \\
 $\roll$                     &   1\,313 &   21.50 &       11 &    57.67 &       1\,106 &   1\,781 &           (1\,041) &     522 &              (295) \\
\bottomrule
\end{tabular}

}
\end{center}
\vspace*{-7mm}
\end{table}

More precisely, we first compared the sizes of automata produced by
our settings $\rankermaxrank$ and $\rankerrankrestr$ to see which of them behaves
better (cf.\ \cref{fig:maxrank_vs_rankrestr})
and then we compared $\rankermaxrank$, which had better results, with the
$\algschewerao$ procedure implemented in \goal (parameters \texttt{-m rank
-tr -ro}).
Scatter plots of the results are given in \cref{fig:ranker_vs_schewe} and
summarizing statistics in the upper part of \cref{tab:summary}.

We note that although $\rankermaxrank$ produces in the vast majority of cases
(1,810) smaller automata than $\rankerrankrestr$, in a~few cases
(109) $\rankerrankrestr$ still outputs a~smaller result (in 1 case this is due
to the timeout of $\rankermaxrank$).
The comparison with $\algschewerao$ shows that our optimizations indeed have
a~profound effect on the size of the generated state space.
Although the mean and maximum size of complements produced by
$\rankermaxrank$ and $\rankerrankrestr$ are larger than those of
$\algschewerao$, this is because for cases where the complement would be large,
the run of $\algschewerao$ in \goal timeouted before it could produce a~result.
Therefore, the median is a~more meaningful indicator, and it is significantly
(3--4$\times$) lower for both $\rankermaxrank$ and $\rankerrankrestr$.
%
% Although the mean size of the outputs does not differ too much, the median
% (which is a~more meaningful indicator because of the much higher number of
% timeouts) of $\algschewerao$ is almost double that of $\rankermaxrank$.
% We note that although it may seem that $\rankermaxrank$ and $\rankerrankrestr$ produced bigger
% automata (319,119 and 330,608 states respectively compared to 67,780 states of
% $\algschewerao$), this is because the run of $\algschewerao$ in \goal timeouted
% for the cases where the complement was bigger.

\figother

%*******************************************************************************
\vspace{-0.0mm}
\subsection{Comparison with Other Approaches}
\vspace{-0.0mm}
%*******************************************************************************

Further, we evaluated the complements produced by $\rankermaxrank$ and other
approaches.
In this setting, we focused on the size of the output BA
\emph{after} postprocessing (we, again, used \autfilt with simplification
\texttt{--high}; we denote this using ``+PP'').
We evaluated the following algorithms:
\safra~\cite{safra1988complexity},
%implemented in \goal (parameter \texttt{-m safra}),
its optimization
\piterman~\cite{piterman2006nondeterministic}
%implemented in \goal (parameter \texttt{-m piterman})
the optimization implemented in \ltldstar~\cite{KleinB07},
\fribourg~\cite{fribourg}, % implemented as a~plugin of \goal,
\spot (Redziejowski's algorithm~\cite{Redziejowski12}),
\roll's learning-based algorithm~\cite{li2018learning}, and
a~semideterminization-based algorithm~\cite{BlahoudekHSST16} in
\seminator.

In \cref{fig:comparison_other}, we give scatter plots of selected comparisons;
we omitted the results for \safra,
\spot, and $\ltldstar$, which on average performed slightly worse than \piterman.
We give summarizing statistics in the lower part of \cref{tab:summary} and
the run times in \cref{tab:times}.

Let us now discuss the data in the lower part of \cref{tab:summary}.
In the left-hand side, we can see that the mean and median size of BAs obtained
by $\rankermaxrank$ are both the lowest with the exception of \roll.
\roll implements a~learning-based approach, which means that it works on the
level of the \emph{language} of the input BA instead of the \emph{structure}.
Therefore, it can often find a~much smaller automaton than other approaches.
Its practical time complexity, however, seems to grow much faster with the number of
states of the output BA than other approaches
(cf.~\cref{tab:times}).
% (\roll had by far the largest mean and median run time, as shown in \cref{tab:times}).
$\rankermaxrank$ by itself had more timeouts than other approaches, but when
used with the \algbackoff strategy, is on par with \piterman and \spot.

In the right-hand side of \cref{tab:summary}, we give the numbers of times where
$\rankermaxrank$ gave strictly smaller and strictly larger outputs respectively.
Here, we can see that the output of $\rankermaxrank$ is often at least as small
as the output of the other method (this is not in the table, but can be
computed as $2,393 - \textbf{losses}$; the losses were caused mostly by timeouts;
results with the \algbackoff strategy would increase the
number even more) and often a~strictly smaller one (the \textbf{wins} column).
When comparing $\rankermaxrank$ with \emph{the best result of any other tool},
it obtained a~\emph{strictly smaller} BA in 539 cases (22.5\,\%) and
a~BA \emph{at least as small} as the best result of any other tool in 1,518
cases~(63.4\,\%).
Lastly, we note that there were four BAs in the benchmark that \emph{no
tool} could complement and one BA that \emph{only} $\rankermaxrank$ was able to
complement;
%(namely, \texttt{new-s-15-r-1.00-f-0.30--72-of-100.ba-red.hoa} with a~919-state-large complement).
there was no such a~case for any other tool.

Let us now focus on the run times of the tools in \cref{tab:times}.
\goal and \roll are implemented in Java, which
adds a~significant overhead to the run time (e.g., the fastest run time of \goal
was 3.15\,s; it is hard to
predict how their performance would change if they were reimplemented in
a~faster language); the other approaches are implemented in C++.
% The mean time of $\rankermaxrank$ (in particular with the \algbackoff strategy)
% is the second (after \spot) and the median is the third (after \spot and
% \seminator).

%------------------------------------------------------------------------------
\vspace{-0mm}
\subparagraph*{\algbackoff.}
Our \algbackoff setting in the experiments used the set of
constraints
%
% \begin{equation*}
$
\{(\StateSize_1 = 9, \RankMax_1 = 5),
  (\StateSize_2 = 8, \RankMax_2 = 6)\}
$
% \end{equation*}
%
and \piterman as the surrogate algorithm.
The \algbackoff strategy was executed 873 times and managed to decrease the
number of timeouts of $\rankermaxrank$ from 360 to 17 (row $\rankermaxrank$+BO
in \cref{tab:summary}).

%------------------------------------------------------------------------------
%\tablebo
\subparagraph*{Discussion.}
The results of our experiments show that our optimizations are key to making
rank-based complementation competitive to other approaches in practice.
Furthermore, with the optimizations, the obtained procedure in the majority of
cases produces a~BA at least as small as a~BA produced by any other approach,
and in a~large number of cases \emph{the smallest} BA produced by any existing
approach.
We emphasize the usefulness of the \algbackoff heuristic:
as there is no clear ``best'' complementation algorithm---different techniques
having different strengths and weaknesses---knowing which technique to use for
an input automaton is important in practice.
In \cref{tab:winners_bo}, we give a modification of the right-hand side of
\cref{tab:summary} giving wins and losses for $\rankermaxrank$+BO.
It seems that the combination of these two completely different algorithms
yields a quite strong competitor.

\begin{table}[t]
	\begin{minipage}{0.4\textwidth}
		\caption{Run times of the tools [s]}
    \vspace{-2mm}
		\label{tab:times}
		\scalebox{0.9}{
		\footnotesize
		\begin{tabular}{lrrr}
\toprule
 \centercell{\bf method}                 &
  \centercell{\bf mean} &
  \centercell{\bf ~med.~} &
  \centercell{\bf ~std.\ dev} \\
\midrule
 $\rankermaxrank$     & 10.21 &    0.84  &   28.43  \\
 $\rankermaxrank$+BO  &  9.40 &    3.03  &   16.00  \\
 \piterman~\goalpic            &  7.47 &    6.03  &    8.46  \\
 \safra~\goalpic               & 15.49 &    7.03  &   35.59  \\
 \spot                &  1.07 &    0.02  &    8.94  \\
 \fribourg~\goalpic            & 19.43 &   10.01  &   32.76  \\
 \ltldstar            &  4.17 &    0.06  &   22.19  \\
 \seminator           & 11.41 &    0.37  &   34.97  \\
 \roll                & 42.63 &   14.92  &   67.31  \\
\bottomrule
\end{tabular}

		}
	\end{minipage}
	\hfill
	\begin{minipage}{0.5\textwidth}
		\caption{Wins and losses for $\rankermaxrank$+BO}
    \vspace{-2mm}
		\label{tab:winners_bo}
		\scalebox{0.9}{
		\footnotesize
		\newcolumntype{d}[1]{D{.}{.}{#1}}
\begin{tabular}{lrrrr}
\toprule
 \multicolumn{1}{c}{\bf method}                     &
  \multicolumn{1}{c}{\bf wins} &
  \multicolumn{1}{c}{\bf (TO)} &
  \multicolumn{1}{c}{\bf losses} &
  \multicolumn{1}{c}{\bf (TO)} \\
\midrule
 \piterman~\goalpic      &   1\,160 &              (4) &     112 &               (9) \\
 \safra~\goalpic         &   1\,255 &            (147) &     222 &               (6) \\
 \spot          &      985 &              (8) &     328 &              (12) \\
 \fribourg~\goalpic      &   1\,076 &             (71) &     287 &              (10) \\
 \ltldstar      &   1\,208 &            (118) &     272 &               (7) \\
 \seminator     &   1\,236 &            (333) &     253 &               (5) \\
 \roll          &   1\,923 &         (1\,096) &     360 &               (7) \\
\hline
\end{tabular}

		}
		\vspace*{1.7em}
	\end{minipage}
  \vspace{-5mm}
\end{table}

%%%%%%%%%%%%%%%%%%%%%%%%%%%%%%%%%%%%%%%%%%%%%%%%%%%%%%%%%%%%%%%%%%%%%%%%%%%%%%%%
\vspace{-0.0mm}
\section{Related Work}\label{sec:related}
\vspace{-0.0mm}
%%%%%%%%%%%%%%%%%%%%%%%%%%%%%%%%%%%%%%%%%%%%%%%%%%%%%%%%%%%%%%%%%%%%%%%%%%%%%%%%

The problem of BA complementation has attracted researchers since B\"{u}chi's seminal
work~\cite{buchi1962decision}.
Since then, there have appeared several directions of
BA complementation approaches. \emph{Ramsey-based complementation} using
B\"{u}chi's original argument, decomposing the language accepted by an automaton
into a~finite number of equivalence classes, was later improved
in~\cite{breuers-improved-ramsey}.
\emph{Determinization-based complementation} was
introduced by Safra in~\cite{safra1988complexity}, later improved by
Piterman in~\cite{piterman2006nondeterministic}.
Determinization-based approaches convert an input BA
into an equivalent intermediate deterministic automaton with different accepting
condition (e.g. Rabin automaton) that can be easily complemented.
The result is then converted back into a~BA (often for the price of some blow-up).
\emph{Slice-based complementation} uses a~reduced abstraction on a run tree to
track the acceptance
condition~\cite{vardi2008automata,kahler2008complementation}.
\emph{A~learning-based approach} was presented in~\cite{li2018learning,roll}.
A~novel
optimal complementation algorithm by Allred and Utes-Nitsche was
presented in~\cite{fribourg}.
There are also specific approaches for
complementation of special types of BAs, e.g., deterministic~\cite{Kurshan87},
semi-deterministic~\cite{BlahoudekHSST16}, or unambiguous~\cite{li-unambigous}.
\emph{Semi-determinization based complementation} then uses a~conversion of
a~standard BA to a~semi-deterministic version~\cite{CourcoubetisY88} followed by
its complementation~\cite{seminator}.

\emph{Rank-based complementation}, studied
in~\cite{KupfermanV01,GurumurthyKSV03,FriedgutKV06,Schewe09,KarmarkarC09},
extends the subset construction for determinizing finite automata with
additional information kept in each macrostate to track the acceptance condition
of all runs of the input automaton.
We have described the refinement of the basic procedure from~\cite{KupfermanV01}
towards~\cite{FriedgutKV06} and~\cite{Schewe09} in \cref{sec:complement}.
The work in~\cite{GurumurthyKSV03} contains optimizations of an alternative
(sub-optimal) rank-based construction from~\cite{KupfermanV01} that goes through
\emph{alternating B\"{u}chi automata}.
Furthermore, the work in~\cite{KarmarkarC09} proposes an optimization of
\algschewe that in some cases produces smaller automata
(the construction is not compatible with our optimizations).
Rank-based construction can be optimized using simulation relations as shown
in~\cite{ChenHL19}.
Here the direct and delayed simulation relations can be used
to prune macrostates that are redundant for accepting a~word or to saturate
macrostates with simulation-smaller states.

% Since BA complementation is interesting from both practical and theoretical
% side, there have emerged a~lot of tools dealing with complementation of automata over
% infinite words. The tool \goal~\cite{goal} implements various algorithms related
% to $\omega$-languages, including various approaches for complementation of
% B\"{u}chi automata. Fribourg plugin for GOAL implements a complementation
% algorithm by Allred and Utes-Nitsche~\cite{freiburg}. The tool Seminator~2
% complements generalized transition-based BAs using
% semi-determinization~\cite{seminator}. The Spot library provides support for
% handling $\omega$-automata including the complementation~\cite{spot} and the ROLL
% tool creates the complemented automaton using learning~\cite{roll}.

%%%%%%%%%%%%%%%%%%%%%%%%%%%%%%%%%%%%%%%%%%%%%%%%%%%%%%%%%%%%%%%%%%%%%%%%%%%%%%%%
\vspace{-0.0mm}
\section{Conclusion and Future Work}\label{sec:label}
\vspace{-0.0mm}
%%%%%%%%%%%%%%%%%%%%%%%%%%%%%%%%%%%%%%%%%%%%%%%%%%%%%%%%%%%%%%%%%%%%%%%%%%%%%%%%

We developed a series of optimizations reducing the state space generated in
rank-based BA complementation.
Our experimental evaluation shows that our approach is competitive with other
state-of-the-art complementation techniques and often outperforms them.

% There are several possible directions for our future work.
% There are still ways to improve our procedure.
There are several possible directions for further improving our procedure.
In particular, we have ideas about refining \ranksimred to obtain an even larger
reduction.
Furthermore,
\algdelay can be further refined by a~smarter choice of \emph{when} to perform
the transition from the waiting to the tight part of the BA.
Currently, this is
done when a~cycle in the waiting part is closed, which does not need to be the
best choice---we could utilize information about the number of successors of
selected states to choose a better point.

In this paper, in order to keep the presentation easier to follow, we focused
on BAs with one state-based acceptance condition.
Our optimizations can, however, be extended to \emph{generalized BAs}
and also to \emph{transition-based generalized BAs} (TGBAs)
with a~modification of \succrankred (final states cannot be considered any more).
It is an open question whether the richer structure of TGBAs brings other
opportunities for reductions.
We also plan to extend our approach to efficient (TG)BA language inclusion
checking, where even more reduction (in the style of~\cite{AbdullaCHMV10}) of
the state space is possible.

%%%%%%%%%%%%%%%%%%%%%%%%%%%%%%%%%%%% BIBLIO %%%%%%%%%%%%%%%%%%%%%%%%%%%%%%%%%%
\bibliographystyle{plainurl}
\bibliography{literature}
%%%%%%%%%%%%%%%%%%%%%%%%%%%%%%%%%%%% BIBLIO %%%%%%%%%%%%%%%%%%%%%%%%%%%%%%%%%%

\newpage
\appendix
%%%%%%%%%%%%%%%%%%%%%%%%%%%%%%%%%%%%%%%%%%%%%%%%%%%%%%%%%%%%%%%%%%%%%%%%%%%%%%%%
\vspace{-0.0mm}
\section{Omitting Max-Rank Runs in $\algschewerao$}\label{sec:schewe_bug_example}
\vspace{-0.0mm}
%%%%%%%%%%%%%%%%%%%%%%%%%%%%%%%%%%%%%%%%%%%%%%%%%%%%%%%%%%%%%%%%%%%%%%%%%%%%%%%%

\ol{THIS NEEDS TO BE CHECKED}

In this appendix, we show an example where the algorithm $\algschewerao$
from~\cite[Section~4]{Schewe09} removes some max-rank runs from the complement
and, therefore, cannot be used with our optimizations.

Let us now recall the construction of $\algschewerao$
from~\cite[Section~4]{Schewe09}.
Let $\aut = (Q, \delta, I, F)$.
For $S \subseteq Q$,
we say that an~$S$-tight ranking~$f$ with rank~$r$ is \emph{maximal wrt.}~$S$ if
it maps every $q \in S \cap F$ to $r-1$, exactly one state of~$S$ to every odd
number $o < r$ and all remaining states of~$S$ to~$r$.
Let $\M_S$ be the set of tight rankings $\M_S = \{f \in
\cT \mid f \text{ is maximal wrt.\ } S\}$.

Then, $\dut = \algschewerao(\aut)$ is defined as the BA $\dut = (Q_1 \cup Q_2,
\gamma, I', F')$ where $Q_1, Q_2, I', F', \delta_1, \delta_2$, and~$\delta_3$ are the same
as in $\algschewe(\aut)$ and $\gamma = \delta_1 \cup \gamma_2 \cup \gamma_3 \cup
\gamma_4$ where (we emphasize the main differences from
\algmaxrank in \textcolor{red}{red})
\begin{itemize}
  \item $\gamma_2(S, a) := \textcolor{red}{\{(S', O, g, i) \in Q_2 \mid (S', O,
    g, i) \in \delta_2 \land g \in \M_{S'}\}}$,
	\item $\gamma_3(\sofi, a)$:
    Let~$\fpmax = \maxrankof{\sofi, a}$.
    Then, we set
    \vspace{-1mm}
    \begin{itemize}
      \setlength{\itemsep}{0mm}
      \item  $\gamma_3(\sofi, a) := \{(S', O', \fpmax, i')\}$ when $(S', O',
        \fpmax, i') \in \delta_3(\sofi, a)$ (i.e., if $\fpmax$ is tight) and
      \item  $\gamma_3(\sofi, a) := \emptyset$ otherwise.
    \end{itemize}
  \item $\gamma_4(\sofi, a)$: Let
    $\gamma_3(\sofi, a) = \{(S', P', h', i')\}$ and let
    \vspace{-1mm}
    \begin{itemize}
      \setlength{\itemsep}{0mm}
      \item  $f' = \variant{h'}{\{ u\mapsto h'(u)-1 \mid \textcolor{red}{u\in
        P'} \}}$ and
      \item  $O' = P'\cap f'^{-1}(i')$.
    \end{itemize}
    \vspace{-1mm}
    Then, if $i' \neq 0$, the algorithm sets $\gamma_4(\sofi, a) := \{\sofiprime\}$, else
    it sets \mbox{$\gamma_4(\sofi, a) := \emptyset$}.
\end{itemize}

% Recall that $\gamma_4$ was defined as follows (the other components are the same
% as in $\algmaxrank(\aut)$ from \cref{sec:maxrank}; the difference is emphasized
% in \textcolor{red}{red}):
%
% \begin{itemize}
%   \item  $(S', O'', g', i') \in \gamma_4(\sofi, a)$ iff
%     $(S', O', g, i') \in \eta_3(\sofi, a)$, $O'' = \emptyset$, $i' \neq 0 \lor
%     O' = \emptyset$, and $g'(q) = g(q) - 1$ for all $q \in O'$ and $g'(q) =
%     g(q)$ otherwise.
%   \ol{blah}
% \end{itemize}

% \begin{figure}[t]
%   \centering
%   \caption{A B\"{u}chi automaton over alphabet $\Sigma =\{a\}$ accepting the empty language.}
%   \label{fig:ba-ro-ex}
% \end{figure}

Next, consider the BA $\aut = (Q, \delta, I, F)$ over the single-symbol alphabet
$\Sigma = \{a\}$ given in the following figure:

\vspace{-2mm}
\begin{center}
  \begin{tikzpicture}[->,>=stealth',shorten >=0pt,auto,node distance=1.5cm,
                      scale = 0.8,transform shape,initial text={}]
    \tikzstyle{every state}=[inner sep=3pt,minimum size=5pt]

    \node[state,initial,label={[color=blue,label distance=-1mm]-45:5}] (p1) {$p_1$};
    \node[state,accepting,label={[color=blue,label distance=-1mm]-45:4}] (p2) [right of=p1] {$p_2$};
    \node[state,label={[color=blue,label distance=-1mm]-45:3}] (p3) [right of=p2] {$p_3$};
    \node[state,accepting,label={[color=blue,label distance=-1mm]-45:2}] (p4) [right of=p3] {$p_4$};
    \node[state,label={[color=blue,label distance=-1mm]-45:1}] (p5) [right of=p4] {$p_5$};

    \path (p1) edge[loop above]  node {$a$} (p1)
          (p1) edge  node {$a$} (p2)
          (p2) edge  node {$a$} (p3)
          (p3) edge[loop above]  node {$a$} (p3)
          (p3) edge  node {$a$} (p4)
          (p4) edge  node {$a$} (p5)
          (p5) edge[loop above]  node {$a$} (p5);

    \node[state,initial,label={[color=blue,label distance=-1mm]-45:5}] (q1) [below of=p1, node distance=20mm] {$q_1$};
    \node[state,accepting,label={[color=blue,label distance=-1mm]-45:4}] (q2) [right of=q1] {$q_2$};
    \node[state,label={[color=blue,label distance=-1mm]-45:3}] (q3) [right of=q2] {$q_3$};
    \node[state,accepting,label={[color=blue,label distance=-1mm]-45:2}] (q4) [right of=q3] {$q_4$};
    \node[state,label={[color=blue,label distance=-1mm]-45:1}] (q5) [right of=q4] {$q_5$};

    \path (q1) edge[loop above]  node {$a$} (q1)
          (q1) edge  node {$a$} (q2)
          (q2) edge  node {$a$} (q3)
          (q3) edge[loop above]  node {$a$} (q3)
          (q3) edge  node {$a$} (q4)
          (q4) edge  node {$a$} (q5)
          (q5) edge[loop above]  node {$a$} (q5);

  \end{tikzpicture}
\end{center}
\vspace{-2mm}
\noindent
Note that $\langof \aut = \emptyset$ and that $\Sigma^\omega = \{a^\omega\}$.
The highest rank assigned to a~node in the run DAG of~$\aut$ over~$a^\omega$
is~5 (ranks of nodes with a~given state in the run DAG are given in the figure
above in \textcolor{blue}{blue}).
It follows that a~super-tight (and therefore also a~max-rank) run of
$\algschewe(\aut)$ on $a^\omega$ will also have the rank~5.

% this can
% be obtained from the fact that both $p_i \dirsimby q_i$ and $q_i \dirsimby p_i$
% for each $i\in\{1,\ldots,5\}$).

Let us now consider a~max-rank run and $\algschewerao$~\cite[Section~4]{Schewe09}.
The only interesting macrostates from~$\gamma_2$ have the following form:
$p_1$~or~$q_1$ have the rank~5, $p_3$~or~$q_3$ have the rank~3, and,
finally, $p_5$~or~$q_5$ have the rank~1.
Other cases are not interesting because they cannot lead to an acceptance
of~$a^\omega$.
Also recall (from the definition of $\M_S$) that
only one state has the rank~3 and only one has the rank~1
(the remaining accepting states have the rank~4 and nonaccepting states have
the rank~5).

Now let us focus on the macrostate~$M = (Q, \emptyset, f, 0)$ from~$\gamma_2$ where
$f$ is given in \textcolor{red}{red} in the figure below:
\vspace{-2mm}
\begin{center}
  \begin{tikzpicture}[->,>=stealth',shorten >=0pt,auto,node distance=1.5cm,
                      scale = 0.8,transform shape,initial text={}]
    \tikzstyle{every state}=[inner sep=3pt,minimum size=5pt]

    \node[state,initial,label={[color=red,label distance=-1mm]-45:5}] (p1) {$p_1$};
    \node[state,accepting,label={[color=red,label distance=-1mm]-45:4}] (p2) [right of=p1] {$p_2$};
    \node[state,label={[color=red,label distance=-1mm]-45:3}] (p3) [right of=p2] {$p_3$};
    \node[state,accepting,label={[color=red,label distance=-1mm]-45:4}] (p4) [right of=p3] {$p_4$};
    \node[state,label={[color=red,label distance=-1mm]-45:5}] (p5) [right of=p4] {$p_5$};

    \path (p1) edge[loop above]  node {$a$} (p1)
          (p1) edge  node {$a$} (p2)
          (p2) edge  node {$a$} (p3)
          (p3) edge[loop above]  node {$a$} (p3)
          (p3) edge  node {$a$} (p4)
          (p4) edge  node {$a$} (p5)
          (p5) edge[loop above]  node {$a$} (p5);

    \node[state,initial,label={[color=red,label distance=-1mm]-45:5}] (q1) [below of=p1, node distance=20mm] {$q_1$};
    \node[state,accepting,label={[color=red,label distance=-1mm]-45:4}] (q2) [right of=q1] {$q_2$};
    \node[state,label={[color=red,label distance=-1mm]-45:5}] (q3) [right of=q2] {$q_3$};
    \node[state,accepting,label={[color=red,label distance=-1mm]-45:4}] (q4) [right of=q3] {$q_4$};
    \node[state,label={[color=red,label distance=-1mm]-45:1}] (q5) [right of=q4] {$q_5$};

    \path (q1) edge[loop above]  node {$a$} (q1)
          (q1) edge  node {$a$} (q2)
          (q2) edge  node {$a$} (q3)
          (q3) edge[loop above]  node {$a$} (q3)
          (q3) edge  node {$a$} (q4)
          (q4) edge  node {$a$} (q5)
          (q5) edge[loop above]  node {$a$} (q5);

  \end{tikzpicture}
\end{center}
\vspace{-3mm}

% ranking function corresponds to $f = \{ p_1 \mapsto 5, q_1 \mapsto 5, p_3 \mapsto 3,
% q_3 \mapsto 5, p_5 \mapsto 5, q_5 \mapsto 1 \} \cup \{ q_i \mapsto 4, p_i
% \mapsto 4 \mid i \in \{2,4\} \}$. In particular
% %
% $$
%   \M = (\{ p_i, q_i \mid i \in \{1, \dots, 5\} \}, \emptyset, f, 0).
% $$
%
\noindent
(Other cases with a different distribution of ranks 3 and 1 can be proved
analogously.)
Since~$q_3$ has at the beginning the rank 5 and~$q_2$ has the rank~4,
$q_3$~will have the rank~4 in all possible (one-step) successors of~$M$.
The only possible way how the rank of $q_3$ can be decreased is
a~$\gamma_4$ transition, with an additional assumption that~$q_3$ is in the
$O$~component of the corresponding macrostate.

Consider the first macrostate $M'$ (in an arbitrary run from~$M$) where~$q_3$
goes to the $O$-set.
Recall that~$q_3$ has the rank~4 and this value is not changed because~$q_3$
did not belong to the $O$-set.
Since~$q_2$ has also the rank~4, the $O$-set of~$M'$ contains at least
the states~$q_2$ and~$q_3$.
Moreover, for all successors of~$M'$ it holds that the $O$-set contains (at
least) the states~$q_3$ and~$q_4$.
Since~$q_4$ is a~final state, we are not allowed to use the
transition~$\gamma_4$ in any of the successors (the ``ranking function'' after
triggering the~$\gamma_4$ transition is not a~ranking function any more, since
an accepting state would be assigned an odd rank).
We do not flush the $O$-set and therefore, we do not accept the word
$a^\omega$ in a max-rank run.

%%%%%%%%%%%%%%%%%%%%%%%%%%%%%%%%%%%%%%%%%%%%%%%%%%%%%%%%%%%%%%%%%%%%%%%%%%%%%%%%
\vspace{-0.0mm}
\section{Experimental Results for BAs from LTL formulae}\label{sec:expr_ltl}
\vspace{-0.0mm}
%%%%%%%%%%%%%%%%%%%%%%%%%%%%%%%%%%%%%%%%%%%%%%%%%%%%%%%%%%%%%%%%%%%%%%%%%%%%%%%%

\newcommand{\figrankltl}[0]{
\begin{figure}[t]
\begin{subfigure}[b]{0.49\linewidth}
\includegraphics[width=6cm,keepaspectratio]{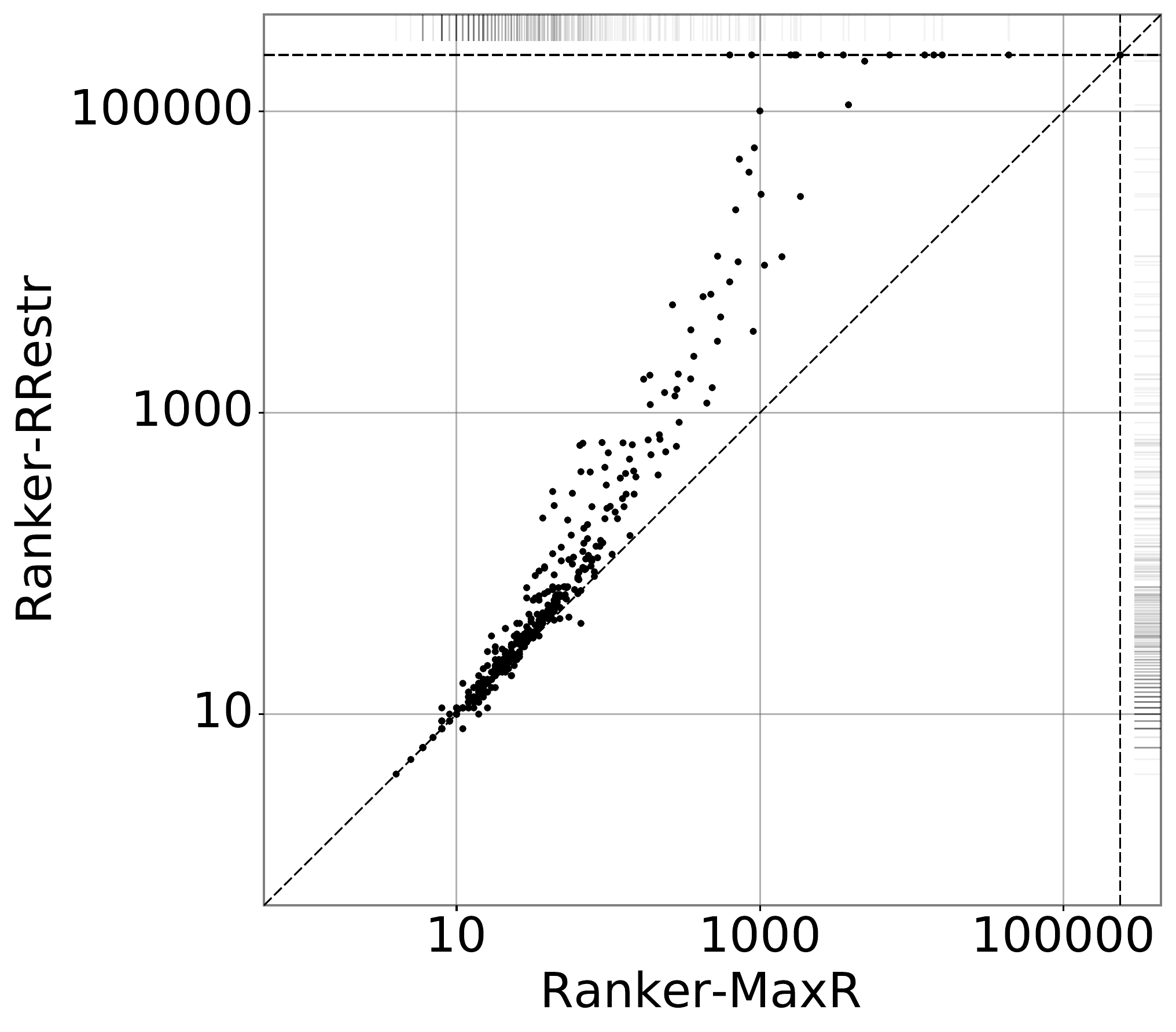}
\caption{\centering \footnotesize $\rankermaxrank$ vs $\rankerrankrestr$}
\label{fig:maxrank_vs_rankrestr_ltl}
\end{subfigure}
\begin{subfigure}[b]{0.49\linewidth}
\includegraphics[width=6cm,keepaspectratio]{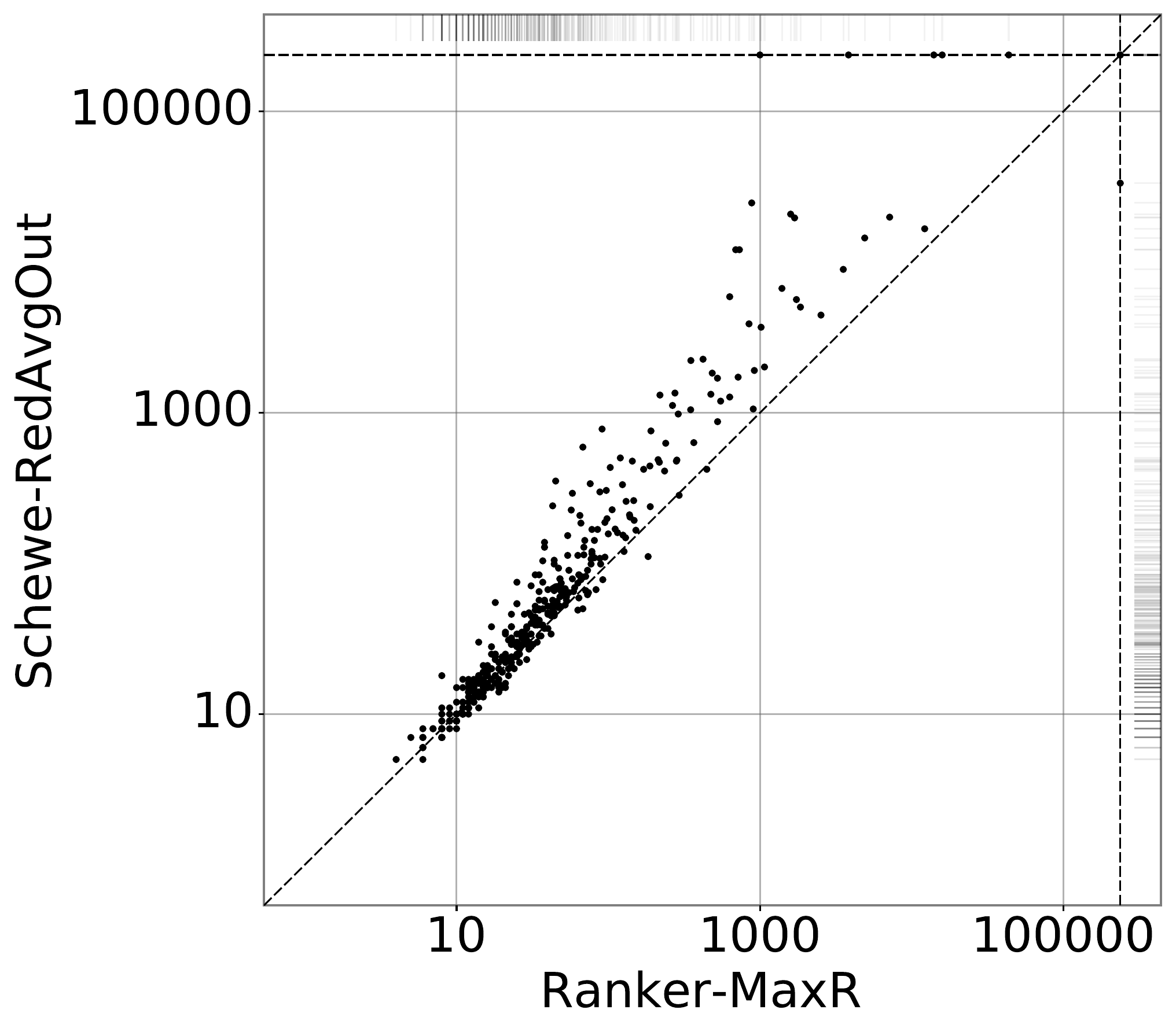}
\caption{\centering \footnotesize $\rankermaxrank$ vs $\algschewerao$}
\label{fig:ranker_vs_schewe_ltl}
\end{subfigure}
\caption{\footnotesize Evaluation of the effectiveness of our optimizations on the generated
  state space for the LTL benchmark.
  Both axes are logarithmic.
  A~point on the horizontal or vertical dashed lines represents a~timeout.}
\label{fig:rank-based-ltl}
\vspace*{-5mm}
\end{figure}
}

\newcommand{\figotherltl}[0]{
\begin{figure}[t]
\begin{subfigure}[b]{0.49\linewidth}
  \centering
\includegraphics[width=6cm,keepaspectratio]{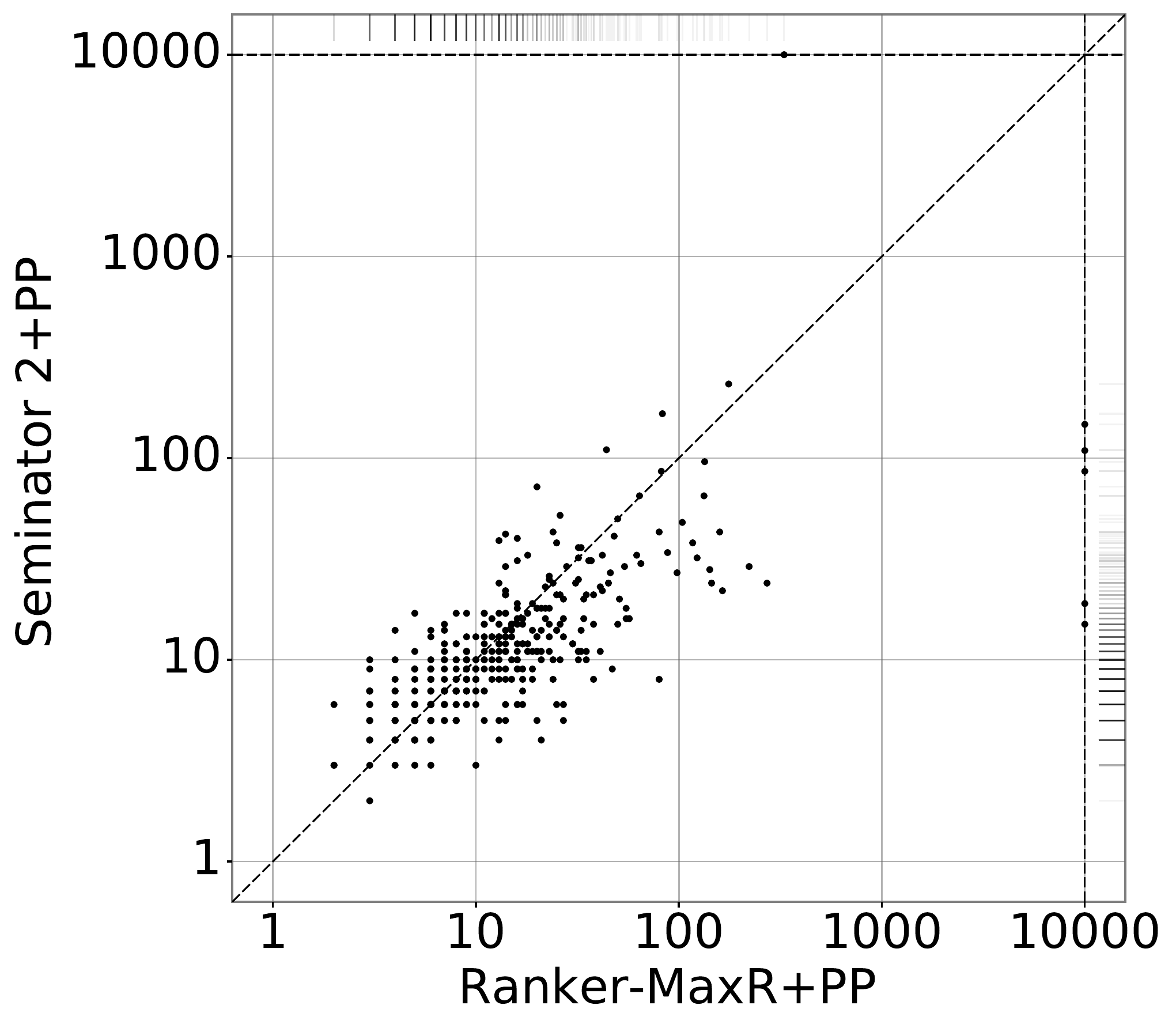}
\vspace*{0mm}
\caption{\centering \footnotesize $\rankermaxrank$ vs \seminator}
\label{fig:ranker_vs_seminator_ltl}
\vspace*{2mm}
\end{subfigure}
\begin{subfigure}[b]{0.49\linewidth}
  \centering
  \includegraphics[width=6cm,keepaspectratio]{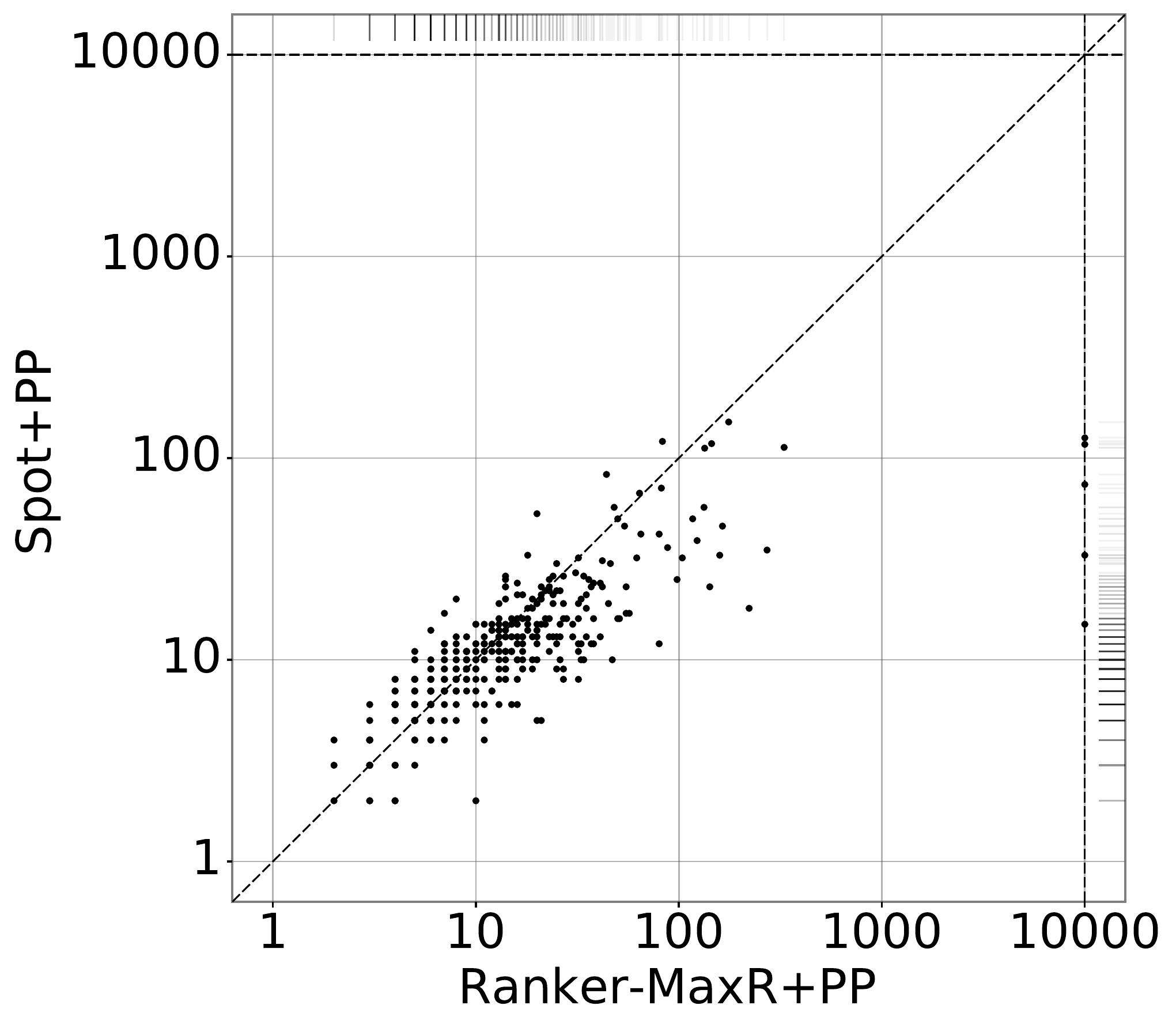}
  \vspace*{0mm}
  \caption{\centering \footnotesize $\rankermaxrank$ vs \spot}
  \label{fig:ranker_vs_spot_ltl}
\vspace*{2mm}
\end{subfigure}
\begin{subfigure}[b]{0.49\linewidth}
  \centering
  \includegraphics[width=6cm,keepaspectratio]{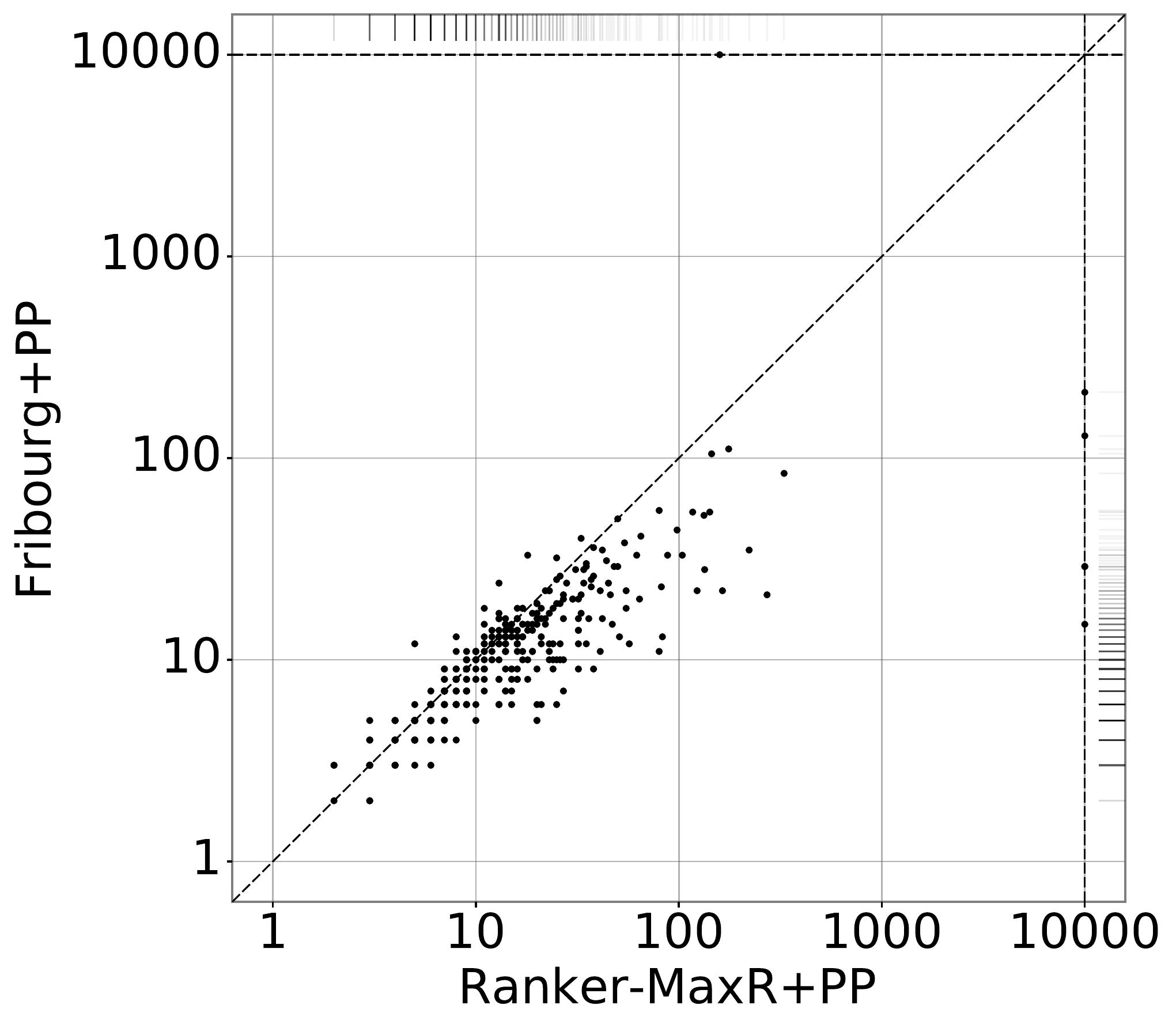}
  \vspace*{0mm}
  \caption{\centering \footnotesize $\rankermaxrank$ vs \fribourg}
  \label{fig:ranker_vs_fribourg_ltl}
\end{subfigure}
\hfill
\begin{subfigure}[b]{0.49\linewidth}
  \centering
  \includegraphics[width=6cm,keepaspectratio]{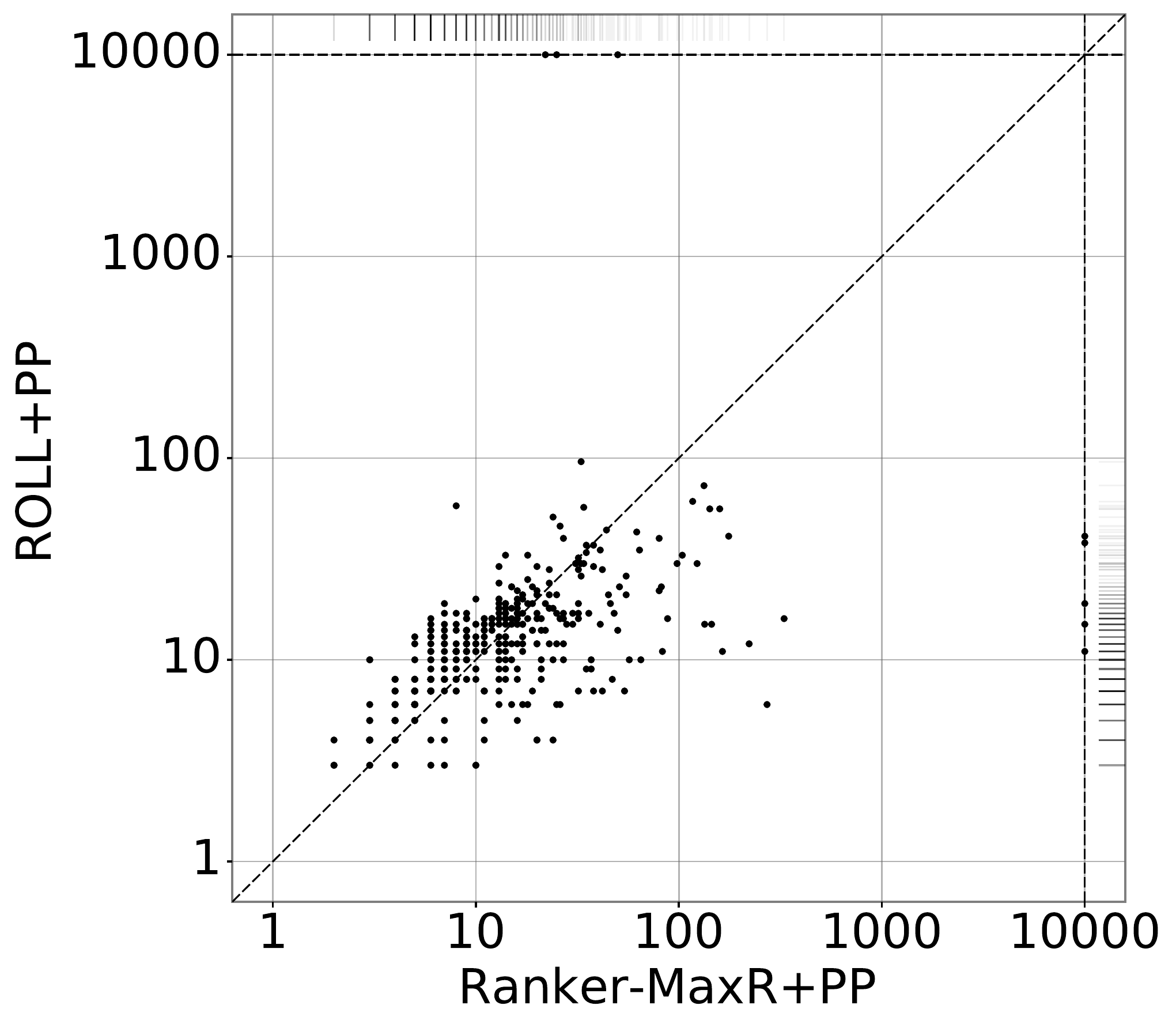}
  \vspace*{0mm}
  \caption{\centering \footnotesize $\rankermaxrank$ vs \roll}
  \label{fig:ranker_vs_roll_ltl}
\end{subfigure}
\caption{\footnotesize Comparison of the size of the BAs constructed using our optimized
  rank-based construction and other approaches on the LTL benchmark.
  Timeouts (5\,mins) are on the edges.
  }
\label{fig:comparison_other_ltl}
\vspace*{-5mm}
\end{figure}
}

\newcommand{\ltltablestats}[0]{
\begin{table}[b!]
\caption{Statistics for our experiments on the LTL benchmark (see the
  explanation of the columns in the description of \cref{tab:summary}).}
  \vspace*{-2mm}
  \begin{center}
  \scalebox{0.8}{
    \newcolumntype{d}[1]{D{.}{.}{#1}}
\begin{tabular}{lrd{4.2}d{2.1}d{5.2}r@{\hskip 1mm}|rrrr}
\toprule
 \multicolumn{1}{c}{\bf method}                     &
  \multicolumn{1}{c}{\bf ~~max~~}&      \multicolumn{1}{c}{\bf ~~~mean~~~~} &
  \multicolumn{1}{c}{\bf ~med.~} &
  \multicolumn{1}{c}{\bf ~~std.\ dev~~} &   \multicolumn{1}{c|}{\bf TO} &
  \multicolumn{1}{c}{\bf wins} &
  \multicolumn{1}{c}{\bf (TO)} &
  \multicolumn{1}{c}{\bf losses} &
  \multicolumn{1}{c}{\bf (TO)} \\
\midrule
  \emphcell  $\rankermaxrank$ &  43\,527 &  357.71  &     29   &  2\,510.31   &          5 & \centercell{---} & \centercell{---} & \centercell{---} & \centercell{---}\\
  $\rankerrankrestr$          & 214\,946 & 1\,948.26   &     33   & 13\,928      &         17 &    281 &              (12) &      35 &                (0)\\
  $\algschewerao$~\goalpic             &  33\,345 &  665.27  &     35   &  3\,081.89   &          9   &    292 &               (5) &      83 &                (1) \\
\midrule
\emphcell $\rankermaxrank$ &     330 &   20.07 &     10   &    32.86 &          5 & \centercell{---} & \centercell{---} & \centercell{---} & \centercell{---}\\
 $\rankermaxrank$+BO       &     440 &   21.09 &     11   &    38.28 &          2 & \centercell{---} & \centercell{---} & \centercell{---} & \centercell{---}\\
 $\piterman$~\goalpic               &     436 &   21.98 &     14.5 &    30.66 &          2 &    287 &              (1) &      76 &                (4) \\
 $\safra$~\goalpic                  &     361 &   30.39 &     17   &    44.24 &         14 &    330 &             (10) &      56 &                (1) \\
 $\spot$                   &     151 &   14.52 &     10   &    17.71 &          0 &    138 &              (0) &     195 &                (5) \\
 $\fribourg$~\goalpic               &     212 &   12.85 &      9   &    16.22 &          1 &     56 &              (1) &     238 &                (5) \\
 $\ltldstar$               &     223 &   21.05 &     13   &    24.17 &          6 &    249 &              (5) &     111 &                (4) \\
 $\seminator$              &     233 &   14.62 &     10   &    19.97 &          1 &    121 &              (1) &     218 &                (5) \\
 $\roll$                   &      96 &   13.81 &     11   &    10.62 &          3 &    243 &              (3) &     150 &                (5) \\
\bottomrule
\end{tabular}

  }
  \end{center}
\label{tab:ltl_stats}
\end{table}
}

\newcommand{\ltltabletimes}[0]{
\begin{table}[b!]
\caption{\footnotesize Run times of the tools on the LTL benchmark (in seconds)}
\vspace*{-2mm}
\centering
\label{tab:times_ltl}
\footnotesize
\begin{tabular}{lrrr}
\toprule
 \centercell{\bf method}                 &
  \centercell{\bf mean} &
  \centercell{\bf ~med.~} &
  \centercell{\bf ~std.\ dev} \\
\midrule
 $\rankermaxrank$     & 1.99 &    0.04  &  16.51 \\
 $\rankermaxrank$+BO  & 1.27 &    0.05  &   8.62 \\
 \piterman~\goalpic            & 6.65 &    5.62  &   3.73 \\
 \safra~\goalpic               & 8.37 &    5.80  &  13.45 \\
 \spot                & 0.06 &    0.02  &   0.71 \\
 \fribourg~\goalpic            & 7.22 &    5.48  &  13.22 \\
 \ltldstar            & 0.11 &    0.02  &   0.89 \\
 \seminator           & 0.08 &    0.02  &   0.83 \\
 \roll                & 7.28 &    2.74  &  16.06 \\
\bottomrule
\end{tabular}

\end{table}
}

\begin{table}[b!]
\caption{Statistics for our experiments on the LTL benchmark (see the
  explanation of the columns in the description of \cref{tab:summary}).}
  \vspace*{-2mm}
  \begin{center}
  \scalebox{0.8}{
    \newcolumntype{d}[1]{D{.}{.}{#1}}
\begin{tabular}{lrd{4.2}d{2.1}d{5.2}r@{\hskip 1mm}|rrrr}
\toprule
 \multicolumn{1}{c}{\bf method}                     &
  \multicolumn{1}{c}{\bf ~~max~~}&      \multicolumn{1}{c}{\bf ~~~mean~~~~} &
  \multicolumn{1}{c}{\bf ~med.~} &
  \multicolumn{1}{c}{\bf ~~std.\ dev~~} &   \multicolumn{1}{c|}{\bf TO} &
  \multicolumn{1}{c}{\bf wins} &
  \multicolumn{1}{c}{\bf (TO)} &
  \multicolumn{1}{c}{\bf losses} &
  \multicolumn{1}{c}{\bf (TO)} \\
\midrule
  \emphcell  $\rankermaxrank$ &  43\,527 &  357.71  &     29   &  2\,510.31   &          5 & \centercell{---} & \centercell{---} & \centercell{---} & \centercell{---}\\
  $\rankerrankrestr$          & 214\,946 & 1\,948.26   &     33   & 13\,928      &         17 &    281 &              (12) &      35 &                (0)\\
  $\algschewerao$~\goalpic             &  33\,345 &  665.27  &     35   &  3\,081.89   &          9   &    292 &               (5) &      83 &                (1) \\
\midrule
\emphcell $\rankermaxrank$ &     330 &   20.07 &     10   &    32.86 &          5 & \centercell{---} & \centercell{---} & \centercell{---} & \centercell{---}\\
 $\rankermaxrank$+BO       &     440 &   21.09 &     11   &    38.28 &          2 & \centercell{---} & \centercell{---} & \centercell{---} & \centercell{---}\\
 $\piterman$~\goalpic               &     436 &   21.98 &     14.5 &    30.66 &          2 &    287 &              (1) &      76 &                (4) \\
 $\safra$~\goalpic                  &     361 &   30.39 &     17   &    44.24 &         14 &    330 &             (10) &      56 &                (1) \\
 $\spot$                   &     151 &   14.52 &     10   &    17.71 &          0 &    138 &              (0) &     195 &                (5) \\
 $\fribourg$~\goalpic               &     212 &   12.85 &      9   &    16.22 &          1 &     56 &              (1) &     238 &                (5) \\
 $\ltldstar$               &     223 &   21.05 &     13   &    24.17 &          6 &    249 &              (5) &     111 &                (4) \\
 $\seminator$              &     233 &   14.62 &     10   &    19.97 &          1 &    121 &              (1) &     218 &                (5) \\
 $\roll$                   &      96 &   13.81 &     11   &    10.62 &          3 &    243 &              (3) &     150 &                (5) \\
\bottomrule
\end{tabular}

  }
  \end{center}
\label{tab:ltl_stats}
\end{table}

\figrankltl

In the following, we give experimental evaluation of our optimizations on the
benchmark from~\cite{seminator}.
The benchmark contains automata obtained from LTL formulae
\begin{inparaenum}[(i)]
  \item  from literature (221) and
  \item  randomly generated (1500).
\end{inparaenum}
(The motivation behind complementing BAs
obtained from LTL formulae might be, e.g., testing correctness of LTL-to-BA
translation algorithms; otherwise, when handling LTL formulae, it is a well
known fact negating the formula and constructing a BA directly for the negation
only increases the size of the BA \emph{linearly} instead of
\emph{exponentially}).
From the benchmark, we selected only 414 \emph{hard} BAs (in the same way as
in~\cref{sec:experiments}).
Note that although we selected only 414 \emph{hard} instances, their structure is still
simpler than the structure of the BAs considered in~\cref{sec:experiments},
since LTL does not have the full power of $\omega$-regular languages (this
difference in the difficulty can be seen from the summary statistics in
\cref{tab:ltl_stats}).
The timeout was again set to 5\,minutes.

\begin{table}[b!]
\caption{\footnotesize Run times of the tools on the LTL benchmark (in seconds)}
\vspace*{-2mm}
\centering
\label{tab:times_ltl}
\footnotesize
\begin{tabular}{lrrr}
\toprule
 \centercell{\bf method}                 &
  \centercell{\bf mean} &
  \centercell{\bf ~med.~} &
  \centercell{\bf ~std.\ dev} \\
\midrule
 $\rankermaxrank$     & 1.99 &    0.04  &  16.51 \\
 $\rankermaxrank$+BO  & 1.27 &    0.05  &   8.62 \\
 \piterman~\goalpic            & 6.65 &    5.62  &   3.73 \\
 \safra~\goalpic               & 8.37 &    5.80  &  13.45 \\
 \spot                & 0.06 &    0.02  &   0.71 \\
 \fribourg~\goalpic            & 7.22 &    5.48  &  13.22 \\
 \ltldstar            & 0.11 &    0.02  &   0.89 \\
 \seminator           & 0.08 &    0.02  &   0.83 \\
 \roll                & 7.28 &    2.74  &  16.06 \\
\bottomrule
\end{tabular}

\end{table}

In \cref{tab:ltl_stats}, we can see that the average sizes of the complements produced by
$\rankermaxrank$, compared to the other tools, are larger than in our main benchmark
in~\cref{sec:experiments}, the median is, however, comparable.
We believe that the larger average size is due to the two following facts:
\begin{inparaenum}[(i)]
  \item  that BAs from LTL formulae have a~simpler structure than general BAs
    that is more suitable for the other approaches and
  \item  our implementation does not take advantage of the symbolic alphabets
    present in the benchmark (we immediately translate the alphabet to an explicit
    alphabet, neglecting any relations among the symbols).
\end{inparaenum}
Moreover, note that $\rankermaxrank$ is no longer a clear winner here.
In particular, it is now comparable to \spot and \seminator (both provide smaller
automata slightly more often); \fribourg is the clear winner on this benchmark.
It is also interesting to see that the results for \piterman are significantly
worse when compared to the other tools than in our main experiments
in~\cref{sec:experiments}.
One possible explanation might be that the benchmarks from LTL formulae contain
symbolic alphabets; as far as we know, the implementation of \piterman in \goal
turns such an alphabet into an explicit one, so it cannot exploit the internal
structure of symbols on transitions).
Scatter plots comparing rank-based approaches are in \cref{fig:rank-based-ltl}.
Furthermore, scatter plots comparing $\rankermaxrank$ with other approaches are
in \cref{fig:comparison_other_ltl}.
Compared with the scatter plots in \cref{fig:comparison_other}, we substituted
\piterman with \spot, which had better results than both \piterman and \safra.

In \cref{tab:times_ltl}, we provide the times needed by the tools.
We note that \seminator performs much better on this benchmark, the performance
of \piterman goes down, and also that \roll does not perform as bad as in our
main benchmark (we believe that this is due to the lower difficulty of this
benchmark).

\figotherltl

\end{document}